\documentclass[%
reprint,pra,
twocolumn,superscriptaddress,amsmath,amssymb,
aps
]{revtex4-2}

\usepackage{graphicx}
\usepackage{xcolor}
\usepackage{amsmath}
\usepackage{amsthm,bbm}
\usepackage{bm,dsfont}% bold math
\usepackage{braket}% bold math
\usepackage{array}
\usepackage{amssymb}
\usepackage{enumerate}
\usepackage{comment}
\usepackage[colorlinks=true, citecolor=blue, urlcolor=blue]{hyperref}
\usepackage{floatrow}
\newfloatcommand{capbtabbox}{table}[][\FBwidth]

\usepackage{soul}

\makeatletter
\newtheorem*{rep@theorem}{\rep@title}
\newcommand{\newreptheorem}[2]{%
\newenvironment{rep#1}[1]{%
\def\rep@title{#2 \ref{##1}}%
\begin{rep@theorem}}%
{\end{rep@theorem}}}
\makeatother
\usepackage{mathrsfs}
\newtheorem{theorem}{Theorem}
\newreptheorem{theorem}{Theorem}

\newtheorem{definition}{Definition}

\newtheorem{assumption}{Assumption}

\newtheorem{lemma}{Lemma}
\newreptheorem{lemma}{Lemma}

\newcommand{\be}{\begin{equation}} 
\newcommand{\ee}{\end{equation}}
\newcommand{\beq}{\begin{eqnarray}}
\newcommand{\eeq}{\end{eqnarray}}
\newcommand{\rank}{\mathrm{rank}}
\newcommand{\id}{\mathds{1}}
\newcommand{\tr}{{\rm tr}}
\newcommand{\ra}{\rangle}
\newcommand{\la}{\langle}
\newcommand{\B}{\mathcal{B}}
\newcolumntype{C}{>{$}c<{$}}

\begin{document}

%%%%%%%%%%%%%%%%%%%%%%%%%%%%%%%%%%%%%%%%%%%%%%%%%%%%%%%%%%%%%%%%%%%

\title{Certifying Sets of Quantum Observables with Any Full-Rank State}

%%%%%%%%%%%%%%%%%%%%%%%%%%%%%%%%%%%%%%%%%%%%%%%%%%%%%%%%%%%%%%%%%%%

\author{Zhen-Peng Xu}
\email{zhen-peng.xu@ahu.edu.cn}
\affiliation{School of Physics and Optoelectronics Engineering, Anhui University, 230601 Hefei, People's Republic of China}
\affiliation{Naturwissenschaftlich-Technische Fakult\"at, Universit\"at Siegen, Walter-Flex-Stra{\ss}e 3, 57068 Siegen, Germany}

\author{Debashis Saha}
\email{saha@iisertvm.ac.in}
\affiliation{School of Physics, Indian Institute of Science Education and Research Thiruvananthapuram, Kerala 695551, India}

\author{Kishor Bharti}
\email{kishor.bharti1@gmail.com}
\affiliation{Institute of High Performance Computing (IHPC), Agency for Science, Technology and Research (A*STAR), 1 Fusionopolis Way, $\#$16-16 Connexis, Singapore 138632, Republic of Singapore}

\author{Ad\'an~Cabello}
\email{adan@us.es}
\affiliation{Departamento de F\'{\i}sica Aplicada II, Universidad de Sevilla, E-41012 Sevilla,
Spain}
\affiliation{Instituto Carlos~I de F\'{\i}sica Te\'orica y Computacional, Universidad de
Sevilla, E-41012 Sevilla, Spain}

%%%%%%%%%%%%%%%%%%%%%%%%%%%%%%%%%%%%%%%%%%%%%%%%%%%%%%%%%%%%%%%%%%%

\begin{abstract}
We show that some sets of quantum observables are unique up to an isometry and have a contextuality witness that attains the same value for any initial state. We prove that these two properties make it possible to certify any of these sets by looking at the statistics of experiments with sequential measurements and using any initial state of full rank, including thermal and maximally mixed states. We prove that this ``certification with any full-rank state'' (CFR) is possible for any quantum system of finite dimension $d \ge 3$ and is robust and experimentally useful in dimensions $3$ and $4$. In addition, we prove that complete Kochen-Specker sets can be Bell self-tested if and only if they enable CFR. This establishes a fundamental connection between these two methods of certification, shows that both methods can be combined in the same experiment, and opens new possibilities for certifying quantum devices.
\end{abstract}

%%%%%%%%%%%%%%%%%%%%%%%%%%%%%%%%%%%%%%%%%%%%%%%%%%%%%%%%%%%%%%%%%%%

\maketitle

%%%%%%%%%%%%%%%%%%%%%%%%%%%%%%%%%%%%%%%%%%%%%%%%%%%%%%%%%%%%%%%%%%%

{\em Introduction.---}Nonlocality \cite{brunner2014bell} and contextuality \cite{budroni2021quantum} are two fundamental predictions of quantum theory. Quantum theory also predicts that, in certain cases, there is an essentially unique way to achieve some specific nonlocal \cite{Yao_self,vsupic2019self,kaniewski2020weak} or contextual~\cite{BRVWCK19,Saha2020sumofsquares,bharti2021graph} correlation. Consequently, the observation of this specific correlation allows us to infer which quantum state has been prepared and which quantum observables have been measured, without making assumptions about the functioning of the devices used in the
experiment \cite{Yao_self,vsupic2019self,BRVWCK19,kaniewski2020weak,Saha2020sumofsquares,bharti2021graph}. 

However, none of the existing ``device-independent'' (DI) certification methods work if the fidelity of the prepared state with respect to a {\em specific pure state} is below a certain threshold. It is this specific pure state that guarantees the uniqueness of the quantum realization in the noiseless (ideal) case. In particular, none of the methods works if the prepared state is maximally mixed. This leads to the question of whether it would be possible to certify quantum observables using correlations produced by measurements on {\em unspecified mixed states}, including the maximally mixed state.

That, in quantum theory, this question may have an affirmative answer is suggested by the observation that, for any quantum system of finite dimension $d\ge 3$, there exist finite sets of observables that produce contextual correlations for any quantum state \cite{Cabello08,Badziag:2009PRL,Yu:2012PRL,Kleinmann:2012PRL}. These sets of observables are called state-independent contextuality (SI-C) sets \cite{SM,CKB2015,budroni2021quantum}. SI-C sets have fundamental applications in quantum information \cite{Cabello:2001PRLb,cleve2004,brassard2005quantum,Cubitt:2010PRL,horodecki2010contextuality,Cabello_QKD,aolita2012pra,Guhne:2013PRA,Canas:PRL2014,Cabello:2018PRLEPSILON,saha2019State,Ji:2021CACM,Gupta2023PRL,Zhen2023PRL} and have been experimentally tested~\cite{kirchmair2009Stateindependent,zhang2013state,huang2003experimental,amselem2009state,ambrosio2013experimental,leupold2018sustained}. 

But the existence of SI-C sets itself leads to another question: are there SI-C sets that are unique up to an isometry? This question is particularly relevant for understanding the mathematical structure of the set of quantum observables. Interestingly, if the answer to this question is positive, then there may be a connection to the question of whether there are quantum observables that can be certified with arbitrary mixed states.

In this Letter, we first show that there are SI-C sets that (i) are unique up to an isometry, and 
(ii) have a SI-C witness ${\cal W}$ that achieves the same value for every initial quantum state. 
These SI-C sets have therefore a characteristic signature that can be experimentally tested: the relations of compatibility between the observables (which are encoded in the expression of ${\cal W}$) and the state-independent value of ${\cal W}$. 

Next, we will show that SI-C sets with properties (i) and (ii) can be certified from the correlations of experiments with sequential measurements performed on {\em any full-rank mixed state}, including thermal and maximally mixed states. As soon as a mixed state of full rank gives the characteristic value for ${\cal W}$, any other state will do so. This leads to a method for certifying quantum observables from correlations that is fundamentally different than
self-testing based on Bell inequalities \cite{Yao_self,vsupic2019self,kaniewski2020weak}, state-dependent contextuality~\cite{BRVWCK19,Saha2020sumofsquares,bharti2021graph}, prepare-and-measure \cite{tavakoli2020self,farkas2019self}, and steering \cite{vsupic2016self,gheorghiu2015robustness,shrotriya2020self}. There are two fundamental differences: (a) The initial state required for the certification is not determined by the set of observables to be certified; any state of full rank can be used. (b) The certification guarantees the state-independent uniqueness (up to an isometry) of the set of observables.

In addition, we show that this new method, named ``certification with any full-rank state'' (CFR), is possible in every finite dimension $d \ge 3$, and provide a way to obtain sets of observables that enable CFR in any $d \ge 3$. We also prove that CFR is robust against experimental imperfections using examples in $d=3$ and $4$, and show how to test the robustness in any other case. 

Finally, we show that, for a fundamental class of SI-C sets, CFR is a necessary condition for Bell self-testing. This points out a connection between two different forms of certification and shows that these two forms can be applied simultaneously in Bell experiments with sequential measurements. This opens up some interesting possibilities which are discussed.

%%%%%%%%%%%%%%%%%%%%%%%%%%%%%%%%%%%%%%%%%%%%%%%%%%%%%%%%%%%%%%%%%%%

{\em Certification with any full-rank state.---}Unless otherwise indicated, hereafter we will focus on SI-C sets of projectors (rather than general self-adjoint operators) and on a special type of contextuality witness that can be defined from them using the following result, which is a generalization of a result in \cite{cabello2016Simple}, whose proof is in \cite{SM}.

{\em Lemma 1.---}Given a finite set of observables $\{\Pi_i\}$, with possible results $0$ or $1$, and graph of compatibility $G$ (in which each $\Pi_i$ is represented by a vertex $i \in V$ and there is an edge $(i,j) \in E$ if $\Pi_i$ and $\Pi_j$ are compatible), the following inequality holds for any noncontextual hidden-variable (NCHV) theory:
\begin{equation}\label{eq:gnci}
{\cal W} := \sum_{i \in V} w_i\ P_i - \sum_{(i,j) \in E} w_{ij} P_{ij} \overset{\rm NCHV}{\le} \alpha(G,\vec{w}),
\end{equation}
where $\vec{w}=\{w_i\}_{i \in V}$ is a set of positive weights for the vertices of $G$,
$w_{ij} \ge \max{(w_i,w_j)}$, $P_i = P(\Pi_i=1)$ is the probability of obtaining outcome $1$ when measuring observable $\Pi_i$, $P_{ij}=P(\Pi_i=1,\Pi_j=1)$ is the probability of obtaining outcomes $1$ and $1$ when measuring $\Pi_i$ and $\Pi_j$, and $\alpha(G,\vec{w})$ is the weighted independence number of $G$ with vertex weight vector $\vec{w}$ (see \cite{SM} for the definition).

Our first result is the following.

{\em Result 1.---}For any quantum system of any finite dimension $d \ge 3$, there is a finite set of observables $S=\{\Pi_i\}_{i=1}^{n}$ and a functional ${\cal W}$ such that, for any quantum state $\rho$,
${\cal W}(S,\rho) = Q$, 
and, if
$ {\cal W}(S',\rho') = Q$
for a set of observables $S'=\{\Pi'_i\}_{i=1}^{n}$ and
a state $\rho'$ of full rank in dimension $D$, then $S'$ and $S$ are equivalent in the sense that there is a unitary transformation $U$ that, for all $i$,
\begin{equation}
\label{imp0}
\Pi_i \otimes \id^{d_1} \oplus \Pi_i^* \otimes \id^{d_2} = U \Pi'_i U^\dagger, 
\end{equation}
where $\id^{d_1}$ is the identity in dimension $d_1$, with $d_1+d_2 = D/d$, $\Pi_i^*$ is the conjugate of $\Pi_i$, $\otimes$ denotes tensor product, $\oplus$ denotes direct sum, and $U^{\dagger}$ is the conjugate transpose of $U$. 
Moreover, ${\cal W}$ is a SI-C witness since $Q > C$ and
\begin{equation}
{\cal W} \le C
\end{equation}
is a state-independent noncontextuality inequality.

For the witnesses ${\cal W}$ of the form \eqref{eq:gnci}, $C =\alpha(G,\vec{w})$.
If those $d$-dimensional $\Pi_i$ are real (rather than complex), then Eq.~\eqref{imp0} becomes
\begin{equation}
\label{imp}
\Pi_i \otimes \id^{(D/d)} = U \Pi'_i U^{\dagger}.
\end{equation}

The practical consequence of Result~1 is that if, in an ideal experiment with sequential measurements, a set of $n$ measurement devices (one for each observable), combined in sequences as dictated by the form of ${\cal W}$, yields ${\cal W} = Q$ for a state of full rank, then we can be sure that these devices implement $S$ [or an equivalent set in the sense of Eqs.~\eqref{imp0} or \eqref{imp}]. Then, we will say that $S$ enables CFR. The case of nonideal experiments will be discussed later.

%%%%%%%%%%%%%%%%%%%%%%%%%%%%%%%%%%%%%%%%%%%%%%%%%%%%%%%%%%%%%%%%%%%
% Table I (BBC-21)
%%%%%%%%%%%%%%%%%%%%%%%%%%%%%%%%%%%%%%%%%%%%%%%%%%%%%%%%%%%%%%%%%%%

\begin{widetext}
\begin{table*}[t!]
\begin{tabular}{CCCCCCCCCCCCCCCCCCCCCC}
\hline\hline 
& v_1 & v_2 & v_3 & v_4 & v_5 & v_6 & v_7 & v_8 & v_9 & v_{10} & v_{11} & v_{12} & v_{13} & v_{14} & v_{15} & v_{16} & v_{17} & v_{18} & v_{19} & v_{20} & v_{21} \\
\hline \rule{0pt}{3.2mm}
v_{i1}\;\; & 0 & 0 & 0 & 1 & 1 & 1 & 1 & 1 & 1 & 1 & 0 & 0 & 1 & 1 & 1 & 1 & 1 & 1 & 1 & 1 & 1 \\
v_{i2}\;\; & 1 & 1 & 1 & 0 & 0 & 0 & \bar{1} & \bar{q} & \bar{g} & 0 & 1 & 0 & 1 & q & g & 1 & q & g & 1 & q & g \\
v_{i3}\;\; & \bar{1} & \bar{q} & \bar{g} & \bar{1} & \bar{q} & \bar{g} & 0 & 0 & 0 & 0 & 0 & 1 & 1 & g & q & q & 1 & g & g & q & 1 \\ 
\hline
w_i\;\; & 4 & 4 & 4 & 4 & 4 & 4 & 4 & 4 & 4 & 7 & 7 & 7 & 7 & 7 & 7 & 7 & 7 & 7 & 7 & 7 & 7 \\ \hline\hline
\end{tabular}
\caption{\label{tab:bbc21}{\bf BBC-21.} Each column $v_i$ corresponds to one observable represented by the projector $|v_i\rangle \langle v_i|$. The column $v_{ij}$ gives the components of $|v_i\rangle$ (unnormalized). $\bar{x}=-x$, $q = e^{2\pi \mathbbm{i}/3}$, and $g=q^2$. Compatible observables correspond to orthogonal vectors. The last row contains optimal weights $w_i$ for a SI-C witness ${\cal W}$ of the form~\eqref{eq:gnci}. The weights $w_{ij}$ in~\eqref{eq:gnci} can be chosen in any way that satisfies $w_{ij} \geq \max\{w_i, w_j\}$.} 
\end{table*}
\end{widetext}

%%%%%%%%%%%%%%%%%%%%%%%%%%%%%%%%%%%%%%%%%%%%%%%%%%%%%%%%%%%%%%%%%%%

{\em Proof.---}The proof is based on identifying sets enabling CFR in any dimension $d \ge 3$. We will name the SI-C sets using the initials of the authors and the number of projectors in the set. For example, BBC-21~\cite{Bengtsson:2012PLA}, CEG-18~\cite{Cabello:1996PLA}, and YO-13~\cite{Yu:2012PRL}. In other cases, we use the full name rather than the initial, as in Peres-24~\cite{Peres:1991JPA}. In other cases, we use the standard name, as in the Peres-Mermin square~\cite{Peres1990PLA,Mermin1990PRL}. While the details of the proof are specific for each SI-C set, a common step in all proofs is showing that the violation of a full-rank state $\rho'$ implies the same violation for any state.

The proof starts by showing that, in $d=3$, the set of $21$ rank-one projectors in Table~\ref{tab:bbc21} enables CFR. This set, hereafter called BBC-21, was introduced in \cite{Bengtsson:2012PLA} and is the smallest SI-C set of rank-one projectors requiring complex numbers known. The proof that BBC-21 is unique up to unitary transformations, which guarantees that condition (i) for CFR holds, is in \cite{SM}. Using the weights in the last row of Table~\ref{tab:bbc21}, the noncontextual bound of the witness ${\cal W}$ defined in Eq.~\eqref{eq:gnci} is $\alpha(G,\vec{w})=36$, while, for any initial quantum state, the value of ${\cal W}$ is $\vartheta(G,\vec{w})=40$. This proves that BBC-21 also satisfies condition (ii) for CFR.

In $d=4$, we show that three related fundamental SI-C sets enable CFR: (I) CEG-18 \cite{Cabello:1996PLA}, which is the smallest KS set \cite{SM} of rank-one projectors in any dimension (as proven in \cite{Xu:2020PRL}), (II) Peres-24 \cite{Peres:1991JPA}, which is the smallest complete KS set (see Definition~4) of rank-one projectors known, 
and (III) the Peres-Mermin square \cite{Peres1990PLA,Mermin1990PRL}, which is the smallest SI-C set of arbitrary self-adjoint operators (rather than projectors) known. The proofs that these sets are unique up to unitary transformations and the corresponding optimal state-independent contextuality witnesses yielding the same value for any state are in \cite{SM}.

Finally, for any finite dimension $d\ge5$, we prove (see \cite{SM}) that each of the members of a family of SI-C sets of rank-one projectors generated from Peres-24 using a method introduced in \cite{Cabello05} is unique up to unitary transformations and has a SI-C witness producing the same value for any initial state. \hfill\qedsymbol

While existing correlation-based certification methods require preparing a state with a high overlap with a target pure state, the SI-C sets that enable CFR can be certified using any unspecified full-rank state, something that is easier to prepare. A simple strategy is to let an arbitrary state go through randomly chosen measurements \cite{leupold2018sustained}, resulting in a maximally mixed state. Another strategy is to let the system interact with the environment, resulting in a thermal state. Both types of states are of full rank.

Not all SI-C sets enable CFR. For example, Peres-33 \cite{Peres:1991JPA}, which is the KS set of rank-one projectors in $d=3$ with the smallest number of bases known, is not unique up to unitary transformations. 
Interestingly, YO-13~\cite{Yu:2012PRL}, which is the SI-C set with smallest number of rank-one projectors in any dimension (as proven in \cite{Cabello:2016JPA}) and is a subset of Peres-33, enables CFR if two additional conditions are satisfied: (I') The relations of orthogonality between the elements $S'$ are the same as the relations of orthogonality between the elements $S$, and (II') for $\rho'$, the probabilities are normalized for every set of mutually orthogonal projectors summing up to the identity. This is shown in \cite{SM}. Both (I') and (II') can be experimentally tested (as in \cite{leupold2018sustained}).

%%%%%%%%%%%%%%%%%%%%%%%%%%%%%%%%%%%%%%%%%%%%%%%%%%%%%%%%%%%%%%%%%%%

{\em Robustness.---}The possibility of CFR of SI-C sets is a prediction of quantum theory. Now the question is whether this prediction can be tested in actual experiments or it requires idealizations that cannot be achieved in realistic experiments such as the requirement of perfectly sharp and compatible measurements for all pairs of compatible observables in the SI-C set. In other words, the question is whether CFR is robust against experimental imperfections.

Answering this question requires an additional analysis based on semidefinite programming whose size is related to the size of the SI-C sets. Here, we have performed this analysis for three of the SI-C sets, in dimensions 3 and 4, that we have proven that enable CFR. In all cases, the analysis was performed on a laptop computer and the computational execution time was less than $1$ h. The analysis of the robustness of the CFR of the other SI-C sets can be carried out using higher computational power.

Our result here is that the CFRs based on BBC-21, CEG-18, and Peres-24 are robust. We will also show that the CFR of YO-13 is robust under an extra assumption. Our result requires introducing some definitions. 

%%%%%%%%%%%%%%%%%%%%%%%%%%%%%%%%%%%%%%%%%%%%%%%%%%%%%%%%%%%%%%%%%%%

{\em Definition 1.---} A set of projectors $\{\Pi_i\} $ is said to be a $(\theta,\epsilon)$ realization of a SI-C set with respect to a contextuality witness ${\cal W}$ of type \eqref{eq:gnci} if, for all states $\ket{\psi}$,
\begin{subequations}
\begin{align}
\sum_{i=1}^n w_i \braket{\psi|\Pi_i|\psi} & \geqslant \theta > \alpha(G,w), \label{eq:rob_realizationa}\\
\braket{\psi|\Pi_i \Pi_j \Pi_i|\psi} & \leqslant \epsilon,\label{eq:rob_realizationb}
\end{align}
\end{subequations}
whenever $i$ and $j$ are adjacent in $G$ (i.e., whenever the corresponding projectors are orthogonal).

The conditions in Eqs.~\eqref{eq:rob_realizationa} and~\eqref{eq:rob_realizationb} are related to the sum of probabilities $\sum_i w_i P_i$ and joint probability $P_{ij}$ in Eq.~\eqref{eq:gnci}. In the ideal case, $\theta=Q$ (defined in Result~1), and $\epsilon=0$, which implies that the quantum value of ${\cal W}$ is $Q$. As $\theta$ is close enough to $Q$ and $\epsilon$ is close enough to $0$, the projectors $\{\Pi_i\}$ have the same rank. See~\cite{SM} for details.

{\em Definition 2.---}A noncontextuality inequality of the form \eqref{eq:gnci} provides an $(\epsilon,r)$-robust CFR of a $(Q,0)$ realization $\{\Pi_i\} $ of a SI-C set, if, for any $(Q - \epsilon,\epsilon)$ realization $\{\Pi'_i\} $ of the SI-C set, there is an isometry $\Phi$ such that
\begin{equation}
|\Phi(\Pi_i) - \Pi'_i| \le \mathcal{O}(\epsilon^r).
\end{equation}

{\em Result 2.---}The contextuality witnesses ${\cal W}$ of the form~\eqref{eq:gnci} for BBC-21, CEG-18, Peres-24, and YO-13 used in Result~1 provide $(\epsilon,1/2)$ robustness when $\epsilon$ is smaller than $0.132$, $0.134$, $0.177$, and $0.208$, respectively. For YO-13, the proof requires the extra assumption that the probabilities of every three mutually orthogonal projectors sum~$1$.

For more details on the proof, see \cite{SM}.

Any witness ${\cal W}$ of the form \eqref{eq:gnci} can be expressed with the joint probabilities of the outcomes of two sequential measurements from $\{A_j\}$. From the observed values satisfying conditions related to the ideality and the orthogonality relations of the projectors, one can certify the projectors and the measurements $A_i$. Moreover, when the experimental value of ${\cal W}$ is close enough to the quantum value, the robustness of the CFR is also ensured. See~\cite{SM} for more details.

%%%%%%%%%%%%%%%%%%%%%%%%%%%%%%%%%%%%%%%%%%%%%%%%%%%%%%%%%%%%%%%%%%%

{\em Bell self-testing and CFR.---}Bell self-testing \cite{Yao_self} is the task of certifying quantum states and measurements using only the statistics of Bell experiments.
One advantage of Bell self-testing with respect to CFR is that the former does not require projective measurements. One disadvantage, however, is that Bell self-testing requires spacelike separation between the tests. Therefore, an interesting question is whether SI-C sets that allow for CFR can be Bell self-tested and, if so, what is the relation between Bell self-testing and CFR. To address these questions, the following definitions will be useful.

{\em Definition 3.---}(Generalized KS set) A generalized Kochen-Specker (KS) set is a set of projectors of arbitrary rank (not necessarily of rank-one as it is the case in a KS set \cite{Pavicic:2005JPA}) which does not admit an assignment of $0$ or $1$ satisfying that: (I) two orthogonal projectors cannot both be assigned $1$, (II) for every set of mutually orthogonal projectors summing up to the identity, one and only one of them must be assigned $1$.

{\em Definition 4.---}(Complete KS set) The complete KS set associated to a generalized KS set $S$ is the set obtained by adding to $S$ the projectors $\id - \Pi_i - \Pi_j$ for every pair of orthogonal projectors $(\Pi_i,\Pi_j)$ in $S$ that does not belong to a complete basis.

For example, Peres-24 is a complete KS set, but CEG-18 and Peres-33 are not (BBC-21 and YO-13 are not KS sets). 
A complete KS set enables CFR if it satisfies properties (i) and (ii). 

Now we need a way to produce Bell nonlocality using a complete KS set. For that aim, we will define the following nonlocal game.

{\em Definition 5.---}(Context-projector KS game \cite{cleve2004,brassard2005quantum,aolita2012pra})
In each round of the game, a referee gives to one of the players, Alice, one of the contexts (i.e., a set of commuting projectors summing up the identity) of a complete KS set $S$ and asks her to output one of the projectors of this context. In the same round, the referee gives to one spatially separated player, Bob, one of the projectors of the same context and asks him to output $1$ or $0$. Alice and Bob win the round either if Alice outputs the projector given to Bob and Bob outputs $1$, or if Alice outputs a projector different than the one given to Bob and Bob outputs $0$.

This is a game that cannot be won with probability $1$ with classical resources and no communication, but that can be won with probability $1$ if the parties share copies of a qudit-qudit maximally entangled state with $d \ge 3$ and measure a complete KS set in dimension $d$.

Now, we can address the question of whether the SI-C sets that allow for CFR can be Bell self-tested.

{\em Result 3.---}The projectors of a complete KS set can be Bell self-tested if and only if the KS set enables CFR.

The proof is in \cite{SM}. Here, we will focus on some implications of this result. One is that Bell self-testing and CFR can be accomplished simultaneously in an experiment that combines Bell and sequential tests \cite{Cabello:2010PRL,Liu:2016PRL,Gonzales-Ureta2023}. Consider two spatially separated parties, Alice and Bob~1, sharing copies of a qudit-qudit maximally entangled state and performing local measurements of the projectors of a complete KS set $S$. In addition, consider a third party, Bob~2, that receives the system that Bob~1 has measured (we assume that Bob~1's measurements are nondemolition measurements \cite{leupold2018sustained,WangSAdv2022}). Suppose that Bob~2 measures elements of $S$ on this system. Then: (a)~The Alice-Bob~1 statistics can Bell self-test $S$ in Alice's and Bob~1's sides. (b)~The Bob~1-Bob~2 statistics enable CFR of $S$ in Bob~1's and Bob~2's sides (and the Alice-Bob~1 Bell self-test can guarantee that Bob~1's input state is of full rank). (c)~The Alice-Bob~2 statistics conditioned to that Bob~2's measurement is compatible to Bob~1's can Bell self-test $S$ in Alice's and Bob~2's sides. This allows for the simultaneous certification of Bob~1's $S$ by two different methods and opens new possibilities.

%%%%%%%%%%%%%%%%%%%%%%%%%%%%%%%%%%%%%%%%%%%%%%%%%%%%%%%%%%%%%%%%%%%

{\em Conclusions and future directions.---}In this Letter, we have presented three results that push the field of certification of quantum processes based only on correlations beyond its established limits. Results 1 and 2 allow us to circumvent a conceptual limitation of existing methods, namely, the need of targeting specific pure states. We have proven that this is not necessary: for any quantum system of any finite dimension $d\ge3$, there are sets of quantum observables that can be certified using any full-rank quantum state. This ``certification with any full rank state'' offers interesting possibilities. For example, suppose that 
the same preparation is used to certify via CFR two sets of observables: one of them in dimension $d_1$ and the other in dimension $d_2$. This automatically certifies via CFR that the dimension of the system is lower bounded by the lowest common denominator of $d_1$ and $d_2$. This provides a method to certify quantum systems of high dimensions, something that is difficult in a DI way \cite{PhysRevLett.110.150501,PhysRevLett.112.140503}. Moreover, in principle, CFR becomes more useful as the dimension grows, since preparing a full-rank mixed state is easier than preparing a state with a high overlap with a pure target state.

Result 3 pushes the field in a different sense. It shows that, for a general class of sets of observables, CFR is possible if and only if Bell self-testing is possible. This indicates that there may be a general unified framework for certification based solely on correlations, so that all existing methods can be viewed as particular cases. The precise characterization of this framework constitutes an interesting challenge. On the other hand, Result~3 shows that there are sets of observables that can be simultaneously Bell self tested (using Alice-Bob~1 correlations) and certified via CFR (using Bob~1-Bob~2 correlations). This is interesting as it may lead to a robust method for self-testing L\"uders processes \cite{Luders:1951APL,Pokorny:PRL20} in any finite dimension $d \ge 3$ (which is where observables represented by rank-one projectors have one outcome whose quantum post-measurement state depends on the input state). In the framework of general probabilistic theories, L\"uders processes correspond to ``ideal (or sharp) measurements'' \cite{chiribella2014measurement}: processes that yield the same outcome when repeated and are minimally disturbing (only disturb incompatible observables). The existence of ideal measurements is ``one of the fundamental predictions of quantum mechanics'' \cite{Pokorny:PRL20}. The DI certification of ideal measurements in arbitrary (finite) dimension would require the DI certification of the corresponding quantum 
instruments (which capture both the classical outputs and the corresponding quantum post-measurement states~\cite{Davies:JMA70,Ozawa:JMP1984,Dressel:PRA2013}). Previous works have explored the DI \cite{Wagner2020deviceindependent} and semi-DI \cite{PhysRevResearch.2.033014} certification of instruments corresponding to nonideal qubit measurements. The DI certification of ideal measurements 
would operationally ``bridge the gap between general probabilistic theories and the DI framework'' \cite{Chiribella:2016IC}, blurring the boundaries between three different approaches for understanding quantum theory: DI, general probabilistic theories, and general Bayesian theories \cite{PhysRevResearch.2.042001}, where ideal measurements are central. Future research should go in these directions.
\\

%%%%%%%%%%%%%%%%%%%%%%%%%%%%%%%%%%%%%%%%%%%%%%%%%%%%%%%%%%%%%%%%%%%

This work was supported by the EU-funded project \href{10.3030/101070558}{FoQaCiA} and the \href{10.13039/501100011033}{MCINN/AEI} (Project No.\ PID2020-113738GB-I00). Z.-P. X.\ acknowledges support from the \href{10.13039/501100001809}{National Natural Science Foundation of China} (Grant No.\ 12305007), Anhui Provincial Natural Science Foundation (Grant No.\ 2308085QA29), and the \href{10.13039/100005156}{Alexander von Humboldt Foundation}.

%%%%%%%%%%%%%%%%%%%%%%%%%%%%%%%%%%%%%%%%%%%%%%%%%%%%%%%%%%%%%%%%%%%

\appendix

%%%%%%%%%%%%%%%%%%%%%%%%%%%%%%%%%%%%%%%%%%%%%%%%%%%%%%%%%%%%%%%%%%%

\section*{Supplemental material}

%%%%%%%%%%%%%%%%%%%%%%%%%%%%%%%%%%%%%%%%%%%%%%%%%%%%%%%%%%%%%%%%%%%

\section{Concepts}
\label{app:A}

%%%%%%%%%%%%%%%%%%%%%%%%%%%%%%%%%%%%%%%%%%%%%%%%%%%%%%%%%%%%%%%%%%%

\begin{definition}[Ideal measurement~\cite{Chiribella:2016IC}]
	An ideal measurement of an observable $A$
	is a measurement of $A$ that yields the same outcome when it is repeated on the same system
	and does not disturb any observable compatible with $A$. 
\end{definition}

\begin{definition}[Compatible observables] 
	Two observables $A$ and $B$ are compatible
	or jointly measurable if there exists an observable $C$ such that, for every initial
	state, for every outcome $a$ of $A$, the probability of obtaining outcome $a$ for $A$ is
	\begin{equation}
		P(A = a) = \sum_{o\in c_a} P (C = o)
	\end{equation}
	and, for every outcome $b$ of $B$,
	\begin{equation}
		P(B = b) = \sum_{o\in c_b} P (C = o),
	\end{equation}
	where the disjoint union of $\{c_a\}_a$ and the disjoint union of $\{c_b\}_b$ are both equal to the complete set of outcomes of~$C$.
	$C$ is called a {\em refinement} of $A$
	(and $B$). $A$ (and $B$) is called a {\em coarse-graining} of $C$. Therefore, two observables are compatible when they have a common refinement or are both coarse-grains of the same observable.
\end{definition}

\begin{definition}[Ideal observable]
	An ideal or sharp observable is one that can be measured with ideal measurements, and that all its possible coarse-grained versions can also be measured with ideal measurements.
\end{definition}

In quantum theory, ideal observables are represented by self-adjoint operators.

\begin{definition}[SI-C set]
	A state-independent contextuality (SI-C) set in dimension $d$ is a set of ideal observables that produces contextuality for any initial state in dimension $d$.
\end{definition} 

In particular, a set of $n$ ideal observables represented in quantum theory by $n$ projectors $\{\Pi_i\}_{i=1}^n$ is a SI-C set if there is a set of weights $\vec{w}=\{w_i\}_{i=1}^n$ for the vertices of the graph $G$ of compatibility of $\{\Pi_i\}_{i=1}^n$ (in which vertices represent observables and edges connect pairwise compatible observables) such that a noncontextuality inequality of the form \eqref{eq:gnci} is violated by any quantum state in dimension $d$. 

\begin{definition}[KS set]
	A Kochen-Specker (KS) set is a SI-C set of rank-one projectors which does not admit an assignment of $0$ or $1$ satisfying that: (I) two orthogonal projectors cannot both be assigned $1$, (II) for every set of mutually orthogonal projectors summing the identity, one of them must be assigned $1$.
\end{definition} 

There are SI-C sets of rank-one projectors that are not KS sets. Examples are YO-13 \cite{Yu:2012PRL} and BBC-21 \cite{Bengtsson:2012PLA}.

\begin{definition}[Egalitarian SI-C set] 
A SI-C set is egalitarian if it produces, for any state, the same violation of a given noncontextuality inequality.
\end{definition}

In particular, a SI-C set $\{\Pi_i\}$ is egalitarian if there is a set of weights $\vec{w}=\{w_i\}$ for the vertices of the graph $G$ of compatibility of $\{\Pi_i\}$ such that, for any quantum state, the left-hand side of \eqref{eq:gnci} yields $Q > \alpha(G,\vec{w})$.

\begin{definition}[Independence number]
The independence number $\alpha(G,\vec{w})$ of a vertex-weighted graph $(G,\vec{w})$ is the maximum 
$w(S)=\sum_{v \in S}w(v)$ taken over all independent sets $S$ of $G$.
A set of vertices of $G$ is independent if all the vertices in it are pairwise nonadjacent.
\end{definition}

\begin{definition}[Lov\'{a}sz number]
The Lov\'{a}sz number $\vartheta(G,\vec{w})$ of a vertex-weighted graph $(G,\vec{w})$ is the maximum of $\sum_i w_i |\langle v_i | \psi\rangle|^2$ over all unit vectors $| \psi\rangle$ and $| v_i\rangle$ such that $ \langle v_j | v_i\rangle=0$ whenever $i$ and $j$ are adjacent vertices of $G$.
\end{definition}

\begin{definition}[Egalitarian Lov\'{a}sz-optimum SI-C set] 
\label{def:elos}
An egalitarian SI-C set $\{\Pi_i\}$ is Lov\'{a}sz-optimum if, for any quantum state, the left-hand side of \eqref{eq:gnci} equals the Lov\'{a}sz number of the weighted graph $(G,\vec{w})$, denoted $\vartheta(G,\vec{w})$, where $\vec{w}$ is a set of weights for which the SI-C set is egalitarian. 
\end{definition}

%%%%%%%%%%%%%%%%%%%%%%%%%%%%%%%%%%%%%%%%%%%%%%%%%%%%%%%%%%%%%%%%%%%

\section{Tools used in the proof of Result~1}
\label{app:B}

%%%%%%%%%%%%%%%%%%%%%%%%%%%%%%%%%%%%%%%%%%%%%%%%%%%%%%%%%%%%%%%%%%%

\subsection{Proof of Lemma~1}

%%%%%%%%%%%%%%%%%%%%%%%%%%%%%%%%%%%%%%%%%%%%%%%%%%%%%%%%%%%%%%%%%%%

Let us denote by $a$ the maximum of the left-hand side of \eqref{eq:gnci} that is achievable by a noncontextual hidden-variable (NCHV) theory. Since set of correlations for NCHV theories forms a convex polytope, $a$ can be obtained with a deterministic assignment of outcomes to the observables. For a given deterministic assignment achieving $a$, if $\Pi_i = 1$ and $\Pi_j = 1$ for one $(i,j) \in E$, then let us consider the part in Eq.~\eqref{eq:gnci} which contains $\Pi_j=1$,
\begin{align}
&w_j P(\Pi_j = 1) - \sum_{k:(k,j)\in E} w_{kj} P(\Pi_k=1,\Pi_j=1) \nonumber \\
\le &w_j P(\Pi_j = 1) - w_{ij} P(\Pi_i=1,\Pi_j=1) \nonumber \\
= & w_j - w_{ij} \le 0.
\end{align}
This implies that, by setting $\Pi_j=1$ to be $0$ in the deterministic probability assignment, the value does not decrease. Hence, the maximal value $a$ can always be achieved by one deterministic probability assignment where $\Pi_i$ and $\Pi_j$ are not both assigned $1$ if $(i,j) \in E$. Therefore, $a$ can only be $\alpha(G,w)$.

In the case that $w_{ij} > \max \{w_i,w_j\} $, by setting $\Pi_j=1$ to be $0$ in the deterministic probability assignment, the value increases. Hence, the maximal value $a$ can never be achieved by the deterministic assignment where $\Pi_i=1$ and $\Pi_j=1$ for one $(i,j) \in E$.

%%%%%%%%%%%%%%%%%%%%%%%%%%%%%%%%%%%%%%%%%%%%%%%%%%%%%%%%%%%%%%%%%%%

\subsection{Common procedure in all the proofs of uniqueness up to unitary transformations}

%%%%%%%%%%%%%%%%%%%%%%%%%%%%%%%%%%%%%%%%%%%%%%%%%%%%%%%%%%%%%%%%%%%

Let us first consider the case in which the graph of compatibility $G$ is equal to the graph of orthogonality of the observables/projectors $\{\Pi_i\}_{i=1}^n$. That this is the case, can be experimentally tested as described in Lemma~2 in the main text. In some special cases, this can even be tested using the maximal quantum violation of the witness ${\cal W}$ of the form~\eqref{eq:gnci}.

\begin{theorem}\label{thm:bridge}
For any noncontextuality inequality of the form \eqref{eq:gnci}, if there exists a set of observables $\{\Pi_i\}$ such that, for any quantum state, the left-hand side of \eqref{eq:gnci} equals the maximum value attainable in quantum theory, denoted by $Q$, then 
\begin{equation}
	Q=\vartheta(G,\vec{w}),
\end{equation} 
where $\vartheta(G,\vec{w})$ is the Lov\'asz number of the graph $G$ with weights $\vec{w}$, and the observables $\{\Pi_i\}$ must be of the type represented by projectors $\{\Pi_i\}$ (here, we use the same symbol for the observable and the projector that represents it). Moreover, if $w_{ij} > \max \{w_i, w_j\}$, then, for any $(i,j) \in E$,
\begin{equation}
	\Pi_i\Pi_j = 0.
\end{equation}
\end{theorem}

%%%%%%%%%%%%%%%%%%%%%%%%%%%%%%%%%%%%%%%%%%%%%%%%%%% 

\begin{proof} 
Let us first consider the case in which $w_{ij} > \max \{w_i, w_j\}$.
Consider the state $\rho_i = \Pi_i /\tr(\Pi_i)$ associated to $\Pi_i$, with $i \in V$. For this state,
\begin{equation}
	P_{\rho_i}(\Pi_i=1) = 1,
\end{equation}
and, for any $j$ such that $(i,j) \in E$,
\begin{equation}
	P_{\rho_i}(\Pi_i=1,\Pi_j=1)= P_{\rho_i}(\Pi_j=1).
\end{equation}
If there is $\Pi_j$ such that $(i,j) \in E$ and $\Pi_i\Pi_j \neq 0$, then the terms in the left-hand side of \eqref{eq:gnci} that contain $\Pi_j$ satisfy
\begin{align}
	& w_j P_{\rho_i}(\Pi_j = 1) - \sum_{k : (k,j) \in E} w_{kj} P_{\rho_i}(\Pi_k=1,\Pi_j=1) \nonumber \\ 
	\le & w_j P_{\rho_i}(\Pi_j=1) - w_{ij} P_{\rho_i}(\Pi_i=1,\Pi_j=1) \nonumber \\ 
	= & (w_j-w_{ij}) P_{\rho_i}(\Pi_j=1)
	< 0.
\end{align}
Hence, by setting $\Pi_j = 0$, the quantum value of the left-hand side of \eqref{eq:gnci} for state $\rho_i$ increases, which contradicts the assumption that $Q$ is the maximum quantum value. Therefore, we can conclude that, for all $(i,j)\in E$,
\begin{equation}\label{eq:exclusivity}
	\Pi_i\Pi_j = 0.
\end{equation}
Under this condition, for a given state $\rho$, the left-hand side of~\eqref{eq:gnci} equals to $\sum_i w_i p(\Pi_i=1)_\rho$, which is upper bounded by $\vartheta(G,\vec{w})$. Therefore, $q\le \vartheta(G,\vec{w})$ (see Sec.~\ref{app:A}). On the other hand, by the definition of $\vartheta(G,\vec{w})$, there is always $\{\Pi_i\}_{i \in V}$ such that the quantum value is $\vartheta(G,\vec{w})$, which implies $q \ge \vartheta(G,\vec{w})$. Therefore, we can conclude that $q=\vartheta(G,\vec{w})$.

Similarly, $Q=\vartheta(G,\vec{w})$ also holds in the case that $w_{ij} = \max\{w_i, w_j\}$. However, in this case, $\Pi_i \Pi_j = 0$ does not need to hold.
\end{proof}

The requirement of a full-rank state in our main result's proof stems from the fact that if a full-rank state attains the maximal value $Q$, then any state achieves $Q$. Consequently, as per the above theorem, the projectors will satisfy the orthogonality relation according to $G$. To be precise, Eq.~\eqref{eq:exclusivity} implies that the compatibility graph $G$ represents also the orthogonality relations between the observables $\{\Pi_i\}$. This step is crucial in the proof of Result~1. 

%%%%%%%%%%%%%%%%%%%%%%%%%%%%%%%%%%%%%%%%%%%%%%%%%%%%%%%

In general, $Q$ may be difficult to determine. Then, one cannot decide whether or not $Q$ is achievable for all the states. However \cite{CabelloPRA2016}, for $w_{ij} > \max \{w_i, w_j\}$,
\begin{equation}
\label{qth}
Q \le \vartheta(G',\vec{w}') - \sum_{(i,j) \in E} w_{ij},
\end{equation}
where $(G',\vec{w}')$ is the (weighted) graph of exclusivity of the weighted events $\{w_i(\Pi_i=1)\}_{i \in V}$ and $\{w_{ij}(\Pi_i=0,\Pi_j=0), w_{ij}(\Pi_i=0,\Pi_j=1), w_{ij}(\Pi_i=1,\Pi_j=0)\}_{(i,j) \in E}$. Therefore, if for any quantum state,
the left-hand side of \eqref{eq:gnci} equals the right-hand side of \eqref{qth}, we can conclude that $Q = \vartheta(G,\vec{w})$ and that $\Pi_i \Pi_j=0$ for any $(i,j) \in E$. 
In fact, this is the case for the optimal vertex-weighted graphs of compatibility of BBC-21, CEG-18, and Peres-24. For the one of YO-13, we need $w_{ij} \ge 2\max\{w_i, w_j\}$ in the form~\eqref{eq:gnci} and the extra normalization conditions in Eq.~\eqref{eq:yoextra}. The computations needed for checking Peres-39 and beyond it cannot be carried out with a laptop computer.

The proof of uniqueness up to unitary transformations of a given SI-C set $\{\Pi_i\}_{i=1}^n$ is then based on two facts:
\begin{enumerate}
\item All the projectors in $\{\Pi_i\}$ have the same rank, that we call $\kappa$.
\item If $(i,j)\not\in E$ in the graph of orthogonality $G$, then $\rank(\Pi_j \Pi_i \Pi_j) = \rank(\Pi_i) = \rank(\Pi_j)$.
\end{enumerate}
The first fact is ensured by the second one if the complement graph of $G$ is connected, which is indeed true for all the cases considered here.
The second fact holds also for all the cases considered here. This can be verified by semi-definite programming (SDP).

To be more explicit, if there is $i$ and $j$ such that $(i,j)\not\in E$ and $\rank(\Pi_j \Pi_i \Pi_j) < \rank(\Pi_j)$, then there should exist a state $|s\rangle$ such that $\langle s|\Pi_j|s\rangle = 1$ and $\langle s|\Pi_i|s\rangle = 0$. 
Notice that, $\{\Pi_i\}$ is a SI-C set with some weights $w$ and quantum violation $Q$. 
Denote $T_{ij} = \langle s|\Pi_{i-1} \Pi_{j-1} |s\rangle$, where $\Pi_0=\id$. Then, the matrix $T$ satisfies the following conditions:
\begin{align}\label{eq:sdp_relaxation}
&T \succeq 0, \nonumber\\
&T_{11} = 1, T_{1k} = T_{kk}, \forall k, \nonumber\\
&T_{kl} = 0, \text{ if } (k,l)\in E, \nonumber\\
&T_{1j} = 1, T_{1i} = 0, \sum_{k\ge 2} w_{k-1} T_{1k} = Q.
\end{align}
For YO-13, we have to add some extra linear conditions [see Eq.~\eqref{eq:yoextra}].
To check whether or not condition 2 holds, we can check whether or not the SDP in Eq.~\eqref{eq:sdp_relaxation} is feasible, which is a relaxation of the original problem.
In all the cases considered here, the relaxation in Eq.~\eqref{eq:sdp_relaxation} is not feasible. This implies that $\rank(\Pi_j \Pi_i \Pi_j) < \rank(\Pi_j)$ cannot be true for $(i,j)\not\in E$. Therefore, the second fact is also ensured.

Then, we can choose a complete basis as the computational basis such that $\{\Pi_i\}_{i=1}^c$ are the projectors into the subspace occupying dimensions from $(c-1)\kappa+1$ to $c\kappa$. 
Notice that it is always possible to choose a complete basis for the SI-C sets considered here.
For any other projector $\Pi_i$, we have 
\begin{equation}\label{eq:dec1}
\Pi_i = {L'}_i^{\dagger} L'_i, \ L'_i = [B_{1i}, B_{2i}, \ldots, B_{ci}].
\end{equation}
The reason is as follows. Since $\Pi_i$ is a projector of rank $\kappa$, it can be written as $\Pi_i = \sum_{t=1}^{\kappa} |v_i\rangle\langle v_i|$, where the $|v_i\rangle$'s are orthogonal normalized vectors. Denote $L_i'=[|v_1\rangle,\ldots,|v_\kappa\rangle]^\dagger$, we have $\Pi_i = {L'}_i^\dagger {L'}_i$. Here, the dimension of $L_i'$ is $\kappa\times d$, where $d$ is the dimension and $d=c\kappa$. Hence, we can always write $L_i'$ as $[B_{1i}, B_{2i},\ldots, B_{ci}]$, where $B_{ti}$ is a $\kappa\times\kappa$ matrix for any $t$. Then, for $t=1,\ldots,c$, we have $\Pi_t \Pi_i \Pi_t = B_{ti}^\dagger B_{ti}$. Since $\Pi_t \Pi_i \Pi_t$ is of rank $\kappa$ when $(t,i)\not\in E$, we have $B_{ti}$ to be invertible in this case. In the case that $(t,i)\in E$, we have $\Pi_t\Pi_i\Pi_t = 0$. Consequently, in this case, $B_{ti}=0$.

By definition of $L_i'$, $L_i' {L'}_i^\dagger = \id_{\kappa}$.
If $B_{1i}$ is invertible, then we introduce 
$L_i := B_{1i}^{-1}L'_i = [\id, B_{1i}^{-1}B_{2i}, \ldots, B_{1i}^{-1}B_{ci}].$
Further, we can verify that
\begin{equation}\label{lem1fact3}
L_i^\dagger (L_i L_i^\dagger)^{-1} L_i = {L'}_i^\dagger ({L'}_i {L'}_i^\dagger)^{-1} {L'}_i = \Pi_i.
\end{equation}
In addition, 
\begin{equation}\label{lem1fact32}
L_iL_j^{\dagger} = 0 \Leftrightarrow L'_i L'_j{^{\dagger}} = 0 \Leftrightarrow \Pi_i \Pi_j = 0.
\end{equation}
Hence, in the proofs of uniqueness, we will adopt $\{L_i\}$ for convenience. From the uniqueness of $\{L_i\}$, we can recover the uniqueness of $\{\Pi_i\}$. For completeness, a similar result is proven when point (III) of Lemma~\ref{lemma:rank} is proven.

%%%%%%%%%%%%%%%%%%%%%%%%%%%%%%%%%%%%%%%%%%%%%%%%%%%%%%%%%%%%%%%%%%%

\subsection{Proof that BBC-21 is unique up to unitary transformations}

%%%%%%%%%%%%%%%%%%%%%%%%%%%%%%%%%%%%%%%%%%%%%%%%%%%%%%%%%%%%%%%%%%%

The uncharacterized projectors in this case are $\{\Pi_i\}$ that are defined by the $\{L_i\}$ operators in Eq.~\eqref{lem1fact3}.
For convenience, we relabel $\{L_i\}$ by following the same order of the vectors $\{v_i\}_{i=1}^{21}$ in Table~\ref{tab:bbc21} as follows:
\begin{equation} 
\begin{split}
	& y^0_1, y^1_1, y^2_1, y^0_2, y^1_2, y^2_2, y^0_3, y^1_3, y^2_3, z_1, z_2, z_3, \\
	& h^0_1, h^1_1, h^2_1, h^0_2, h^1_2, h^2_2, h^0_3, h^1_3, h^2_3.
\end{split}
\end{equation}
For example, $\Pi_{10} = z_1^\dagger (z_1 z_1^\dagger)^{-1} z_1$, where $z_1 = L_{10}$. These operators satisfy the additional conditions in Eq.~\eqref{lem1fact32} due to the orthogonality relations of the projectors $\{\Pi_i\}$.

Without loss of generality, we assume that $\{z_1,z_2,z_3\}$ forms a complete basis. Moreover, the fact that $z_1 \perp y_1^0$, $z_2 \perp y_2^0$, $z_1 \perp y_1^1$, we can take
\begin{equation}\label{eq:wlog} 
\begin{split}
	& z_1 = [\id,0,0],\;\;\; z_2 = [0,\id,0],\;\;\; z_3 = [0,0,\id],\\
	& y^0_1 = [0,\id,A],\;\;\; y^0_2 = [\id,0,B],\;\;\; y^1_1 = [0,\id,C],
\end{split}
\end{equation}
where $A$, $B$, and $C$ are matrix variables to be determined.
The reason is the following. The vertices related to ${z_1,z_2,z_3}$ form a clique of size $3$, and the corresponding three projectors sum up to identity. Hence, as argued around Eq.~\eqref{eq:dec1}, we could have the first line of Eq.~\eqref{eq:wlog}. 
As argued around Eq.~\eqref{lem1fact3} and Eq.~\eqref{lem1fact32}, we could have the second line of Eq.~\eqref{eq:wlog}. Notice that the zero matrix in $y_1^0$ follows from the fact that the projector related to $y_1^0$ is orthogonal to the one 
related to $z_1$. Similarly for the others.
The same reasoning is used in the proofs for the other SI-C sets.

Since $h^0_1 \perp y^0_1$ and $h^0_1 \perp y^0_2$,
\begin{equation}
h^0_1 = [-B^\dagger, -A^\dagger, \id].
\end{equation}
The reason is the following. Denote $h_1^0 = [M_1, M_2, M_3]$, where $M_i$'s are invertible since the corresponding projector is not orthogonal to any projectors corresponding to the $z_i$'s. From the discussion around Eqs.~\eqref{lem1fact3} and \eqref{lem1fact32}, we can set $M_3 = \id$.
From the orthogonality relation between projectors, we know that $h_1^0 \perp y_1^0$. Thus, $0 = h_1^0 {y_1^0}^\dagger = M_2+A^\dagger$, which implies that $M_2 = -A^\dagger$. Similarly, the orthogonality relation $h_1^0 \perp y_2^0$ leads to $M_1 = -B^\dagger$.

Similarly,
$ y^0_3 \perp z_3$ and $y^0_3 \perp h^0_1 $ implies 
\begin{equation}
y^0_3 = [-AB^{-1},\id,0],
\end{equation}
$ h^2_3 \perp y^1_1$ and $h^2_3 \perp y^0_2 $ implies 
\begin{equation}
h^2_3 = [-B^\dagger,-C^\dagger,\id],
\end{equation}
$ y^2_3 \perp h^2_3$ and $y^2_3 \perp z_3 $ implies 
\begin{equation}
y^2_3 = [-CB^{-1},\id,0],
\end{equation}
$ h^0_2 \perp y^1_1$ and $h^0_2 \perp y^0_3 $ implies 
\begin{equation}
h^0_2 = [(BA^{-1})^\dagger,\id,-(C^{-1})^\dagger],
\end{equation}
$ y^1_2 \perp h^0_2$ and $y^1_2 \perp z_2 $ implies 
\begin{equation}
y^1_2 = [AB^{-1},0,C],
\end{equation}
$ h^1_3 \perp y^0_1$ and $h^1_3 \perp y^1_2 $ implies 
\begin{equation}
h^1_3 = [-(BA^{-1}C)^\dagger,-A^\dagger,\id],
\end{equation}
$ y^1_3 \perp h^1_3$ and $y^1_3 \perp y^1_2 $ implies 
\begin{equation}
y^1_3 = [-C^{-1}AB^{-1},A^{-1},\id],
\end{equation}
$ h^1_1 \perp y^1_1$ and $h^1_1 \perp y^1_3 $ implies 
\begin{equation}
h^1_1 = [(BA^{-1}CA^{-1})^\dagger,\id,-(C^{-1})^\dagger],
\end{equation}
$ y^2_2 \perp h^1_1$ and $y^2_2 \perp z_2 $ implies 
\begin{equation}
y^2_2 = [C^{-1}AC^{-1}AB^{-1},0,\id],
\end{equation}
$ h^1_2 \perp y^0_2$ and $h^1_2 \perp y^1_3 $ implies 
\begin{equation}
h^1_2 = [\id,(AC^{-1}AB^{-1})^\dagger,-(B^{-1})^\dagger],
\end{equation}
$ y^2_1 \perp h^1_2$ and $y^2_1 \perp z_1 $ implies 
\begin{equation}
y^2_1 = [0,\id,AC^{-1}A],
\end{equation}
$ h^0_3 \perp y^0_3$ and $h^0_3 \perp y^2_2 $ implies 
\begin{equation}
h^0_3 = [\id,(AB^{-1})^\dagger,-(C^{-1}AC^{-1}AB^{-1})^\dagger],
\end{equation}
$ h^2_1 \perp y^1_2$ and $h^2_1 \perp y^2_1 $ implies 
\begin{equation}
h^2_1 = [-(BA^{-1}C)^\dagger, -(AC^{-1}A)^\dagger,\id],
\end{equation}
$ h^2_2 \perp y^0_1$ and $h^2_2 \perp y^2_3 $ implies 
\begin{equation}
h^2_2 = [(BC^{-1})^\dagger,\id,-(A^{-1})^\dagger].
\end{equation}
Then, $h^2_1 \perp y^2_3$ implies that 
\begin{equation}
(AC^{-1})^3 = \id. 
\end{equation}
In addition,
$ h^0_1 \perp h^1_1 $ implies 
\begin{equation}
B^\dagger B A^{-1}CA^{-1} + A^\dagger + C^{-1} = 0,
\end{equation}
$ h^0_1 \perp h^2_1 $ implies 
\begin{equation}
B^\dagger BA^{-1}C + A^\dagger AC^{-1}A + \id = 0,
\end{equation}
$ h^1_1 \perp h^2_1 $ implies 
\begin{equation}
(BA^{-1}CA^{-1})^\dagger (BA^{-1}C) + AC^{-1}A + (C^{-1})^\dagger = 0.
\end{equation}

By making use $(AC^{-1})^3 = \id$, we obtain
\begin{align}\label{eq:bbcbb}
& B^\dagger B + A^\dagger C + C^{-1}A = 0\\
& B^\dagger B + A^\dagger A C^{-1}A + A^{-1}C = 0,\\
& C^\dagger + A^{-1} = C^\dagger C A^{-1} + C^{-1}.
\end{align}
Similarly, $h^0_2 \perp h^1_2$ implies
\begin{align}
B^\dagger B + A^\dagger AC^{-1}A + (C^{-1}A)^\dagger = 0,
\end{align}
$h^1_2 \perp h^2_2$ implies 
\begin{align}
B^\dagger B + (C^\dagger A C^{-1}A)^\dagger + A^{-1}C = 0.
\end{align}
Hence, we have
\begin{align}\label{eq:bbcconds}
& A^\dagger A = C^\dagger C, A A^\dagger = C C^\dagger,\nonumber\\
& C^\dagger + A^{-1} = A^\dagger + C^{-1}.
\end{align}
Since we still have the freedom to choose the basis for the subspaces related to $z_2$ and $z_3$, we can assume that $A$ is diagonal and non-negative. Since $A$ is invertible, all the diagonal items are positive. We claim that $A = \id$, otherwise, without loss of generality, denote
\begin{equation}
A = \begin{bmatrix} X&0\\0&\id \end{bmatrix},
\end{equation}
where $X$ is a diagonal matrix whose diagonal terms are positive and different than $1$. 

By solving Eq.~\eqref{eq:bbcconds}, we obtain that
\begin{equation}
C = \begin{bmatrix} X & 0\\ 0& C_{22} \end{bmatrix},
\end{equation}
where $C_{22}$ is invertible. 
Hence, there is a state $\ket{s}$ in the subspace corresponding to the block $X$ such that
\begin{equation}
A\ket{s} = C \ket{s} = C^\dagger \ket{s} = x \ket{s},
\end{equation}
where $x>0$. Therefore,
\begin{equation}
\bra{s} (B^\dagger B + A^\dagger C + C^{-1}A) \ket{s} = \bra{s} B^\dagger B \ket{s} + (x^2+1) > 0,
\end{equation}
which contradicts Eq.~\eqref{eq:bbcbb}. Hence, 
\begin{align}\label{eq:bbccs}
& A = \id,\;\;\; CC^\dagger = \id,\;\;\; C^3 = \id,\nonumber\\
& B^\dagger B + C + C^{-1} = 0.
\end{align}
Note that, if we rotate the basis of the subspaces span by $z_2$ and $z_3$ with the same unitary $U$, this does not affect $A = \id$ and $C$ is changed to $UCU^\dagger$. By choosing a suitable $U$, $UCU^\dagger$ is a diagonal matrix according to the spectral theorem of norm matrix. Therefore, we can assume that $C$ is diagonal. 
Notice that $\id - C^3 = (\id -C)(\id + C + C^2) = 0$, then we have $\id + C + C^2 = 0$, otherwise $\id = C$, which leads to the contradiction that $BB^\dagger = -2\id$. 
Consequently, all the diagonal terms in $C$ are $e^{\pm 2\pi \mathbbm{i}/ 3}$ and $B^\dagger B = \id$. That is, $B$ is unitary. Since we still have the freedom to choose the basis of the subspace spanned by $z_1$, we can assume that $B=\id$.

Although we cannot change all the diagonal terms $e^{\pm 2\pi \mathbbm{i}/ 3}$ in $C$ to be $e^{2\pi \mathbbm{i} /3}$ with unitary, it can be done with the time-reversal operator in some dimensions which changes $\mathbbm{i}$ into $-\mathbbm{i}$. The time-reversal operator is also isometric. Therefore, up to isometry, $\Pi_i = \ket{v_i}\bra{v_i} \otimes \id$, where $|v_i\rangle$ is the normalized vector of the $i$-th column, $\bar{x}=-x$, $q = e^{2\pi \mathbbm{i}/3}$, and $g=q^2$.

%%%%%%%%%%%%%%%%%%%%%%%%%%%%%%%%%%%%%%%%%%%%%%%%%%%%%%%%%%%%%%%%%%%

\subsection{CEG-18 and its SI-C witness}

%%%%%%%%%%%%%%%%%%%%%%%%%%%%%%%%%%%%%%%%%%%%%%%%%%%%%%%%%%%%%%%%%%%

CEG-18 is the set of $18$ rank-one projectors in $d=4$ shown in Table~\ref{tab:ceg18}. 
Its graph of compatibility of CEG-18 is depicted in Fig.~\ref{fig:ceg18}.
CEG-18 was introduced in \cite{Cabello:1996PLA} and is an egalitarian Lov\'asz-optimum SI-C set and a critical KS set. It can be proven that CEG-18 is the 
KS set of rank-one projectors with the smallest possible cardinality (in any dimension!) \cite{Xu:2020PRL}.

%%%%%%%%%%%%%%%%%%%%%%%%%%%%%%%%%%%%%%%%%%%%%%%%%%%%%%%%%%%%%%%%%%%
% Table II (CEG-18)
%%%%%%%%%%%%%%%%%%%%%%%%%%%%%%%%%%%%%%%%%%%%%%%%%%%%%%%%%%%%%%%%%%%

\begin{widetext}
\begin{table*}[t!]
	\begin{tabular}{CCCCCCCCCCCCCCCCCCC}
		\hline\hline
		& v_1 & v_2 & v_3 & v_4 & v_5 & v_6 & v_7 & v_8 & v_9 & v_{A} & v_{B} & v_{C} & v_{D} & v_{E} & v_{F} & v_{G} & v_{H} & v_{I} \\
		\hline
		\rule{0pt}{3.2mm}
		v_{i1}\;\; & 1 & 0 & 0 & 0 & 1 & 1 & \bar{1} & 1 & t & 0 & t & \bar{1} & 1 & 1 & t & 0 & 0 & 0 \\
		v_{i2}\;\; & 0 & 1 & 0 & 0 & 1 & \bar{1} & 1 & 1 & 0 & t & 0 & 1 & 1 & \bar{1} & 0 & t & t & 0 \\
		v_{i3}\;\; & 0 & 0 & 1 & 0 & 0 & t & t & 0 & 1 & \bar{1} & \bar{1} & 0 & t & 0 & \bar{1} & \bar{1} & 1 & 1 \\
		v_{i4}\;\; & 0 & 0 & 0 & 1 & 0 & 0 & 0 & t & \bar{1} & \bar{1} & 1 & t & 0 & t & \bar{1} & 1 & \bar{1} & 1 \\ 
		\hline
		w_i\;\; & 1 & 1 & 1 & 1 & 1 & 1 & 1 & 1 & 1 & 1 & 1 & 1 & 1 & 1 & 1 & 1 & 1 & 1 \\
		\hline\hline
	\end{tabular}
	\caption{\label{tab:ceg18}{\bf CEG-18.} Each column $v_i$ corresponds to one observable represented by the projector $|v_i\rangle \langle v_i|$. The labels correspond to those in with Fig.~\ref{fig:ceg18}. The rows $v_{ij}$ give the components of $|v_i\rangle$ (unnormalized). $\bar{1}=-1$ and $t=\sqrt{2}$. The last row contains the weights of the optimal SI-C witness of the form~\eqref{eq:gnci}. The weights $w_{ij}$ in~\eqref{eq:gnci} can be chosen in any way that satisfies $w_{ij} \geq \max\{w_i, w_j\}$. With these weights, $\alpha(G,\vec{w})=4$ and $Q(G,\vec{w})=\vartheta(G,\vec{w})=\frac{9}{2}$.
	}
\end{table*}
\end{widetext}

%%%%%%%%%%%%%%%%%%%%%%%%%%%%%%%%%%%%%%%%%%%%%%%%%%%%%%%%%%%%%%%%%%%
% Fig. CEG-18
%%%%%%%%%%%%%%%%%%%%%%%%%%%%%%%%%%%%%%%%%%%%%%%%%%%%%%%%%%%%%%%%%%%

\begin{figure}[t!]
\centering \includegraphics[width=0.6\textwidth]{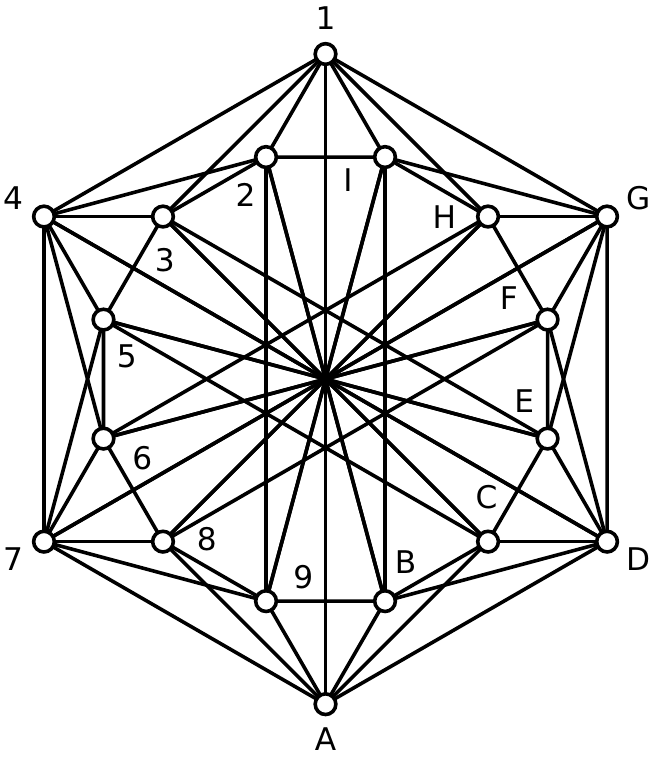}
\caption{\label{fig:ceg18} Graph of compatibility of CEG-18. Nodes represent observables and edges connect compatible observables. The labels refer to the observables in Table~\ref{tab:ceg18}.}
\end{figure}

%%%%%%%%%%%%%%%%%%%%%%%%%%%%%%%%%%%%%%%%%%%%%%%%%%%%%%%%%%%%%%%%%%%

The last row of Table~\ref{tab:ceg18} provides the weights corresponding to the optimal SI-C witness of the form~\eqref{eq:gnci}.

%%%%%%%%%%%%%%%%%%%%%%%%%%%%%%%%%%%%%%%%%%%%%%%%%%%%%%%%%%%%%%%%%%%

\subsection{Proof that CEG-18 is unique up to unitary transformations}

%%%%%%%%%%%%%%%%%%%%%%%%%%%%%%%%%%%%%%%%%%%%%%%%%%%%%%%%%%%%%%%%%%%

Suppose that you have the $L_1,\ldots,L_{18}$ with the same relations of orthogonality as the vectors $v_1,\ldots,v_9,v_A,\ldots,v_I$ in Table~\ref{tab:ceg18} and Fig.~\ref{fig:ceg18}.
Without loss of generality, we can assume that
\begin{equation} 
\begin{split}
	& L_1 = [\id,0,0,0],\;\;\;\;L_2 = [0,\id,0,0],\\ 
	& L_3 = [0,0,\id,0],\;\;\;\;L_4 = [0,0,0,\id].
\end{split}
\end{equation} 
%Hereafter, for simplicity, we will use $\id$ rather than $\id_d$.

Without loss of generality, we can assume that
\begin{equation}
L_5 = [\id, A, 0, 0].
\end{equation}
Then, $L_6\perp L_4$ and $L_6\perp L_5$ imply 
\begin{align}
L_6 = [\id, -(A^{-1})^\dagger,B,0].
\end{align}
Similarly, $L_7\perp L_4$ and $L_7\perp L_5$ imply
\begin{align}
L_7 = [\id, -(A^{-1})^\dagger,C,0].
\end{align}
Since $L_7 \perp L_6$, 
\begin{equation}
C = - [\id+(AA^\dagger)^{-1}] (B^{-1})^\dagger.
\end{equation}
From $L_{18} \perp L_1$ and $L_{18}\perp L_2$, we can assume
\begin{equation}
L_{18} = [0,0,\id,D].
\end{equation}
Since $L_{17} \perp L_1$ and $L_{17} \perp L_{18}$,
\begin{align}
L_{17} = [0, E, \id, -(D^{-1})^\dagger].
\end{align}
In addition, $L_{16}\perp L_1$ and $L_{16}\perp L_{18}$ imply 
\begin{align}
L_{16} = [0,F,\id, -(D^{-1})^\dagger].
\end{align}
Since $L_{17} \perp L_{16}$ and $L_{16}\perp L_7$, 
\begin{equation}
F = [\id+(D D^\dagger)^{-1}] (E^{-1})^\dagger = -B^{-1}(A+(A^{-1})^\dagger).
\end{equation}
Hence
\begin{equation}\label{eq:uniD}
D D^\dagger = (B^{-1}AA^\dagger B)^{-1}.
\end{equation}
The relation $L_{17}\perp L_6$ implies
\begin{equation}
E = B^\dagger A.
\end{equation}
In addition, $L_{15}\perp L_2$, $L_{15}\perp L_{16}$, and $L_{15} \perp L_{17}$ imply 
\begin{align}
L_{15} = [\id,0,-(B^{-1})^\dagger,-(B^{-1})^\dagger D],
\end{align}
$L_8\perp L_3$, $L_8\perp L_6$, and $L_8\perp L_{15}$ imply
\begin{align}
L_8 = [\id,A,0,B(D^{-1})^\dagger],
\end{align}
$L_{14}\perp L_3$, $L_{14}\perp L_5$, and $L_{14}\perp L_{16}$ imply
\begin{align}
L_{14} = [\id,-(A^{-1})^\dagger,0,-(A^{-1})^\dagger F^\dagger D),
\end{align}
$L_{13}\perp L_4$, $L_{13}\perp L_{14}$, and $L_{13}\perp L_{15}$ imply
\begin{align}
L_{13} = [\id,A,B,0],
\end{align}
$L_9\perp L_2$, $L_9\perp L_7$, and $L_9\perp L_{18}$ imply
\begin{align}
L_9 = [\id,0,-(C^{-1})^\dagger, (D^{-1}C^{-1})^\dagger],
\end{align}
$L_{10}\perp L_1$, $L_{10}\perp L_7$, and $L_{10} \perp L_8$ imply
\begin{align}
L_{10} = [0,\id,(C^\dagger A)^{-1}, -A(B^\dagger)^{-1}D].
\end{align}
Since $L_9 \perp L_8$, we have
\begin{equation}
C = -B(D D^\dagger)^{-1} = -AA^\dagger B.
\end{equation}
From $L_{10} \perp L_9$ and $ L_{10}\perp L_{13}$, 
\begin{equation}
C = - (A^\dagger)^{-2}B,\;\;\;\;A = A^\dagger.
\end{equation}
Therefore, we obtain that $A$ is hermitian and $A^4 = \id$. Hence, the eigenvalues of $A$ can only be $\pm 1$ and $A$ is automatically unitary. We still have some freedom to choose different $A$, $B$, and $D$ by applying a global unitary. We can choose $A=\id$. Eq.~\eqref{eq:uniD} implies that $D$ is also unitary, we can also choose it to be $\id$.

Then, by $L_{16}\perp L_{13}$, we obtain that $BB^\dagger = 2\id$. Hence, we can similarly set $B = \sqrt{2}\id$. Consequently,
\begin{equation}
C = -\sqrt{2}\id,\;\;\;E = \sqrt{2}\id,\;\;\; F = -\sqrt{2}\id.
\end{equation}
$L_{12}\perp L_{3}$, $L_{12}\perp L_5$, and $L_{12}\perp L_{14}$ imply
\begin{align}
L_{12} = [\id,-\id,0,-\sqrt{2}\id].
\end{align}
$L_{11}\perp L_{10}$, $L_{11}\perp L_{12}$, and $L_{11}\perp L_{13}$ imply
\begin{align}
L_{11} = [\sqrt{2}\id,0,-\id,\id].
\end{align}
Therefore, up to a global unitary, 
\begin{equation}
\Pi_i = \ket{v_i}\bra{v_i} \otimes \id, \forall i=1,\ldots,18,
\end{equation}
where $|v_i\rangle$'s are the normalized columns in Table~\ref{tab:ceg18}.

%%%%%%%%%%%%%%%%%%%%%%%%%%%%%%%%%%%%%%%%%%%%%%%%%%%%%%%%%%%%%%%%%%%
% Table III (Peres-24)
%%%%%%%%%%%%%%%%%%%%%%%%%%%%%%%%%%%%%%%%%%%%%%%%%%%%%%%%%%%%%%%%%%%

\begin{widetext}
\begin{table*}[t!]
	\begin{tabular}{CCCCCCCCCCCCCCCCCCCCCCCCC}
		\hline\hline 
		& v_1 & v_2 & v_3 & v_4 & v_5 & v_6 & v_7 & v_8 & v_9 & v_{10} & v_{11} & v_{12} & v_{13} & v_{14} & v_{15} & v_{16} & v_{17} & v_{18} & v_{19} & v_{20} & v_{21} & v_{22} & v_{23} & v_{24} \\
		\hline \rule{0pt}{3.2mm}
		v_{i1}\;\; & 1 & 0 & 0 & 0 & 1 & 1 & 1 & 1 & 1 & 1 & 1 & \bar{1} & 1 & 1 & 0 & 0 & 1 & 1 & 0 & 0 & 1 & 1 & 0 & 0 \\
		v_{i2}\;\; & 0 & 1 & 0 & 0 & 1 & 1 & \bar{1} & \bar{1} & 1 & 1 & \bar{1} & 1 & 1 & \bar{1} & 0 & 0 & 0 & 0 & 1 & 1 & 0 & 0 & 1 & 1 \\
		v_{i3}\;\; & 0 & 0 & 1 & 0 & 1 & \bar{1} & 1 & \bar{1} & 1 & \bar{1} & 1 & 1 & 0 & 0 & 1 & 1 & 1 & \bar{1} & 0 & 0 & 0 & 0 & 1 & \bar{1} \\
		v_{i4}\;\; & 0 & 0 & 0 & 1 & 1 & \bar{1} & \bar{1} & 1 & \bar{1} & 1 & 1 & 1 & 0 & 0 & 1 & \bar{1} & 0 & 0 & 1 & \bar{1} & 1 & \bar{1} & 0 & 0 \\
		\hline
		w_i\;\; & 1 & 1 & 1 & 1 & 1 & 1 & 1 & 1 & 1 & 1 & 1 & 1 & 1 & 1 & 1 & 1 & 1 & 1 & 1 & 1 & 1 & 1 & 1 & 1 \\ \hline\hline
	\end{tabular}
	\caption{\label{tab:peres24} {\bf Peres-24.} Each column $v_i$ corresponds to one observable represented by the projector $|v_i\rangle \langle v_i|$. The rows $v_{ij}$ give the components of $|v_i\rangle$ (unnormalized). $\bar{1}=-1$. The last row contains the weights $w_i$ of the optimal SI-C witness of the form~\eqref{eq:gnci}. The weights $w_{ij}$ in~\eqref{eq:gnci} can be chosen in any way that satisfies $w_{ij} \geq \max\{w_i, w_j\}$.
		With these weights, $\alpha(G,\vec{w})=5$ and $Q(G,\vec{w})=\vartheta(G,\vec{w})=6$.
	}
\end{table*}
\end{widetext}

%%%%%%%%%%%%%%%%%%%%%%%%%%%%%%%%%%%%%%%%%%%%%%%%%%%%%%%%%%%%%%%%%%%

\subsection{Peres-24 and its SI-C witness}

%%%%%%%%%%%%%%%%%%%%%%%%%%%%%%%%%%%%%%%%%%%%%%%%%%%%%%%%%%%%%%%%%%%

Peres-24 is the set of $24$ rank-one projectors in $d=4$ shown in Table~\ref{tab:peres24}. It was introduced in \cite{Peres:1991JPA}. Unlike BBC-21 and CEG-18, Peres-24 is not critical (in the sense of Zimba and Penrose \cite{Zimba:1993SHPS}): some observables can be removed while still having a SI-C set. In turn, Peres-24 is a complete SI-C set.

The last row of Table~\ref{tab:peres24} provides the weights corresponding to the optimal SI-C witness of the form~\eqref{eq:gnci}. As shown in Table~\ref{tab:peres24}, Peres-24 is an egalitarian Lov\'{a}sz-optimum SI-C set.

%%%%%%%%%%%%%%%%%%%%%%%%%%%%%%%%%%%%%%%%%%%%%%%%%%%%%%%%%%%%%%%%%%%

\subsection{Proof that Peres-24 is unique up to unitary transformations}

%%%%%%%%%%%%%%%%%%%%%%%%%%%%%%%%%%%%%%%%%%%%%%%%%%%%%%%%%%%%%%%%%%%

Suppose that you have the projectors $\Pi_1, \ldots, \Pi_{24}$, and, correspondingly, $L_1,\ldots,L_{24}$ with the same relations of orthogonality as the vectors $v_1,\ldots,v_{24}$ in Table~\ref{tab:peres24}.
Without loss of generality, we can take $L_1 = [\id_d, 0, 0, 0]$, $L_2 = [0,\id_d,0,0]$, $L_3 = [0,0,\id_d,0]$, and $L_4 = [0,0,0,\id_d]$, where $\id_d$ is the identity operator in the Hilbert space of dimension $d$. Hereafter, for simplicity, we will omit the subindex $d$, and simply write $\id$. 

Let us now consider the projectors $L_{13}$, $L_{14}$, $L_{15}$, and $L_{16}$.
Since $L_{13} \perp L_{3}$, $L_{13} \perp L_4$, $L_{15} \perp L_{1}$, and $L_{15}\perp L_2$, by using Eq.~\eqref{lem1fact3}, we can assume that
\begin{equation}\label{eq:peresd}
L_{13} = [\id, D, 0, 0],\;\;\;\; L_{15} = [0,0,\id,E],
\end{equation}
where $D$ and $E$ are invertible matrices. Then, due to Eq.~\eqref{lem1fact32}, $L_{14} \perp L_3$, $L_{14} \perp L_4$, and $L_{14} \perp L_{13}$ imply
\begin{equation}
L_{14} = [\id,-(D^{-1})^\dagger,0,0].
\end{equation}
Similarly, $L_{16} \perp L_1$, $L_{16} \perp L_2$, and $L_{16} \perp L_{15}$ imply
\begin{equation}
L_{16} = [0,0,\id,-(E^{-1})^\dagger].
\end{equation}

Let us now consider the projectors $L_{17}$, $L_{18}$, $L_{19}$, and $L_{20}$.
Since $L_{17}\perp L_{2}$, $L_{17} \perp L_4$, $L_{19} \perp L_{1}$, and $L_{19}\perp L_3$, by applying Eq.~\eqref{lem1fact3}, we can take
\begin{equation}\label{eq:peresf}
L_{17} = [\id, 0, F, 0],\;\;\;\; L_{19} = [0,\id,0,G],
\end{equation}
where $F$ and $G$ are invertible matrices. Then, 
$L_{18} \perp L_2$, $L_{18} \perp L_4$, and $L_{18} \perp L_{17}$ imply
\begin{equation}
L_{18} = [\id,0,-(F^{-1})^\dagger,0].
\end{equation}
Similarly, $L_{20} \perp L_1$, $L_{20} \perp L_3$, and $L_{20} \perp L_{19}$ imply
\begin{equation}
L_{20} = [0,\id,0,-(G^{-1})^\dagger].
\end{equation}

Let us now consider the projectors $L_{5}$, $L_{6}$, $L_{7}$, and $L_{8}$. Then, $L_5\perp L_{14}$, $L_5\perp L_{16}$, and $L_5\perp L_{18}$ imply
\begin{equation}
L_5 = [\id,D,F,FE].
\end{equation}
Similarly, $L_6\perp L_{14}$, $L_6\perp L_{16}$, and $L_6\perp L_{17}$ imply
\begin{equation}
L_6 = [\id,D,-(F^{-1})^\dagger,-(F^{-1})^\dagger E].
\end{equation}
Also, $L_7\perp L_{13}$, $L_7\perp L_{15}$, and $L_7\perp L_{18}$ imply
\begin{equation}
L_7 = [\id,-(D^{-1})^\dagger,F,-F(E^{-1})^\dagger].
\end{equation}
Finally, $L_8\perp L_{13}$, $L_8\perp L_{15}$, and $L_8\perp L_{17}$ imply
\begin{equation}
L_8 = [\id,-(D^{-1})^\dagger,-(F^{-1})^\dagger,(F^{-1})^\dagger(E^{-1})^\dagger].
\end{equation}

From $L_5\perp L_{20}$ and $L_6\perp L_{19}$, 
\begin{equation}\label{eq:relation1}
DG = FE,\;\;\;\;F^\dagger D = E G^\dagger.
\end{equation}

Let us now consider the projectors $L_{9}$, $L_{10}$, $L_{11}$, and $L_{12}$.
Then, $L_9\perp L_{14}$, $L_9\perp L_{15}$, and $L_9\perp L_{18}$ imply
\begin{equation}
L_9 = [\id,D,F,-F(E^{-1})^\dagger].
\end{equation}
Also, $L_{10}\perp L_{14}$, $L_{10}\perp L_{15}$, and $L_{10}\perp L_{17}$ imply
\begin{equation}
L_{10} = [\id,D,-(F^{-1})^\dagger,(F^{-1})^\dagger (E^{-1})^\dagger].
\end{equation}
Similarly, $L_{11}\perp L_{13}$, $L_{11}\perp L_{16}$, and $L_{11}\perp L_{18}$ imply
\begin{equation} 
L_{11} = [\id,-(D^{-1})^\dagger,F,FE].
\end{equation}
Finally,
$L_{12}\perp L_{13}$, $L_{12}\perp L_{16}$, and $L_{12}\perp L_{18}$
imply
\begin{equation}L_{12} = [\id,-(D^{-1})^\dagger,-(F^{-1})^\dagger,-(F^{-1})^\dagger E].
\end{equation}

From $L_9\perp L_{19}$ and $L_{10}\perp L_{20}$,
\begin{equation}\label{eq:relation2}
D = F(E^{-1})^\dagger G^\dagger,\;\;\;\;D = (F^{-1})^\dagger E G^{-1}.
\end{equation}
From Eqs.~\eqref{eq:relation1} and \eqref{eq:relation2}, 
\begin{equation}\label{eq:unitariescond}
D D^\dagger = EE^\dagger = FF^\dagger = GG^\dagger = \id.
\end{equation}
We still have the freedom to apply a unitary on subspaces related to $L_1$, $L_2$, $L_3$, and $L_4$. Hence, we can set $D=E=F=\id$. Then, we have also $G=\id$ because of Eq.~\eqref{eq:relation1}.

Finally, let us consider the projectors $L_{21}$, $L_{22}$, $L_{23}$, and $L_{24}$.
Then,
$L_{21}\perp L_2$, $L_{21}\perp L_3$, and $L_{21}\perp L_6$ imply
\begin{equation}
L_{21} = [\id,0,0,\id].
\end{equation}
Also $L_{22}\perp L_2$, $L_{22}\perp L_3$, and $L_{22}\perp L_5$
imply
\begin{equation}
L_{22} = [\id,0,0,-\id].
\end{equation}
Similarly, $L_{23}\perp L_1$, $L_{23}\perp L_4$, and $L_{23}\perp L_6$ imply 
\begin{equation}
L_{23} = [0,\id,\id,0].
\end{equation}
Finally, $L_{24}\perp L_1$, $L_{24}\perp L_4$, and $L_{24}\perp L_5$
imply
\begin{equation}
L_{24} = [0,\id,-\id,0].
\end{equation}
Then, it easy to check that all the remaining orthogonality relations in Table~\ref{tab:peres24} are satisfied.

Using that for $i\in V$ and $(j,t) \in \{1,\dots,c\}$ such that $(i,j) \notin E$ and $(i,t) \in E$, each rank-$\kappa$ projector $\Pi_i$ that acts on $\mathbbm{C}^d$ can be decomposed as
\begin{equation} %\label{lem1fact3}
\Pi_i = L_i^\dagger (L_iL_i^\dagger)^{-1} L_i, \;\;\;\; L_i = [B_{1i},B_{2i},\ldots,B_{ci}],
\end{equation}
where $d=4\kappa$, it is easy to see that Peres-24 can be written in the form $\Pi_i = |v_i\rangle\langle v_i|\otimes \id$, where the $|v_i\rangle$'s are the normalized columns in Table~\ref{tab:peres24}.

%%%%%%%%%%%%%%%%%%%%%%%%%%%%%%%%%%%%%%%%%%%%%%%%%%%%%%%%%%%%%%%%%%%

\subsection{The Peres-Mermin square and its SI-C witness} 

%%%%%%%%%%%%%%%%%%%%%%%%%%%%%%%%%%%%%%%%%%%%%%%%%%%%%%%%%%%%%%%%%%%

Given 9 observables, $A$, $B$, $C$, $a$, $b$, $c$, $\alpha$, $\beta$, and $\gamma$, with possible outcomes $-1$ or $1$, the following inequality \cite{Cabello08} holds for any NCHV theory:
\begin{equation}\label{eq:pm}
\braket{ABC} + \braket{abc} + \braket{\alpha\beta\gamma} + \braket{Aa\alpha} + \braket{Bb\beta} - \braket{Cc\gamma} \overset{\rm NCHV}{\leqslant} 4.
\end{equation}
However, if we consider the following two-qubit observables:
\begin{equation}\label{eq:pm0}
\begin{bmatrix} A & B & C\\ a & b & c\\ \alpha & \beta & \gamma \end{bmatrix} = \begin{bmatrix} \sigma_z\otimes\id & \id\otimes\sigma_z&\sigma_z\otimes\sigma_z\\ \id\otimes\sigma_x&\sigma_x\otimes\id&\sigma_x\otimes\sigma_x\\ \sigma_z\otimes\sigma_x&\sigma_x\otimes\sigma_z&\sigma_y\otimes\sigma_y 
\end{bmatrix},
\end{equation} 
then the left-hand side of \eqref{eq:pm} is $6$, since, for any two-qubit state, 
\begin{equation} \braket{ABC}=\braket{abc}=\braket{Aa\alpha}=\braket{Bb\beta}=-\braket{Cc\gamma}=1.
\end{equation}
The observables in \eqref{eq:pm0} were introduced by Peres and Mermin \cite{Peres1990PLA,Mermin1990PRL}. The right-hand side of \eqref{eq:pm0} is usually referred to as the Peres-Mermin square or magic square.

The Peres-Mermin square is a SI-C set (although not of rank-one projectors) and inequality~\eqref{eq:pm} is equally violated by any quantum state in $d=4$ [although it is not of the form~\eqref{eq:gnci}].

%%%%%%%%%%%%%%%%%%%%%%%%%%%%%%%%%%%%%%%%%%%%%%%%%%%%%%%%%%%%%%%%%%%

\subsection{The relation between the Peres-Mermin square and Peres-24} 

%%%%%%%%%%%%%%%%%%%%%%%%%%%%%%%%%%%%%%%%%%%%%%%%%%%%%%%%%%%%%%%%%%%

The Peres-Mermin square is related to Peres-24 \cite{Peres:1991JPA}. Each row or column in \eqref{eq:pm0} contains compatible observables represented by operators whose product is $\id$, except for the last column, which is $-\id$. This implies that, according to quantum theory, only four events can happen for every set of three compatible observables. If, e.g., $[-+-|ABC]$ denotes the event: the results $-1$, $1$, and $-1$ are obtained when $A$, $B$, and $C$ are measured, respectively, then only the following $24$ events can happen: $[+++|ABC]$, $[+--|ABC]$, $[-+-|ABC]$, $[--+|ABC]$, $[+++|abc]$, $[+--|abc]$, $[-+-|abc]$, $[--+|abc]$, $[+++|\alpha\beta\gamma]$, $[+--|\alpha\beta\gamma]$, $[-+-|\alpha\beta\gamma]$, $[--+|\alpha\beta\gamma]$, $[+++|Aa\alpha]$, $[+--|Aa\alpha]$, $[-+-|Aa\alpha]$, $[--+|Aa\alpha]$, $[+++|Bb\beta]$, $[+--|Bb\beta]$, $[-+-|Bb\beta]$, $[--+|Bb\beta]$, $[++-|Cc\gamma]$, $[+-+|Cc\gamma]$, $[-++|Cc\gamma]$, and $[---|Cc\gamma]$. Each of these events is represented by a projector, which is the eigenprojector of the corresponding Hermitian operators with the corresponding eigenvalues. The $24$ projectors thus defined have the same orthogonality relations as the $24$ projectors of Peres-24. The relation between the Peres-Mermin square and Peres-24 allows us to prove that the Peres-Mermin square is unique up to unitary transformations.

For any set $\{A, B, C, a, b, c, \alpha,\beta,\gamma\}$ with the same relations of joint measurability given in Eq.~\eqref{eq:pm0} and whose products fulfil the same relationships as those fulfilled by the observables in Eq.~\eqref{eq:pm0}, we can obtain $24$ vectors with the same relations of orthogonality as those of Peres-24. Let us call $\{\Pi_i\}_{i=1}^{24}$ this set of rank-one projectors. As seen before, there is a unitary transformation $U$ such that $U(\Pi_i)U^\dagger = |v_i\rangle\langle v_i|\otimes \id := \tilde{\Pi}_i$, where $|v_i\rangle$ are the columns in Table~\ref{tab:peres24}. 

%%%%%%%%%%%%%%%%%%%%%%%%%%%%%%%%%%%%%%%%%%%%%%%%%%%%%%%%%%%%%%%%%%%

\subsection{Proof that the Peres-Mermin square is unique up to unitary transformations}

%%%%%%%%%%%%%%%%%%%%%%%%%%%%%%%%%%%%%%%%%%%%%%%%%%%%%%%%%%%%%%%%%%%

Since $A$, $B$, and $C$ are compatible, then
\begin{subequations}
\begin{align}
	& A^{+}B^{+} \succeq A^{+}B^{+}C^{+},\\
	& A^{+}B^{-} \succeq A^{+}B^{-}C^{-},\\ 
	& A^{-}B^{+} \succeq A^{-}B^{+}C^{-},\\
	& A^{-}B^{-} \succeq A^{-}B^{-}C^{+}, 
\end{align}
\end{subequations}
where $A^{+}$ and $A^{-}$ denote the projectors for the positive and negative eigenspaces of $A$, respectively. That is, $A = A^{+} - A^{-}$.
Hence, 
\begin{align}
U(A^{+})U^\dagger &= U(A^{+}B^{+} + A^{+}B^{-})U^\dagger \nonumber\\ 
& \succeq U(A^{+}B^{+}C^{+})U^\dagger + U(A^{+}B^{-}C^{-})U^\dagger \nonumber\\ 
& = U(\Pi_1)U^\dagger + U(\Pi_2)U^\dagger\nonumber\\ 
& = \tilde{\Pi}_1 + \tilde{\Pi}_2, \\
U(A^{-})U^\dagger & = U(A^{-}B^{+} + A^{-}B^{-})U^\dagger\nonumber\\ 
& \succeq U(A^{-}B^{+}C^{-})U^\dagger + U(A^{-}B^{-}C^{+})U^\dagger \nonumber\\ 
& = U(\Pi_3)U^\dagger + U(\Pi_4)U^\dagger\nonumber\\
& = \tilde{\Pi}_3 + \tilde{\Pi}_4. 
\end{align}
In addition,
\begin{equation}
A^{+} + A^{-} = \id
\end{equation}
and 
\begin{equation}
\tilde{\Pi}_1 + \tilde{\Pi}_2 + \tilde{\Pi}_3 + \tilde{\Pi}_4 = \id,
\end{equation}
which implies 
\begin{equation}
U(A^{+})U^\dagger = \tilde{\Pi}_1 + \tilde{\Pi}_2,\;\;\;\;U(A^{-})U^\dagger = \tilde{\Pi}_3 + \tilde{\Pi}_4.
\end{equation}
Consequently,
\begin{equation}
U(A) = U(A^{+})U^\dagger - U(A^{-})U^\dagger = \begin{bmatrix} 1&&&\\ &1&&\\ &&\bar{1}&\\ &&&\bar{1}\end{bmatrix} \otimes \id_{\kappa},
\end{equation}
where $\kappa=D/4$.
Similarly, other observables in the realization can be mapped with the same unitary $U$ into the ones in Eq.~\eqref{eq:pm0}.

%%%%%%%%%%%%%%%%%%%%%%%%%%%%%%%%%%%%%%%%%%%%%%%%%%%%%%%%%%%%%%%%%%%
% Table IV (Peres-39)
%%%%%%%%%%%%%%%%%%%%%%%%%%%%%%%%%%%%%%%%%%%%%%%%%%%%%%%%%%%%%%%%%%%

\begin{table*}[t!]
\setlength{\tabcolsep}{0.6pt}
\begin{tabular}{CCCCCCCCCCCCCCCCCCCCCCCCCCCCCCCCCCCCCCCC}
	\hline \hline
	& v_1 & v_2 & v_3 & v_4 & v_5 & v_6 & v_7 & v_8 & v_9 & v_{10} & v_{11} & v_{12} & v_{13} & v_{14} & v_{15} & v_{16} & v_{17} & v_{18} & v_{19} & v_{20} & v_{21} & v_{22} & v_{23} & v_{24} & v_{25} & v_{26} & v_{27} & v_{28} & v_{29} & v_{30} & v_{31} & v_{32} & v_{33} & v_{34} & v_{35} & v_{36} & v_{37} & v_{38} & v_{39} \\
	\hline \rule{0pt}{3.2mm} 
	v_{i1}\;\; & 1 & 0 & 0 & 0 & 1 & 1 & 1 & 1 & 1 & 1 & 1 & 1 & 1 & 1 & 0 & 0 & 1 & 1 & 0 & 0 & 1 & 1 & 0 & 0 & 0 & 0 & 0 & 0 & 0 & 0 & 0 & 0 & 0 & 0 & 0 & 0 & 0 & 0 & 0 \\
	v_{i2}\;\; & 0 & 1 & 0 & 0 & 1 & 1 & \bar{1} & \bar{1} & 1 & 1 & \bar{1} & \bar{1} & 1 & \bar{1} & 0 & 0 & 0 & 0 & 1 & 1 & 0 & 0 & 1 & 1 & 0 & 1 & 1 & 1 & 1 & 1 & 1 & 1 & 1 & 0 & 0 & 0 & 0 & 1 & 1 \\
	v_{i3}\;\; & 0 & 0 & 1 & 0 & 1 & \bar{1} & 1 & \bar{1} & 1 & \bar{1} & 1 & \bar{1} & 0 & 0 & 1 & 1 & 1 & \bar{1} & 0 & 0 & 0 & 0 & 1 & \bar{1} & 0 & 1 & 1 & \bar{1} & \bar{1} & 1 & 1 & \bar{1} & \bar{1} & 0 & 0 & 1 & 1 & 0 & 0 \\
	v_{i4}\;\; & 0 & 0 & 0 & 1 & 1 & \bar{1} & \bar{1} & 1 & \bar{1} & 1 & 1 & \bar{1} & 0 & 0 & 1 & \bar{1} & 0 & 0 & 1 & \bar{1} & 1 & \bar{1} & 0 & 0 & 0 & 1 & \bar{1} & 1 & \bar{1} & 1 & \bar{1} & 1 & \bar{1} & 1 & 1 & 0 & 0 & 0 & 0 \\
	v_{i5}\;\; & 0 & 0 & 0 & 0 & 0 & 0 & 0 & 0 & 0 & 0 & 0 & 0 & 0 & 0 & 0 & 0 & 0 & 0 & 0 & 0 & 0 & 0 & 0 & 0 & 1 & 1 & \bar{1} & \bar{1} & 1 & \bar{1} & 1 & 1 & \bar{1} & 1 & \bar{1} & 1 & \bar{1} & 1 & \bar{1} \\
	\hline
	w_i\;\; & 24 & 8 & 8 & 8 & 4 & 4 & 4 & 4 & 4 & 4 & 4 & 4 & 6 & 6 & 7 & 7 & 6 & 6 & 7 & 7 & 6 & 6 & 7 & 7 & 24 & 4 & 4 & 4 & 4 & 4 & 4 & 4 & 4 & 6 & 6 & 6 & 6 & 6 & 6 \\
	\hline\hline
\end{tabular}
\caption{\label{tab:peres39}{\bf Peres-39.} Each column $v_i$ corresponds to one observable represented by the projector $|v_i\rangle \langle v_i|$. The rows $v_{ij}$ give the components of $|v_i\rangle$ (unnormalized). $\bar{1}=-1$. The last row contains the weights $w_i$ of the optimal SI-C witness of the form~\eqref{eq:peres39}. 
	With these weights, $\alpha(G,\vec{w})=46$ and $\vartheta(G,\vec{w})=50$.}
	\end{table*}
	
	%%%%%%%%%%%%%%%%%%%%%%%%%%%%%%%%%%%%%%%%%%%%%%%%%%%%%%%%%%%%%%%%%%%
	
	\subsection{Peres-39 and its witness}
	
	%%%%%%%%%%%%%%%%%%%%%%%%%%%%%%%%%%%%%%%%%%%%%%%%%%%%%%%%%%%%%%%%%%%
	
	One can obtain a KS set in $d=5$ by taking the $4$-dimensional vectors of Peres-24 and either appending or prepending $0$ to them \cite{Cabello05}. That is, if we call $\{\ket{u_i}\}_{i=1}^{24}$ the set of vectors in Peres-24, then
	\begin{equation}
\mathcal{V} = \{\ket{\mu_i}\}_{i=1}^{24} \cup \{\ket{\nu_i}\}_{i=1}^{24},
\end{equation} 
where
\begin{equation}
\bra{\mu_i} = (\bra{u_i},0),\;\;\;\; \bra{\nu_i} = (0,\bra{u_i}).
\end{equation}
is a KS set in $d=5$. $\mathcal{V}$ only contains $39$~vectors, since
\begin{align}
\ket{\mu_2} &= \ket{\nu_1},\;\; &\ket{\mu_3} &= \ket{\nu_2},\;\; &\ket{\mu_4} &= \ket{\nu_3}, \nonumber \\
\ket{\mu_{15}} &= \ket{\nu_{23}},\;\; &\ket{\mu_{16}} &= \ket{\nu_{24}},\;\; &\ket{\mu_{19}} &= \ket{\nu_{17}}, \nonumber \\
\ket{\mu_{20}} &= \ket{\nu_{18}},\;\; &\ket{\mu_{23}} &= \ket{\nu_{13}},\;\; &\ket{\mu_{24}} &= \ket{\nu_{14}}.
\end{align}

The resulting set is shown in Table~\ref{tab:peres39}. Hereafter, we will call it Peres-39. 

Let us consider the following witness:
\begin{equation}
\label{eq:peres39}
{\cal W}' := 
\sum_{C\in {\cal C}_5} \sum_{i\in C} P(\Pi_i =1) = \sum_{i=1}^{39} w_i P(\Pi_i=1) \le \alpha(G_{39},w),
\end{equation}
where $G_{39}$ is the graph of compatibility of Peres-39, ${\cal C}_5$ is the set of cliques of size $5$ in $G_{39}$, and $w_i$ is the frequency of $i$ in ${\cal C}$, which is shown in Table~\ref{tab:peres39}. Strictly speaking, ${\cal W}'$ is not a SI-C witness like those in Eq.~\eqref{eq:gnci} in which every $\Pi_i$ is in several contexts. Probably, there is a proper SI-C witness of the form \eqref{eq:gnci} for Peres-39, but we have not computational power to obtain it. Instead, we will assume that the projectors $\{\Pi_i\}$ in \eqref{eq:peres39} provides a quantum realization of the graph of orthogonality corresponding to $G_{39}$. Under this assumption, inequality \eqref{eq:peres39} holds.

For any state of $d=5$, $Q(G_{39},w) =\vartheta(G_{39},w) = |{\cal C}|$, which is the number of elements in ${\cal C}$ and therefore is also the algebraic maximum of ${\cal W}'$. This shows that Peres-39 is an egalitarian Lov\'{a}sz-optimum SI-C set.

%%%%%%%%%%%%%%%%%%%%%%%%%%%%%%%%%%%%%%%%%%%%%%%%%%%%%%%%%%%%%%%%%%%

\subsection{Proof that Peres-39 is unique up to unitary transformations}

%%%%%%%%%%%%%%%%%%%%%%%%%%%%%%%%%%%%%%%%%%%%%%%%%%%%%%%%%%%%%%%%%%%

By construction, in Peres-39, $\{\ket{\mu}_i\}_{i=1}^{24}$ contains only vectors orthogonal to $(0,0,0,0,1)$. Hence, all the basis of size $5$ which contains $(0,0,0,0,1)$ are just all the basis of size $4$ which are all orthogonal to $(0,0,0,0,1)$, i.e., all the complete basis in the subspace spanned by $\{\ket{\mu}_i\}_{i=1}^{24}$. 

Let us write Peres-39 as $\{P_i\}_{i=1}^{24} \cup \{Q_i\}_{i=1}^{24}$, where $P_i = Q_j$ if $\ket{\mu_i} = \ket{\nu_j}$. Then, $\{P_i\}_{i=1}^{24}$ is a realization of Peres-24 in the $4$-dimensional subspace spanned by $\{P_i\}_{i=1}^{24}$. 
Similarly, $\{Q_i\}_{i=1}^{24}$ is a realization of Peres-24 in the subspace spanned by $\{Q_i\}_{i=1}^{24}$. Then, we can apply that Peres-24 allows for CFR to each of them in their corresponding subspaces. 
Since the intersection of $\{P_i\}_{i=1}^{24},\{Q_i\}_{i=1}^{24}$ is not empty, all the projectors in the realization of Peres-39 have the same rank.

Without loss of generality, we can assume that
\begin{equation}
P_i = |\mu_i\rangle\langle \mu_i|\otimes \id_{d/5},\;\; \forall i=1,\ldots,24.
\end{equation}
In addition, $Q_4 = |\nu_4\rangle\langle\nu_4|\otimes \id_{d/5}$.
Then, we can copy for $\{Q_i\}_{i=1}^{24} $ the proof we used for Peres-24. Since $Q_i$ is fixed already for $i=1,2,3,13,14,17,18,23,24$, we know that $D=F=\id$ as in Eqs.~\eqref{eq:peresd} and~\eqref{eq:peresf} and, consequently, that
\begin{equation}
E = G,\;\;\;\; EG^\dagger = \id.
\end{equation}
Note that we still have the freedom to apply a local unitary to the subspace represented by $Q_4$. This implies that we can set $E=G=\id$, which fix the whole realization up to unitary transformations. Direct computation shows $\{P_i\}_{i=1}^{24} \cup \{Q_i\}_{i=1}^{24}$ realizes all the orthogonality relations in $G_{39}$.

%%%%%%%%%%%%%%%%%%%%%%%%%%%%%%%%%%%%%%%%%%%%%%%%%%%%%%%%%%%%%%%%%%%

\subsection{Proof that, for any $d \ge 6$, there are KS sets unique up to unitary transformations}

%%%%%%%%%%%%%%%%%%%%%%%%%%%%%%%%%%%%%%%%%%%%%%%%%%%%%%%%%%%%%%%%%%%

For any dimension $d\geqslant 4$, by patching $d-3$ copies of Peres-24 together (as we did in the construction leading to Peres-39), we can obtain a SI-C set in dimension $d$, which allows for CFR with respect to the witness \eqref{eq:peres39} and the orthogonality relations encoded in this set.

We can prove this recursively. Let us assume that we can prove it for the $d$-dimensional case. In the $(d+1)$-dimensional SI-C set, all the vectors whose last element is $0$ constitute, by construction, a $d$-dimensional SI-C set. We denote it as $\mathcal{P}$. In addition, we denote $\mathcal{Q}$ the last added Peres-24 SI-C set. Hence, all the maximal cliques of size $d+1$ which contain $(0,\ldots,0,1)$ correspond to all the maximal cliques of size $d$ for the $d$-dimensional SI-C set $\mathcal{P}$. As we discussed in the CFR of Peres-39, the fact that all the maximal cliques of size $d+1$ are complete bases leads to the fact that all the maximal cliques of size $d$ in $\mathcal{P}$ are complete bases for the subspace spanned by $\mathcal{P}$. Consequently, $\mathcal{P}$ allows for CFR.

By construction, the set of the vectors living in the last $5$~dimensions is a Peres-39, whose intersection with $\mathcal{P}$ is a Peres-24 $\mathcal{P}'$. After the CFR of $\mathcal{P}$, the Peres-24 $\mathcal{P}'$ is fixed. Same as in the CFR of Peres-39, this fixes the last added Peres-24 $\mathcal{Q}$ up to unitary transformations. Summing up, we have shown that the $(d+1)$-dimensional SI-C set allows for CFR.

%%%%%%%%%%%%%%%%%%%%%%%%%%%%%%%%%%%%%%%%%%%%%%%%%%%%%%%%%%%%%%%%%%%
% Table V (YO-13)
%%%%%%%%%%%%%%%%%%%%%%%%%%%%%%%%%%%%%%%%%%%%%%%%%%%%%%%%%%%%%%%%%%%

\begin{table}
\begin{tabular}{CCCCCCCCCCCCCC}
	\hline\hline
	& v_1 & v_2 & v_3 & v_4 & v_5 & v_6 & v_7 & v_8 & v_9 & v_A & v_B & v_C & v_D \\
	\hline 
	\rule{0pt}{3.2mm} 
	v_{i1}\;\; & 1 & 0 & 0 & 0 & 0 & 1 & 1 & 1 & 1 & 1 & 1 & \bar{1} & 1 \\
	v_{i2}\;\; & 0 & 1 & 0 & 1 & 1 & 0 & 0 & 1 & \bar{1} & 1 & 1 & 1 & \bar{1} \\
	v_{i3}\;\; & 0 & 0 & 1 & 1 & \bar{1} & 1 & \bar{1} & 0 & 0 & 1 & \bar{1} & 1 & 1 \\
	\hline
	w_i\;\; & 3 & 3 & 3 & 3 & 3 & 3 & 3 & 3 & 3 & 2 & 2 & 2 & 2 \\
	\hline\hline
\end{tabular}
\caption{\label{tab:yo13}{\bf YO-13.} Each column $v_i$ corresponds to one observable represented by the projector $|v_i\rangle \langle v_i|$. The rows $v_{ij}$ give the components of $|v_i\rangle$ (unnormalized). $\bar{1}=-1$. The last row contains the weights $w_i$ of the optimal SI-C witness of the form~\eqref{eq:gnci}. The weights $w_{ij}$ in~\eqref{eq:gnci} can be chosen in any way that satisfies $w_{ij} \geq \max\{w_i, w_j\}$. With these weights, $\alpha(G,\vec{w})=11$ and, for any qutrit state, $Q(G,\vec{w})=\frac{35}{3}\approx 11.67$. However, $\vartheta(G,\vec{w}) \approx 11.977641$.}
\end{table}

%%%%%%%%%%%%%%%%%%%%%%%%%%%%%%%%%%%%%%%%%%%%%%%%%%%%%%%%%%%%%%%%%%%
% Fig. Yu-Oh 13
%%%%%%%%%%%%%%%%%%%%%%%%%%%%%%%%%%%%%%%%%%%%%%%%%%%%%%%%%%%%%%%%%%%

\begin{figure}[t!]
\centering
\includegraphics[width=0.62\textwidth]{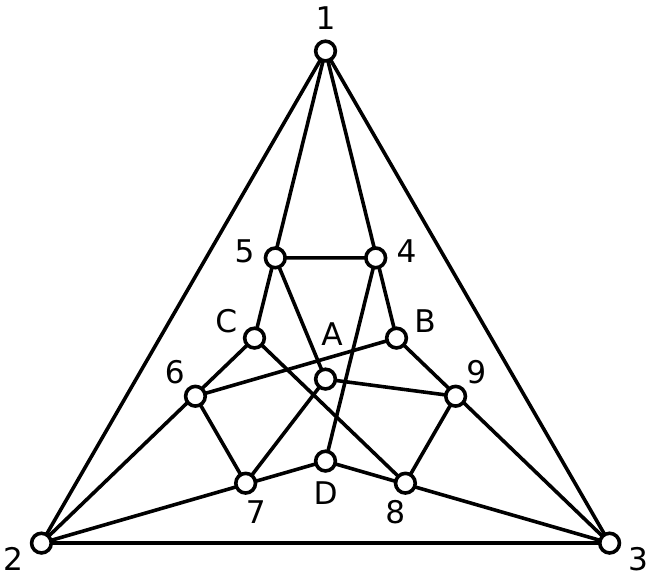}
\caption{\label{fig:yo13} Graph of compatibility of YO-13. Nodes represent observables and edges connect compatible observables. The nodes are labeled as the subindexes in Table~\ref{tab:yo13}.}
\end{figure} 

%%%%%%%%%%%%%%%%%%%%%%%%%%%%%%%%%%%%%%%%%%%%%%%%%%%%%%%%%%%%%%%%%%%

\subsection{YO-13}

%%%%%%%%%%%%%%%%%%%%%%%%%%%%%%%%%%%%%%%%%%%%%%%%%%%%%%%%%%%%%%%%%%%

YO-13 is the set of 13 rank-one projectors in $d =3$ shown in Table~\ref{tab:yo13} and whose graph of compatibility is shown in Fig.~\ref{fig:yo13}. YO-13 was introduced in \cite{Yu:2012PRL}. As shown in Table~\ref{tab:yo13},
YO-13 is an egalitarian SI-C set with respect to a contextuality witness of the form~\eqref{eq:gnci} and YO-13 is not Lov\'{a}sz-optimum. As proven in \cite{Cabello:2016JPA}, YO-13 is the smallest SI-C set of rank-one projectors in quantum theory (in any dimension). As it can be easily checked, YO-13 is not a KS set.

%%%%%%%%%%%%%%%%%%%%%%%%%%%%%%%%%%%%%%%%%%%%%%%%%%%%%%%%%%%%%%%%%%%

\subsection{Proof that YO-13, with normalization constraints, is unique up to unitary transformations}

%%%%%%%%%%%%%%%%%%%%%%%%%%%%%%%%%%%%%%%%%%%%%%%%%%%%%%%%%%%%%%%%%%%

YO-13 is unique up to unitary transformations if we assume the following normalization constraints:
\begin{subequations}
\begin{align}\label{eq:yoextra}
	p_1 + p_2 + p_3 = 1,\\
	p_1 + p_4 + p_5 = 1,\\
	p_2 + p_6 + p_7 = 1,\\
	p_3 + p_8 + p_9 = 1,
\end{align} 
\end{subequations}
where $p_i$ is the probability of obtaining the outcome $1$ when the projector $|v_i\rangle \langle v_i|$ is measured, with $v_i$ defined in Table.~\ref{tab:yo13}.

Following the same notation as before, without loss of generality, we can assume that
\begin{equation}
L_1 = [\id_\kappa,0,0],\;\;\; L_2 = [0,\id_\kappa,0],\;\;\; L_3 = [0,0,\id_\kappa].
\end{equation}
Hereafter, for simplicity, we will omit the subindex $\kappa$.
Since $L_1 \perp L_4$ and $L_1 \perp L_5$, we can assume that
\begin{equation}
L_4 = [0,\id,A],\;\;\;\; L_5 = [0,\id,A'].
\end{equation}
Then, since $L_4 \perp L_5$,
\begin{equation}
A' = -(A^{-1})^\dagger.
\end{equation}
Similarly, 
\begin{align}
&L_6 = [\id,0,B],\;\;\;\; L_7 = [\id,0,-(B^{-1})^\dagger],\\ &L_8 = [\id,C,0],\;\;\;\; L_9 = [\id,-(C^{-1})^\dagger,0].
\end{align}
Let us assume that
\begin{align}
&L_{A} = [\id,D,E],\;\;\;\; L_{B} = [\id,F,G],\\
&L_{C} = [\id,H,I],\;\;\;\; L_{D} = [\id,J,K].
\end{align}
Then, $L_{A} \perp L_{7}$ and $L_{A} \perp L_{9}$ imply
\begin{align}
D=C,\;\;\;\;E=B,
\end{align}
$L_{B} \perp L_{6}$ and $L_{B} \perp L_{9}$ imply 
\begin{align}
F=C,\;\;\;\; G=-(B^{-1})^\dagger,
\end{align}
$L_{C} \perp L_{6}$ and $L_{C} \perp L_{8}$ imply 
\begin{align}
H=-(C^{-1})^\dagger,\;\;\;\; I=-(B^{-1})^\dagger,
\end{align}
$L_{D} \perp L_{7}$ and $L_{D} \perp L_{8}$ imply 
\begin{align}
J=-(C^{-1})^\dagger,\;\;\;\; K=B.
\end{align}
In addition, $L_5 \perp L_{A}$ and $L_5 \perp L_{C}$ imply
\begin{align}
D = -E(A')^\dagger,\;\;\;\; H = -I(A')^\dagger.
\end{align}
This implies that,
\begin{equation}
B = CA,\;\;\;\; C = BA^\dagger.
\end{equation}
Then, $L_4 \perp L_{B}$ and $L_4 \perp L_{D}$ imply
\begin{equation}
F=-GA^\dagger,\;\;\;\; J=-KA^\dagger.
\end{equation}
This implies that,
\begin{equation}
C^{-1}=A B^\dagger,\;\;\;\; C^\dagger = AB^{-1}.
\end{equation}
Hence,
\begin{equation}
AA^\dagger = BB^\dagger = \id.
\end{equation}
Since we still have the freedom to rotate the subspaces corresponding to $L_2$ and $L_3$, we can set $A=B=\id$. Therefore,
\begin{equation}
C = D = E = F = K = \id,\;\;\;\; G = H = I = J= -\id. 
\end{equation}
This implies that there is an isometry between $\{\Pi_i\}_{i}$ and $\{2\ket{v_i}\bra{v_i} - \id_3\}_{i}$, where $\ket{v_i}$ is the normalized $i$-th column in Table.~\ref{tab:yo13}.

%%%%%%%%%%%%%%%%%%%%%%%%%%%%%%%%%%%%%%%%%%%%%%%%%%%%%%%%%%%%%%%%%%% 
% Table VI (Peres-33) 
%%%%%%%%%%%%%%%%%%%%%%%%%%%%%%%%%%%%%%%%%%%%%%%%%%%%%%%%%%%%%%%%%%%

\begin{table*}[th!]
\begin{tabular}{CCCCCCCCCCCCCCCCCCCCCCCCCCCCCCCCCC}
	\hline\hline
	& v_1 & v_2 & v_3 & v_4 & v_5 & v_6 & v_7 & v_8 & v_9 & v_{10} & v_{11} & v_{12} & v_{13} & v_{14} & v_{15} & v_{16} & v_{17} & v_{18} & v_{19} & v_{20} & v_{21} & v_{22} & v_{23} & v_{24} & v_{25} & v_{26} & v_{27} & v_{28} & v_{29} & v_{30} & v_{31} & v_{32} & v_{33} \\ \hline \rule{0pt}{3.2mm} 
	v_{i1} & 1 & 0 & 0 & 0 & 0 & 1 & 1 & 1 & 1 & 0 & 0 & \bar{1} & 1 & \bar{1} & 1 & 0 & 0 & t & t & t & t & \bar{1} & \bar{1} & 1 & 1 & \bar{1} & \bar{1} & 1 & 1 & t & t & t & t \\
	v_{i2} & 0 & 1 & 0 & 1 & 1 & \bar{1} & 1 & 0 & 0 & \bar{1} & 1 & 0 & 0 & t & t & t & t & 0 & 0 & \bar{1} & 1 & \bar{1} & 1 & \bar{1} & 1 & t & t & t & t & \bar{1} & \bar{1} & 1 & 1 \\
	v_{i3} & 0 & 0 & 1 & \bar{1} & 1 & 0 & 0 & \bar{1} & 1 & t & t & t & t & 0 & 0 & \bar{1} & 1 & \bar{1} & 1 & 0 & 0 & t & t & t & t & \bar{1} & 1 & \bar{1} & 1 & \bar{1} & 1 & \bar{1} & 1 \\
	\hline
	w_i & 3& 3& 3& 1& 1& 1& 1& 1& 1& 1& 1& 1& 1& 1& 1& 1& 1& 1& 1& 1& 1& 1& 1& 1& 1& 1& 1& 1& 1& 1& 1& 1& 1\\
	\hline\hline
\end{tabular}
\caption{\label{tab:peres33}{\bf Peres-33.} 
	Each column $v_i$ corresponds to one observable represented by the projector $|v_i\rangle \langle v_i|$. The rows $v_{ij}$ give the components of $|v_i\rangle$ (unnormalized). $\bar{1}=-1$ and $t = \sqrt{2}$. The last row contains the weights $w_i$ of the optimal SI-C witness of the form~\eqref{eq:gnci}, where the weights $w_{ij}$ in~\eqref{eq:gnci} can be chosen in any way that satisfies $w_{ij} \geq \max\{w_i, w_j\}$. With these weights, 
	$\alpha(G,\vec{w})=12$ and $Q(G,\vec{w})=\vartheta(G,\vec{w})=13$ for all qutrit states.}
	\end{table*}
	
	%%%%%%%%%%%%%%%%%%%%%%%%%%%%%%%%%%%%%%%%%%%%%%%%%%%%%%%%%%%%%%%%%%% 
	
	\subsection{Peres-33 and its optimum contextuality witness}
	
	%%%%%%%%%%%%%%%%%%%%%%%%%%%%%%%%%%%%%%%%%%%%%%%%%%%%%%%%%%%%%%%%%%% 
	
	Peres-33 is the KS set of rank-one projectors in $d=3$ shown in Table~\ref{tab:peres33}. It was introduced in \cite{Peres:1991JPA}. 
	It is the KS set in $d=3$ with the smallest number of bases known ($16$).
	
	%%%%%%%%%%%%%%%%%%%%%%%%%%%%%%%%%%%%%%%%%%%%%%%%%%%%%%%%%%%%%%%%%%% 
	
	\subsection{Proof that Peres-33 is not unique up to unitary transformations}
	
	%%%%%%%%%%%%%%%%%%%%%%%%%%%%%%%%%%%%%%%%%%%%%%%%%%%%%%%%%%%%%%%%%%% 
	
	The existence of unitarily inequivalent orthogonality representations of the orthogonality graph of Peres-33 has already been proven in Refs.~\cite{gould2010isomorphism,bengtsson2012gleason}. For completeness, we provide another proof here.
	
	As shown below, Peres-33 contains, induced, three copies of YO-13. Specifically, using the notation of Table~\ref{tab:peres33}, the three copies are
	\begin{subequations}
\label{se131}
\begin{align}
	S_1 & = \{v_1,v_4,v_5,v_2,v_3,v_{30},v_{33},v_{31},v_{32},v_{14},v_{13},v_{15},v_{12}\},\\
	S_2 & = \{v_2,v_8,v_9,v_1,v_3,v_{26},v_{29},v_{27},v_{28},v_{20},v_{11},v_{21},v_{10}\},\\
	S_3 & = \{v_3,v_6,v_7,v_1,v_2,v_{22},v_{25},v_{23},v_{24},v_{18},v_{17},v_{19},v_{16}\}.
\end{align}
\end{subequations}
The graph of compatibility of each copy $S_k$ corresponds to the graph in Fig.~\ref{fig:yo13}, assuming that the ordering of the vectors in Eqs.~\eqref{se131} is the same used in Table~\ref{tab:yo13}.
Notice that $S_i \cap S_j = (1,2,3) := S_0$ for $i \neq j$, but otherwise, the three sets are not tightly connected to each other.

A direct calculation shows that there is another realization of the graph of compatibility of Peres-33, where
\begin{equation}
\langle u_i| = 
\begin{cases}
	(v_{i1}, v_{i2}, v_{i3}), i\in S_1\\
	(v_{i1}, v_{i2}, \mathbbm{i} v_{i3}), i\in S'_2\\
	(v_{i1}, - \mathbbm{i} v_{i2}, v_{i3}), i\in S'_3,
\end{cases}
\end{equation}
where $v_{ij}$ are the ones in Table~\ref{tab:peres33} and $S'_2 = S_2 \setminus S_0$, $S'_3 = S_3 \setminus S_0$. However, $\{|v_i\rangle\}_{i=1}^{33}$ and $\{|u_i\rangle\}_{i=1}^{33}$ cannot be transformed to each other by either a unitary or an antiunitary transformation, since the set $\{|\langle u_i | u_j\rangle|\}$ is different from $\{|\langle v_i | v_j \rangle|\}$.

%%%%%%%%%%%%%%%%%%%%%%%%%%%%%%%%%%%%%%%%%%%%%%%%%%%%%%%%%%%%%%%%%%%

\section{Tools used in the proof of Result~2}
\label{app:C}

%%%%%%%%%%%%%%%%%%%%%%%%%%%%%%%%%%%%%%%%%%%%%%%%%%%%%%%%%%%%%%%%%%%

Consider the situation in which we are given a set of black boxes, each of them supposedly implementing an ideal measurement of one of the elements of a SI-C set $\{\Pi_i\}$ with graph of compatibility $G$ (with vertex set $V$ and edge set $E$) and whose optimal noncontextuality inequality of the form \eqref{eq:gnci} has weights $\{w_i\}$ given in the proof of Result~1.

Firstly, we consider the case with perfect orthogonality relations and the imperfectness is only in the violation.
\begin{lemma}\label{lemma:rank}
If, for any quantum state, 
\begin{equation}
	P(\Pi_i=\Pi_j=1)=0
\end{equation}
for any $(i,j) \in E$ and 
\begin{equation}
	\sum_{i \in V} w_i P(\Pi_i=1) > Q - \epsilon,
\end{equation}
where $Q$ is $\frac{35}{3}$, $\frac{9}{2}$, $6$, and $10$, and $\epsilon$ is $0.13159$, $0.13397$, $0.17712$, and $0.20808$ for, respectively, BBC-21, CEG-18, Peres-24, and YO-13, and, only in the case of YO-13, for any quantum state satisfying
\begin{subequations}
	\label{eq:yo13add}
	\begin{align}
		P(\Pi_1=1) + P(\Pi_2=1) + P(\Pi_3=1) &=1, \\ 
		P(\Pi_1=1) + P(\Pi_4=1) + P(\Pi_5=1) &=1, \\
		P(\Pi_2=1) + P(\Pi_6=1) + P(\Pi_7=1) &=1, \\
		P(\Pi_3=1) + P(\Pi_8=1) + P(\Pi_9=1) &=1,
	\end{align}
\end{subequations}
then,
\begin{itemize}
	\item[(I)] $\forall (i,j) \notin E$, $\rank(\Pi_j\Pi_i\Pi_j) = \rank(\Pi_j)$.
	
	\item[(II)] $\forall i \in V$, rank of $\Pi_i$ is the same, say, $\rank(\Pi_i) = \kappa$.
	
	\item[(III)] For $i\in V$ and $(j,t) \in \{1,\dots,c\}$ such that $(i,j) \notin E$ and $(i,t) \in E$, each rank-$\kappa$ projector $\Pi_i$ that acts on $\mathbbm{C}^d$ can be decomposed as
	\begin{equation}%\label{lem1fact3}
		\Pi_i = L_i^\dagger (L_iL_i^\dagger)^{-1} L_i, \;\;\;\; L_i = [B_{1i},B_{2i},\ldots,B_{ci}],
	\end{equation}
	where $L_i$ are $d \times \kappa$ matrices, $B_{ji}$ are $\kappa\times \kappa$ invertible matrices, $B_{ti}=0$ and we can take $B_{ji} = \id$ for any~$j$. Moreover, 
	\be %\label{lem1fact32}
	\Pi_i\Pi_t = 0 \implies L_i (L_t)^\dagger = 0 .
	\ee 
\end{itemize}
\end{lemma}
\begin{proof}
The proof of (I) follows from SDP. Notice that $\Pi_j \Pi_i \Pi_j \preceq \Pi_j$. If there is a pair $(i,j) \not\in E$ such that $\rank(\Pi_j\Pi_i\Pi_j) < \rank(\Pi_j)$, then there is a unit vector $\ket{v}$ in the subspace represented by $\Pi_j$ that is orthogonal to the subspace spanned by the image of $\Pi_j\Pi_i\Pi_j$. This leads to the linear conditions
\begin{equation}\label{eq:extralinear}
	\bra{v} \Pi_j \ket{v} = 1,\; \bra{v} \Pi_j\Pi_i\Pi_j \ket{v} = \bra{v} \Pi_i \ket{v} = 0.
\end{equation}
That is,
\be
T_{jj} = 1,\;\;\;\; T_{ii} = 0,
\ee
where $T_{lk} = \bra{v} \Pi_l\Pi_k \ket{v}$. By definition, $T_{lk}$ is positive semidefinite.
In addition, the orthogonal relations imply that
\begin{equation}\label{eq:orgorthogonal}
	T_{kl} = 0,\; \forall (k, l) \in E.
\end{equation}
The quantum violation of state $\ket{v}$ is a linear function of $T_{kk}$ whose upper bound can be calculated through the SDP under the conditions $\Pi_i\Pi_j = 0$ for any $(i,j) \in E$ and \eqref{eq:extralinear}. It cannot always be smaller than $Q - \epsilon$ for all $(i,j) \not\in E$. Hence, if, for any quantum state, the violation is not smaller than $Q - \epsilon$, then $\rank(\Pi_j\Pi_i\Pi_j) = \rank(\Pi_j)$ for any $(i,j) \not\in E$.\\

Proof of (II). From (I), $\forall (i,j) \not\in E$,
\begin{equation}
	\rank(\Pi_i) \ge \rank(\Pi_j\Pi_i\Pi_j) = \rank(\Pi_j),
\end{equation}
which implies that $\rank(\Pi_i) = \rank(\Pi_j)$ for any $(i,j) \not\in E$.
Since the complement of $G$ is connected, we conclude that the rank of all the projectors are same. \\ 

Proof of (III). Due to (I) and the existence of complete basis with $c$ projectors, we have the relation $d=c\kappa$. Then, each rank-$\kappa$ projector $\Pi_i$ can be decomposed as
\begin{equation}
	\Pi_i = (L'_i)^\dagger L'_i, \;\;\;\; L'_i = [B_{1i},B_{2i},\ldots,B_{ci}],
\end{equation} 
where $L'_i$ are some $d \times \kappa$ matrices and $B_{ji}$'s are some $\kappa\times \kappa$ matrices.
In the scenarios we considered, we can always choose $\{\Pi_j\}_{j=1}^c$ to be the complete basis and take it to be the standard basis, that is,
\begin{subequations}
	\begin{align}
		& L'_1 = [\id_\kappa,0,\ldots,0],\\
		& L'_2 = [0,\id_\kappa,\ldots,0],\\
		& \vdots \nonumber \\ 
		& L'_c = [0,0,\ldots,\id_\kappa].
	\end{align}
\end{subequations}
Using the above decomposition, we have
\begin{equation}\label{eq:lemma1eq1}
	\Pi_j \Pi_i \Pi_j = B_{ji}^\dagger B_{ji}.
\end{equation}
Due to (I) and Eq.~\eqref{eq:lemma1eq1}, for $(i,j) \not\in E$, $j\in\{1,\ldots,c\}$, the rank of $B_{ji}$ should be $\kappa$ and hence, this matrix $B_{ji}$ is invertible. While for $(i, t)\in E$ and $t\in\{1,\ldots,c\}$ the left-hand-side of Eq.~\eqref{eq:lemma1eq1} is zero and therefore $B_{ti}=0$. \\
Using that $B_{ji}$'s are either invertible or zero-matrices, a straight-forward calculation shows that 
\be \label{eq:lem1eq5}
\forall i, \ \Pi_i^2 = \Pi_i \implies L'_i (L'_i)^\dagger = \id_\kappa 
\ee 
and 
\be \label{eq:lem1eq6}
\Pi_i\Pi_t = 0 \implies L'_i (L'_t)^\dagger = 0 .
\ee 
By taking $L_i = B_{ji}^{-1}L'_i$ for some invertible $B_{ji}$ and by using Eq.~\eqref{eq:lem1eq5}, we can express $\Pi_i$ as
\beq 
\Pi_i & =& (L'_i)^\dagger (L'_i(L'_i)^\dagger)^{-1} L'_i \nonumber \\
& =& L_i^\dagger (L_iL_i^\dagger)^{-1} L_i .
\eeq 
Finally, notice that Eq.~\eqref{eq:lem1eq6} implies $L_iL_t^\dagger = 0$.
\end{proof}

%%%%%%%%%%%%%%%%%%%%%%%%%%%%%%%%%%%%%%%%%%%%%%%%%%%%%%%%%%%%%%%%%%%

\subsection{Robustness with imperfect orthogonality relations} \label{app:robustness}

%%%%%%%%%%%%%%%%%%%%%%%%%%%%%%%%%%%%%%%%%%%%%%%%%%%%%%%%%%%%%%%%%%%

To analyze the robustness when the orthogonality relations have not been exactly ensured, we introduce the following lemmas and one assumption.

\begin{lemma}\label{cl:id}
For a given matrix $T \preceq \lambda_1 \id$ in $\mathcal{H}$, if $\braket{s|T|s} = \lambda_1$, $\forall \ket{s} \in S$, where $S$ is a linear basis of $\mathcal{H}$, we have $T = \lambda_1 \id$ in $\mathcal{H}$.
\end{lemma}
\begin{proof}
Obviously, $\lambda_1$ should be the maximal eigenvalue of $T$. Consequently, $\braket{s|T|s} = \lambda_1$ implies that $\ket{s} \in E_{\lambda_1}$, where $E_{\lambda_1}$ is the eigenspace of $T$ for the eigenvalue $\lambda_1$.
Hence, by assumption, $S \subseteq E_{\lambda_1}$. By definition, $\mathcal{H} = \operatorname{span}(S)$, which implies
\begin{equation}
	\mathcal{H} \subseteq E_{\lambda_1}.
\end{equation}
Thus, $E_{\lambda_1} = \mathcal{H}$. Equivalently, $T = \lambda_1 \id$.
\end{proof}

\begin{lemma}
For any two projectors $\Pi_1, \Pi_2$ from a given setting, if $\rank(\Pi_1) > \rank(\Pi_2)$, then $\exists \ket{s} \in \mathcal{S}$ s.t. 
\begin{equation}
	\braket{s|\Pi_1|s} = 1,\;\;\;\; \braket{s|\Pi_2|s} = 0.
\end{equation}
\end{lemma}
\begin{proof}
The condition $\rank(\Pi_1) > \rank(\Pi_2)$ implies that the intersection $S_0$ of the subspace s represented by $\Pi_1$ and $\id-\Pi_2$ is not empty, denote $\Pi_0$ the projector onto subspace $S_0$. Von Neumann has proven that~\cite{von1950functional}
\begin{equation}
	\Pi_0 = \lim_{n\to\infty} [\Pi_1(\id-\Pi_2)]^n,
\end{equation}
where $\Pi_i, i=0,1,2$ represent the actions of projection instead of matrices.

By definition, there should be at least one state $\ket{s_0} \in \mathcal{S}$ such that $\Pi_0 \ket{s_0} \neq 0$, since $\Pi_0 \neq 0$. Thus, by repeating the projection $\Pi_1, \Pi_2$ with the initial state $\ket{s_0}$, we can finally obtain a state $\ket{s}$ such that $\braket{s|\Pi_0|s} = 1$.
By definition of $\Pi_0$, we know that $\Pi_1 \succeq \Pi_0, \Pi_2 \succeq \Pi_0$. Hence, we have
\begin{equation}
	\braket{s|\Pi_1|s} = \braket{s|\id-\Pi_2|s} = 1.
\end{equation}
\end{proof}
\begin{lemma}
For any two projectors $\Pi_1$ and $\Pi_2$ from a given setting, $\forall \epsilon>0$, $\exists \ket{s} \in \mathcal{S}$ such that
\begin{equation}
	\braket{s|\Pi_1\Pi_2\Pi_1|s} \ge \lambda_1(\Pi_1\Pi_2\Pi_1) - \epsilon,
\end{equation}
where $\lambda_1(\cdot)$ is the maximal eigenvalue.
\end{lemma}
\begin{proof}
If $\Pi_1 \perp \Pi_2$, any choice of $\ket{s}$ gives the conclusion. Otherwise, $\lambda_1(\Pi_1\Pi_2\Pi_1) > 0$. In addition, there is a state $\ket{s}$ such that $\Pi_1\Pi_2\Pi_1\ket{s} \neq 0$. Denote by $\ket{s_n}$ the post-measurement state after the repetition, $n$times, of measurements $\Pi_2$ and $\Pi_1$, that is,
\begin{equation}
	\ket{s_n} \propto (\Pi_1\Pi_2\Pi_1)^n \ket{s}.
\end{equation}
Since $\lim_{n\to \infty}(\lambda/\lambda_1)^n = 0$, for any eigenvalue $\lambda$ of $\Pi_1\Pi_2\Pi_1$ which is less than $\lambda_1$, we that
\begin{equation}
	\left[\frac{\Pi_1\Pi_2\Pi_1}{\lambda_1(\Pi_1\Pi_2\Pi_1)}\right]^n \to \Pi_0, n\to \infty,
\end{equation}
where $\Pi_0$ is the projector of the eigenspace of $\Pi_1\Pi_2\Pi_1$ with the maximal eigenvalue $\lambda_1$. Therefore,
\begin{equation}
	\braket{s_n|\Pi_1\Pi_2\Pi_1|s_n} \to \lambda_1(\Pi_1\Pi_2\Pi_1), n\to \infty.
\end{equation}
\end{proof}

For a given $c$-dimensional SI-C set which is considered here, and another realization of its orthogonality relations and quantum violation of a given witness with errors in experiment, we make the assumption of complete context for this realization.
\begin{assumption}[Complete context]
There is a complete context, i.e., a context with $c$ projectors, in which the relations of mutual exclusivity are perfect.
\end{assumption}
This can be guaranteed if we have a device whose different outcomes correspond to different projectors in this complete context. With out loss of generality, we label this complete context with $\{1,2,\ldots,c\} $.

As we can see, the proof of self-test only relies on two conditions:
\begin{enumerate}
\item The relations of exclusivity hold, i.e., $\Pi_i\Pi_j = 0$, $\forall (i,j)\in E$.
\item Each projector $\Pi_i$ can be decomposed into block form on a complete basis (a complete context in the scenario). That is,
\begin{equation}
	\Pi_i = L_i^\dagger L_i,\;\;\; L_i = [B_{1i}, B_{2i},\ldots,B_{ci}],
\end{equation}
where $B_{ti}$ is square invertible if $(i,t)\not\in E$, otherwise, $B_{ti} = 0$.
\end{enumerate}

In actual experiments, the first condition may not strictly hold due, i.e., to noise, so two mutually exclusive events of the ideal scenario may no be mutually exclusive.
Furthermore, the second condition is linked to the first and to the violation of the SI-C inequality.
If the first condition is valid, and the violation of the SI-C inequality is not too far from the violation in an ideal experiment, then the second condition is likewise true.
As a result, we have CFR {as we have shown before.}

For a given $(Q- \epsilon,\epsilon)$-SI-C realization, the relation of exclusivity $\Pi_i\Pi_j = 0$ is $(\epsilon,1/2) $-robust $\forall (i,j)\in E$. 
Since $ \Pi_i\Pi_j\Pi_i = (\Pi_j\Pi_i)^\dagger (\Pi_j\Pi_i)$, the maximal singular value of $\Pi_j\Pi_i$ is upper bounded by $\sqrt{\epsilon}$. That is, the first condition holds up to $\mathcal{O}(\sqrt{\epsilon})$.

Now we show that in a $(Q-\epsilon,\epsilon)$-realization, the second condition holds if $\epsilon < \min\{\epsilon_\tau, \epsilon_\nu\}$, where $\epsilon_\tau$ and $\epsilon_\nu$ indicate the invertibility of the blocks $B_{ki}$ in the decomposition and the completeness of the basis for the decomposition.
\begin{definition}
For a given graph $G$ and a set of weights $\{w_i\}$,
the tolerance function $\tau_{ij}(\theta,\epsilon)$ is defined as follows:
\begin{align}\label{eq:sdptau}
	\tau_{ij}(\theta,\epsilon) := &\min X_{ij} \nonumber \\
	\text{\rm subject to}& \sum\nolimits_{k=1}^n w_k X_{kk} = \theta, \nonumber \\
	&X_{kk} = X_{0k},\;\;\;1\le k \le n, \nonumber \\
	& X_{ij} \ge 0, \;\;\;X_{00} = X_{0i} = 1, \nonumber \\ 
	&|X_{kt}| \le \epsilon,\;\;\; (k,t)\in E, \nonumber \\
	& X \in \mathcal{S}_{+}^{1+n}.
\end{align}
\end{definition}

For CEG-18 and Peres-24, $\tau_{ij}(Q,0)$ is a strictly positive constant $\forall (i,j)\not\in E$ in both cases. 
For YO-13, we reach the same conclusion with the extra assumption in Eq.~\eqref{eq:yo13add}. When $\epsilon$ is not so large in comparison with $\tau_{ij}(Q,0)$, then the critical values of $\epsilon_\tau$ such that $\min_{(i,j)\not\in E}\tau_{ij}(Q- \epsilon, \epsilon) > 0$ are given in Table~\ref{tab:errorbound}.

For a given $(Q- \epsilon, \epsilon)$-realization, $\tau_{ij}(Q- \epsilon, \epsilon) > 0$ implies that, for any state $|s\rangle$ such that $\langle s|\Pi_i|s\rangle = 1$, we have $\langle s|\Pi_i \Pi_j \Pi_i|s\rangle > 0$. Therefore, $\Pi_i \Pi_j \Pi_i$ is positive definite in the subspace corresponding to $\Pi_i$. Hence, $\rank(\Pi_i\Pi_j\Pi_i)$ is no less than the dimension of this subspace, which is $\rank(\Pi_i)$. On the other hand, $\rank(\Pi_i\Pi_j\Pi_i) \le \rank(\Pi_i)$. Therefore,
\begin{equation}\label{eq:basisdec}
\rank(\Pi_i\Pi_j\Pi_i) = \rank(\Pi_i).
\end{equation}
Hence, $\rank(\Pi_j) \ge \rank(\Pi_i)$. If $\tau_{ji}(Q- \epsilon, \epsilon) > 0$ is also true, then we have $\rank(\Pi_j) = \rank(\Pi_i)$.
If the complement graph of the exclusivity $G$ is connected and $\tau_{ij}(Q-\epsilon,\epsilon) > 0$, $\forall (i,j)\not\in E$, we know that all the projectors should be of the same rank. If $\epsilon < \epsilon_\tau/2$, there is $\tau_0 > 0$ such that $\tau_{ij}(Q- \epsilon,\epsilon) > \tau_0$. This implies that the projection of $\Pi_j$ into the subspace spanned by $\Pi_i$ should be invertible and the inverse is bounded by $1/\tau_0$. In the case that $i=1,\ldots, c$, $\min\tau_{ij}(Q- \epsilon,\epsilon) = \min \lambda_{\min}(B_{ti}^\dagger B_{ti})$. Hence, $\lambda_{\min}(B_{ti}^\dagger B_{ti}) > \tau_0$ implies that $\lambda_{\max} [(B_{ti}^{-1})^\dagger B_{ti}^{-1}] \le 1/\tau_0$. Therefore, the maximal singular value $\sigma_{\max} (B_{ti}^{-1}) \le 1/\sqrt{\tau_0}$.

\begin{definition}
For a given graph $G$ and set of weights $\{w_i\}$,
the completeness function $\nu(\theta,\epsilon)$ is defined as
follows:
\begin{align}\label{eq:sdptau2}
	\nu(\theta,\epsilon):= & \min \sum_{k=1}^c X_{0k} \nonumber \\
	\text{\rm subject to}& \sum\nolimits_{k=1}^n w_k X_{kk} = \theta, \nonumber \\
	&X_{kk} = X_{0k},\;\;\;1\le k \le n, \nonumber \\
	& X_{ij} \ge 0,\;\;\; X_{00} = 1, \nonumber \\ 
	&|X_{kt}| \le \epsilon,\;\;\; (k,t)\in E, \nonumber \\
	& X \in \mathcal{S}_{+}^{1+n}.
\end{align}
\end{definition}
For a $(Q-\epsilon,\epsilon)$-realization, $\nu > 0$ for any small enough $\epsilon$. 
This means that any state is not orthogonal to all the projectors in the context $\{1,2,\ldots,c\} $. Under the complete context assumption that all the projectors in this context are orthogonal to each other, we have, $\sum_{i=1}^c \Pi_i = \id$. 
Consequentially, we have the block decomposition on this complete basis $\{\Pi_1, \Pi_2, \ldots, \Pi_c\}$.
Since $B_{ti}^\dagger B_{ti}$ is the non-trivial part of $\Pi_t \Pi_i \Pi_t$ whose maximal (diagonal) element is upper bounded by $\epsilon$, $|B_{ti}|_{\max} \le \sqrt{|B_{ti}^\dagger B_{ti}|_{\max}} \le \sqrt{\epsilon}$ if $(i,t)\in E$. {Here $|M|_{\max}$ is the max norm of the matrix $M$.}

In fact, rows of $L_i$ are a orthonormal basis of the subspace represented by $\Pi_i$. Then $\Pi_i\Pi_j\Pi_i \le \epsilon \id$ implies $L_i\Pi_j L_i^\dagger \le \epsilon \id$, which leads to 
\begin{equation}\label{eq:errorcase}
\sigma_{\max}(L_i L_j^\dagger) \le \sqrt{\epsilon},\ (i,j)\in E,
\end{equation}

Here we introduce the main idea of robustness analysis. As an example, we give the detailed analysis for YO-$13$ in the next section. The proof of self-testing in the ideal case implies that, the solution of Eq.~\eqref{eq:errorcase} is unique if $\epsilon = 0$. 
Since each $B_{ti}$ in $L_i$ there is invertible and the inverse is bounded, those facts lead to equations similar as in Eq.~\eqref{eq:unitariescond} but with error $O(\sqrt{\epsilon})$. Consequently, this results in the conclusion that all $\Pi_i$ is close to the ideal one up to $O(\sqrt{\epsilon})$. Note that, since a lot of substitution has been done in the proof of the self-testing in the ideal case, the constant in $O(\sqrt{\epsilon})$ depends on the number of equation in the proof of self-testing, i.e., the number of edges in the corresponding exclusivity graph. Apart from that, it depends on the upper bound of the inverse of $B_{ti}$ also for $(i,t)\not\in E$, which in turn depends on the error $\epsilon$ as we discussed before. All in all, our conclusion is that the self-testing is $\sqrt{\epsilon}$-robust when the error $\epsilon$ is small enough as suggested in Table~\ref{tab:errorbound}.

%%%%%%%%%%%%%%%%%%%%%%%%%%%%%%%%%%%%%%%%%%%%%%%%%%%%%%%%%%%%%%%%%%% 
% Table VII
%%%%%%%%%%%%%%%%%%%%%%%%%%%%%%%%%%%%%%%%%%%%%%%%%%%%%%%%%%%%%%%%%%% 

\begin{table}[]
\centering
\begin{tabular}{ccccc}
	\hline \hline
	SI-C set & BBC-21 & CEG-18 & Peres-24 & YO-13 \\
	\hline 
	$\epsilon_{\tau}$ & 0.00359 & 0.00557 & 0.00562 & 0.00296 \\
	$\epsilon_{\nu}$ & 0.00832 & 0.01527 & 0.01954 & 0.02325\\
	\hline\hline
\end{tabular}
\caption{Critical values of $\epsilon_{\tau}$ and $\epsilon_{\nu}$ such that $\min \tau_{ij}(Q-\epsilon_{\tau},\epsilon_{\tau}) > 0$ and $\nu(Q-\epsilon_{\nu},\epsilon_{\nu}) > 0$.}
\label{tab:errorbound}
\end{table}

%%%%%%%%%%%%%%%%%%%%%%%%%%%%%%%%%%%%%%%%%%%%%%%%%%%%%%%%%%%%%%%%%%% 

\subsection{Robustness analysis for YO-13}\label{app:robustyo13}

%%%%%%%%%%%%%%%%%%%%%%%%%%%%%%%%%%%%%%%%%%%%%%%%%%%%%%%%%%%%%%%%%%% 

As we discussed before, when the error $\epsilon \le \epsilon_\tau/2$ as shown in Table~\ref{tab:errorbound}, we can still have the decomposition
\begin{equation}
\Pi_i = L_i L_i^\dagger,\;\;\;\; L_i = [B_{1i}, B_{2i}, B_{3i}],
\end{equation}
where $\sigma_{\max}(B_{ti}) \le \sqrt{\epsilon}$ if $(t,i)\in E$, otherwise $B_{ti}$ is invertible and the inverse is bounded, i.e., $\sigma_{\max}(B_{ti}^{-1}) \le 1/\sqrt{\tau_0}$. Therefore, we can still, for example, change $L_i$ to $B_{ki}^{-1}L_i$, where $(k,i)\not\in E$, to simplify the the procedure of proof. Consequently,
\begin{equation}\label{eq:errorblock}
\sigma_{\max}(B_{ti}) \le \sqrt{\epsilon/\tau_0} \text{ if } (t,i)\in E.
\end{equation}

Without loss of generality, we assume that
\begin{equation}
L_1 = [\id_d,0,0],\;\;\; L_2 = [0,\id_d,0],\;\;\; L_3 = [0,0,\id_d].
\end{equation}
For simplicity, in the following derivation, we will omit the subindex $d$.

Since $L_1 \perp L_4$ and $L_1 \perp L_5$ up to some error (for simplicity, hereafter we will omit `up to some error'), 
\begin{equation}
L_4 = [O(\sqrt{\epsilon}),\id,A],\;\;\; L_5 = [O(\sqrt{\epsilon}),\id,A'],
\end{equation}
where $O(\sqrt{\epsilon})$ represents either a number or a matrix whose maximal singular value is not larger than $O(\sqrt{\epsilon})$. For brevity, hereafter we will denote $o = O(\sqrt{\epsilon})$.
$L_4\perp L_5$ implies 
\begin{equation}
\sigma_{\max}(AA' + \id + o) = o.
\end{equation}
Therefore,
\begin{equation}
\sigma_{\max}(A' + A^{-1}) \le O(\sqrt{\epsilon/\tau_0}) = o.
\end{equation}
In this sense, we denote $A' = -A^{-1} + o$.

Similarly,
\begin{align}
& L_6 = [\id,o,B],\;\;\; L_7 = [\id,o,-(B^{-1})^\dagger + o], \nonumber\\ 
& L_8 = [\id,C,o],\;\;\; L_9 = [\id,-(C^{-1})^\dagger + o, o].
\end{align}
Let us assume that
\begin{align}
& L_{10} = [\id,D,E],\;\;\; L_{11} = [\id,F,G], \nonumber \\ 
& L_{12} = [\id,H,I],\;\;\; L_{13} = [\id,J,K].
\end{align}
Then, 
$L_{10} \perp L_{6}$ and $L_{10} \perp L_{8}$ imply 
\begin{equation}
D=-(C^{-1})^\dagger+o,\;\;\;\; E=-(B^{-1})^\dagger+o,
\end{equation}
$L_{11} \perp L_{7}$ and $L_{11} \perp L_{9}$ imply 
\begin{equation}
F=C+o,\;\;\;\;G=B+o,
\end{equation}
$L_{12} \perp L_{6}$ and $L_{12} \perp L_{9}$ imply 
\begin{equation}
H=C+o,\;\;\;\;I=-(B^{-1})^\dagger+o,
\end{equation}
$L_{13} \perp L_{7}$ and $L_{13} \perp L_{8}$ imply 
\begin{equation}
J=-(C^{-1})^\dagger+o,\;\;\;\; K=B+o.
\end{equation}
In addition, $L_4 \perp L_{10}$ and $L_4 \perp L_{11}$ imply
\begin{equation}
D = -EA^\dagger+o,\;\;\;\; F = -GA^\dagger+o.
\end{equation}
Therefore,
\begin{equation}
(C^{-1})^\dagger = - (B^{-1})^\dagger A^\dagger+o,\;\;\;\;C = -BA^\dagger+o,
\end{equation}
which implies
\begin{equation}
C = -BA^{-1} + o = -BA^\dagger + o.
\end{equation}

Then, $L_4 \perp L_{12}$ and $L_5 \perp L_{13} $ imply 
\begin{equation}
H=IA^{-1}+o,\;\;\;\;J=KA^{-1}+o.
\end{equation}
This implies
\begin{equation}
C=-(B^{-1})^\dagger A^{-1} + o,\;\;\;\;-(C^{-1})^\dagger = BA^{-1} + o,
\end{equation}
which implies
\begin{equation}
C = -(B^{-1})^\dagger A^{-1} + o = -(B^{-1})^\dagger A^\dagger + o.
\end{equation}
Hence,
\begin{equation}
AA^\dagger = \id+o,\;\;\;\; BB^\dagger = \id + o.
\end{equation}
Since we still have the freedom to rotate the subspaces corresponding to $L_2$ and $L_3$, we can set $A$ and $B$ to be Hermitian. Therefore, the square of any eigenvalue of $A$ and $B$ is $1+o$, which means that the eigenvalues of $A, B$ are either $1+o$ or $-1+o$. Without loss of generality, we assume they are all $-1+o$. Hence, $A=-\id+o$ and $B = -\id+o$. Consequently,
\begin{equation}
C, F, G, H, K = -\id+o,\;\;\;\; D, E, I, J= \id+o. 
\end{equation}
Then, by definition of $L_i$, we know that its difference between the ideal one is also $o$, i.e., $O(\sqrt{\epsilon})$. So do $L_iL_i^\dagger$ and its inverse since each element in $L_i$ is bounded.

The fact that 
\begin{equation}
\Pi_i = L_i^\dagger (L_iL_i^\dagger)^{-1} L_i
\end{equation}
implies that $\Pi_i$ is also $O(\sqrt{\epsilon})$ close to the ideal one in the sense of maximal singular value.

%%%%%%%%%%%%%%%%%%%%%%%%%%%%%%%%%%%%%%%%%%%%%%%%%%%%%%%%%%%%%%%%%%% 

\section{Tools used in the proof of Result~3}
\label{app:D}

%%%%%%%%%%%%%%%%%%%%%%%%%%%%%%%%%%%%%%%%%%%%%%%%%%%%%%%%%%%%%%%%%%% 

Initially, it is important to note that a witness can be naturally constructed for any given complete KS set. Furthermore, the maximal violation of this witness is attained by any set of projectors in an arbitrary dimension that satisfies the orthogonality and completeness relations according to the orthogonality graph $G$ associated with the given KS set. In Result~3, we specifically refer to the CFR of a complete KS set with respect to the following particular witness.

Given a complete KS set, consider an SI-C witness of the form
\begin{equation}\label{Wks}
{\cal W}_{KS} = \sum_i w_i \ P(\Pi_i=1) \leqslant \alpha(G,\vec{w}) ,
\end{equation}
where $w_i$ is the number of bases in which projector $\Pi_i$ appears, and $\alpha(G,\vec{w})$ is the independence number of the orthogonality graph $G$ of the projectors $\{\Pi_i\}$. 
Let us denote the maximal cliques, each of which corresponds to a complete basis, by the subsets of the set of vertices $T_x \subset V$, wherein $x=1,\dots,m,$ and there are $m$ number of complete bases. Now the quantity
\beq \label{eq:ksset1}
\sum_i w_i \ \Pi_i 
= \sum_{x} \left( \sum_{i\in T_x} \Pi_i \right) = m \id,
\eeq 
where we have used the fact that the sum of projectors in each of these maximal cliques is the identity. Hence, the quantum value of the witness \eqref{Wks} is $m$ for any state. Note that $m$ is strictly greater than $\alpha(G,\vec{w})$, which follows from the definition of the KS set. Furthermore, since the maximal value of $\sum_{i\in T_x} P(\Pi_i=1)$ cannot be more than $1$, $m$ serves as the quantum upper bound. On the reverse direction, the witness \eqref{Wks} is maximally violated by a set of projectors $\{\Pi_i\}$ only when
\be 
\sum_{i\in T_x} \Pi_i = \id, \ \forall x.
\ee
Additionally, these projectors adhere to the orthogonality graph. Thus, any set of projectors providing the maximum violation of witness \eqref{Wks} must be a KS set according to the orthogonality graph $G$.

%%%%%%%%%%%%%%%%%%%%%%%%%%%%%%%%%%%%%%%%%%%%%%%%%%%%%%%%%%%%%%%%%%% 

\subsection{Concepts and previous results}

%%%%%%%%%%%%%%%%%%%%%%%%%%%%%%%%%%%%%%%%%%%%%%%%%%%%%%%%%%%%%%%%%%% 

\begin{definition}[Bipartite game]
A bipartite game $G = (X \times Y, A \times B, W)$ is a game involving two players, Alice and Bob. In each round of the game, Alice receives an input $x \in X$ and provides an output $a \in A$, and Bob receives an input $y \in Y$ and provides an output $b \in B$. Alice and Bob win the round if the inputs and outputs satisfy a winning condition $W \subseteq (X \times Y) \times (A \times B)$.
\end{definition}

\begin{definition}[Wining strategy]
A winning strategy for a bipartite game $G = (X \times Y, A \times B, W)$
is a strategy according to which for every $(x, y) \in X \times Y$, Alice and Bob output $a$ and $b$, respectively, such that $(x, y, a, b) \in W$.
\end{definition}

It is known (see, e.g., \cite{aolita2012pra}) that, for any KS set, there exists a context-projector KS game with a quantum winning strategy and no classical winning strategy.
Consider the graph of orthogonality $G$ of a complete KS set of $n$ projectors $\{\Pi_i\}_{i=1}^n$ in $\mathbbm{C}^d$. Each projector is represented by a vertex in the graph, and orthogonal projectors are represented by adjacent vertices. A clique in $G$ represents a set of mutually orthogonal projectors. A maximal clique of $G$ is a clique that cannot be extended by including one more adjacent vertex.
Let $T_x$ denote maximal cliques of $G$, having $|T_x|$ distinct elements. The elements of $T_x$ are $\{T_{x,a}\}$, where $x$ is the label of the clique and $a= 1,\ldots,|T_x|$ indicates the elements in that clique. In other words, $T_x := \{T_{x,1},\ldots, T_{x,|T_x|}\}$.
Suppose there are $m$ different maximal cliques, i.e., $x=1,\dots,m$. According to the graph $G$, Alice is given a maximal clique from the set of maximal cliques and has to output one of the vertices from that clique, while Bob receives a vertex from that clique and outputs either $0$ or $1$. They win the game if Alice outputs the vertex that is given to Bob and Bob outputs $1$, or if Alice outputs a vertex that is not given to Bob and Bob outputs $0$. Using the formal notation, Alice receives $x \in \{1,\dots,m\}$ and outputs $a \in \{1,\dots,|T_x|\}$, while Bob receives $y\in \{1,\dots,n\}$ and outputs $b\in \{0,1\}$. The payoff function that they aim to maximize is given by
\be 
\B = \sum_{a,b,x,y} c_{a,b,x,y} \ p(a,b|x,y)
\ee where 
\be
c_{a,b,x,y} = 
\begin{cases}
1, \quad \text{ if } y \in T_x, \ y \neq T_{x,a} \text{ and } b=0, \\
1, \quad \text{ if } y \in T_x, \ y = T_{x,a} \text{ and } b=1, \\
0, \quad\; \text{otherwise.}
\end{cases}
\ee

The following quantum strategy achieves perfect winning. Alice and Bob share the maximally entangled state
\begin{equation}
\ket{\phi^+_d} = \frac{1}{\sqrt{d}} \sum^{d-1}_{i=0} |ii\rangle .
\end{equation} 
Alice measures the following observable $A_x$ for input $x$ that corresponds to the basis $T_x$,
\begin{equation}\label{Ax}
A_x = \{ \Pi^*_{T_{x,1}}, \dots, \Pi^*_{T_{x,|T_x|}}\}.
\end{equation}
And Bob measures
\begin{equation}\label{By}
B_y = \{\id-\Pi_y,\Pi_y\}
\end{equation} 
for input $y$. Using the fact that $(A\otimes B)|\phi^+\ra = (\id \otimes BA^T) |\phi^+\ra$ for any operators $A,B$, we obtain the winning conditions for every pair of inputs.

To show that no winning classical strategy (without communication) exists, note that the best classical (local) strategy can be assumed to be deterministic, in which Bob assigns $0$ or $1$ values to the $n$ vertices. To win every round of the game, no two orthogonal projectors can be assigned both $1$, since every pair of projectors belongs to at least one basis in a complete KS set. Moreover, only one projector is assigned $1$ in every context. Such an assignment is impossible for a generalized KS set.

%%%%%%%%%%%%%%%%%%%%%%%%%%%%%%%%%%%%%%%%%%%%%%%%%%%%%%%%%%%%%%%%%%% 

\subsection{Proof of Result~3}

%%%%%%%%%%%%%%%%%%%%%%%%%%%%%%%%%%%%%%%%%%%%%%%%%%%%%%%%%%%%%%%%%%% 

The implication --- if a complete KS set cannot be certified with CFR then the corresponding context-projector KS game does not admit Bell self-testing --- is straightforward. The uncharacterized KS set acting on $\mathbbm{C}^D$ is denoted by $\{\Pi_i\}$, and the reference KS set acting on $\mathbbm{C}^{d}$ is denoted by $\{\overline{\Pi}_i\}$. Let us assume that these two KS sets are not connected by unitary transformations, that is,
\begin{equation}\label{ne}
\nexists \ U, \quad U \Pi_i U^\dagger = \overline{\Pi}_i \otimes \id \oplus \overline{\Pi}^*_i \otimes \id, \ \forall i.
\end{equation}
Now, consider the following two realizations of a context-projector KS game where the local observables are constructed from $\{\Pi_i\}$ and $\{\overline{\Pi}_i\oplus \overline{\Pi}^*_i\}$ according to \eqref{Ax}-\eqref{By}, and the shared states are $\ket{\phi^+_D}$ and $(1/2) \ket{\phi_d^+}\oplus \ket{\phi_d^+} $, respectively. Due to \eqref{ne}, there is no local unitary transformation on each side that can map the local measurements in one quantum strategy to the other. Therefore, the context-projector KS game does not admit Bell self-testing. 

To show the reverse implication, it suffices to establish that in any quantum winning strategy, Bob's measurements \eqref{By} $\{\id-B_y, B_y\}$ and Alice's measurements $\{A^a_x\}$ should be projective and constitute a KS set. Consequently, it follows that if there is no unitary transformation between two quantum strategies, then the KS set does not satisfy CFR.

Let $\rho$ be the shared state in the quantum strategy, and $d_A$ and $d_B$ be the local dimension of the reduced states. Without loss of generality, we can assume that the POVMs of Alice and Bob are in their respective $d_A$-dimensional and $d_B$-dimensional Hilbert spaces. If not, we can consider the projection of these POVMs into the corresponding subspaces.

\textit{Bob's measurements.} Let us denote $\rho_x^a = \tr_A(\rho A_x^a)$, the unnormalized reduced state on Bob's side when outcome $a$ is observed for the measurement setting $x$. The fact that $\sum_a A_x^a = \id$ implies
\begin{equation}\label{eq:sum1}
\sum_a\tr_B(\rho_x^a) = \tr_B\left(\sum_a \rho_x^a\right) = \tr\left(\rho \sum_a A_x^a\right) = 1.
\end{equation}
From the winning conditions of the context-projector game, we know that
\begin{equation}\label{eq:optcond}
p(\bar{a},1|x,y) + \sum_{a\neq \bar{a}} p(a,0|x,y) = 1,
\end{equation}
where $y = T_{x,\bar{a}}$, which translates to
\begin{equation}\label{eq:sum2}
\tr(\rho_x^{\bar{a}} B_y) + \sum_{a\neq \bar{a}} \tr\left(\rho_x^a (\id-B_y)\right) = 1.
\end{equation}
We represent $S_{x,a}$ as the subspace spanned by $\rho_x^a$, where $\id_{x,a}$ stands for the identity operator in $S_{x,a}$. Additionally, $B|_S$ denotes the restriction of the operator $B$ to the subspace $S$.
By combining Eqs.~\eqref{eq:sum1} and \eqref{eq:sum2}, we obtain
\be 
\label{aeq}
\tr(\rho_x^{\bar{a}} B_y) = \tr(\rho_x^{\bar{a}}) + \tr\left(\sum_{a\neq \bar{a}}\rho_x^{a} B_y\right) .
\ee 
Since $B_y \preceq \id$, it follows from Eq.~\eqref{aeq} that
\begin{equation}
B_y|_{S_{x,\bar{a}}} = \id_{x,\bar{a}},\ B_y|_{\tilde{S}_{x,\bar{a}}} = 0,
\end{equation}
where
\begin{equation}
\tilde{S}_{x,\bar{a}} = \oplus_{a\neq \bar{a}} S_{x,a}.
\end{equation}
Therefore, $S_{x,a} \perp {S}_{x,\bar{a}} $ if $a\neq \bar{a}$. Moreover, since 
$\tilde{S}_{x,\bar{a}}\oplus S_{x,\bar{a}}$ is the full space of Bob's local system,
\begin{equation}\label{ByP}
B_y = \id_{x,\bar{a}}.
\end{equation}
This leads to the orthogonality conditions, $B_yB_{y'} = 0$ whenever $y\neq y'$ and $y,y'\in T_x$. In addition, due to the fact that $\sum_a \rho_x^a$ is of full rank, the normalization condition also holds. That is,
\begin{equation}\label{Bysum}
\sum_{y\in T_x} B_y = \sum_a \id_{x,a} = \id.
\end{equation}

\textit{Alice's measurements.} Similarly, let us denote the reduced states on Alice's side by $\sigma_y = \tr_B(\rho B_y), \bar{\sigma}_y = \tr_B(\rho (\id-B_y))$, for the measurement setting $y$, such that
\begin{equation}
\tr(\sigma_y) + \tr(\bar{\sigma}_y) = 1.
\end{equation}
The winning condition Eq.~\eqref{eq:sum2} implies
\begin{equation}
\tr(\sigma_y A_x^{\bar{a}}) + \sum_{a\neq \bar{a}}\tr(\bar{\sigma}_y A_x^a) = 1,
\end{equation}
which leads to
\begin{equation}
\tr(\sigma_y A_x^{\bar{a}}) = \tr(\sigma_y),\;\;\;\; \tr(\bar{\sigma}_y A_x^{\bar{a}}) = 0.
\end{equation}
Hence,
\begin{equation}
A_x^{\bar{a}}|_{S_y} = \id_y,\;\;\;\;A_x^{\bar{a}}|_{\tilde{S}_y} = 0,
\end{equation}
where $S_y, \tilde{S}_y$ are the space spanned by $\sigma_y, \bar{\sigma}_y$, respectively. Moreover, we can infer that $S_y \perp S_{y'}\ \forall y,y'\in T_x, y\neq y'$. Since $\sigma_y + \bar{\sigma}_y = \tr_B(\rho)$ is of full rank, we know that
\begin{equation}\label{AxaP}
A_x^{\bar{a}} = \id_y.
\end{equation}
Subsequently, we have the orthogonality condition $\id_{y_1} \id_{y_2} = 0$ whenever $y_1 = T_{x,a_1}, y_2 = T_{x,a_2}$ and $a_1 \neq a_2$, and the normalization condition
\be \label{Axasum}
\sum_{a} A_x^a = \sum_{y\in T_x} \id_y = \id .
\ee 
In total, Eqs.~\eqref{ByP}, \eqref{Bysum}, \eqref{AxaP}, and \eqref{Axasum} imply that $\{B_y\}$ and $\{A_x^{\bar{a}}\}$ are two realizations of the KS set, and thus, if the local measurements in a quantum winning strategy cannot be Bell self-tested, then there exist inequivalent realizations of the KS set.

%%%%%%%%%%%%%%%%%%%%%%%%%%%%%%%%%%%%%%%%%%%%%%%%%%%%%%%%%%%%%%%%%%% 

\subsection{Bell self-testing of the maximally entangled state}

%%%%%%%%%%%%%%%%%%%%%%%%%%%%%%%%%%%%%%%%%%%%%%%%%%%%%%%%%%%%%%%%%%% 

If specific conditions are met by the complete KS set, it becomes possible to self-test the maximally entangled state. To simplify matters, here, we will not delve into the necessary and sufficient conditions. Instead, we will demonstrate that the maximally entangled state can be reliably self-tested if a complete KS set admits CFR and satisfies the following criteria: In one realization of the KS set, there exist two bases made of rank-one real projectors, $\{|i\rangle\!\langle i|\}_{i=1}^d$ and $\{|v_i\ra\!\la v_i|\}_{i=1}^d$, such that $\la i|v_j\ra \neq 0$ for all $i,j$. Note that, without loss of generality, we can consider one basis to be the computational basis.

To prove it, we suppose that the uncharacterized projectors $\{A_x^{\bar{a}}\}$ and $\{B_y\}$ act on $\mathbbm{C}^{d_A}$ and $\mathbbm{C}^{d_B}$, respectively, and the unknown shared state is $\rho$. Since the set of projectors admits Bell self-testing, there exist local unitaries $U_A$ and $U_B$ such that
\be 
U_A(A_x^{\bar{a}})U^\dagger_A = \overline{\Pi}_y\otimes \id_{\frac{d_A}{d}} , \ 
U_B(B_y)U_B^\dagger = \overline{\Pi}_y\otimes \id_{\frac{d_B}{d}} ,
\ee 
wherein $\{\overline{\Pi}_y\}$ are the projectors in those two bases in the reference $d$-dimensional KS set. Let us denote the reduced state of $(U_A\otimes U_B)(\rho)(U_A\otimes U_B)^\dagger$ onto the subspace $\mathbbm{C}^{d}\otimes \mathbbm{C}^{d}$ by $\rho_{AB}$. 

The winning conditions given by Eq.~\eqref{eq:optcond} implies
\be \label{eq:optcond1}
p(a,1|x,y) = 0, \text{ if } y\neq T_{x,a}.
\ee 
The relation must hold for any state $|\phi\ra$ that belongs to the support of $\rho_{AB}$, that is,
\beq \label{eq:PyPy'psi}
\forall (y, y')\in E, \ (\overline{\Pi}_y\otimes \overline{\Pi}_{y'}) |\phi\rangle = 0 .
\eeq 
After substituting $|\phi\rangle = \sum_{i,j} c_{ij} |ij\ra$ and the computational basis $\{|i\rangle\!\langle i|\}$ in Eq.~\eqref{eq:PyPy'psi}, we obtain that
\be \label{psiij0}
c_{ij} = 0, \text{ for } i\neq j.
\ee 
Therefore, we can express $\ket{\phi}$ as
\be 
|\phi\ra = (\id \otimes S) |\phi^+_d\ra,
\ee 
where $S$ is a diagonal matrix whose elements can be taken to be positive by exploiting the freedom of local unitary. Taking $y$ and $y'$ from the other basis $\{|v_i\ra\!\la v_i|\}_{i=1}^d$, it follows from \eqref{eq:optcond1} that
\be 
(|v_i\ra\!\la v_i| \otimes |v_j\ra\!\la v_j| ) (\id \otimes S) |\phi^+_d\ra = 0.
\ee 
Using the fact that $(A\otimes B)|\phi^+\ra = (\id \otimes BA^T) |\phi^+\ra$, we obtain 
\be 
(\id \otimes |v_j\ra\!\la v_j| S |v_i\ra\!\la v_i|) |\phi^+_d\ra = 0,
\ee 
which implies 
\be 
\la v_j| S |v_i\ra\ = 0 , \forall i\neq j.
\ee 
This relation can be rephrased as
\be 
\la j| U^\dagger S U |i\ra\ = 0 , \forall i\neq j,
\ee 
where $U$ is the unitary such that $|v_i\ra = U|i\ra$. Therefore, $U^\dagger S U$ is also diagonal. Furthermore, since unitary does not change the eigenvalues, there exists a permutation matrix $T$ such that $V = UT$ and
\begin{equation}
V^\dagger SV = S \Leftrightarrow [S,V] = 0 .
\end{equation}
Finally, due to the fact that $\la i|v_j\ra \neq 0$ for all $i,j$, all the elements in $U$ are non-zero, and thus, all the elements in $V$ are also non-zero. Consequently, the commutation relation $[S,V] = 0$ holds only if all the eigenvalues of $S$ should be the same, i.e., $S = \id$. This means that $|\phi\rangle$ must be $|\phi^+_d\rangle$. This analysis holds for any state $\ket{\phi}$ that belongs to the support of $\rho_{AB}$. As a result, $\rho_{AB} = |\phi^+_d\ra\la\phi^+_d|$, which implies
\be 
(U_A\otimes U_B)(\rho)(U_A\otimes U_B)^\dagger = (|\phi^+_d\ra \la \phi^+_d|) \otimes \rho_\text{aux},
\ee 
for some junk state $\rho_{\rm aux}$.

%%%%%%%%%%%%%%%%%%%%%%%%%%%%%%%%%%%%%%%%%%%%%%%%%%%%%%%%%%%%%%%%%%%

\bibliographystyle{apsrev4-2}

\begin{thebibliography}{70}%
	\makeatletter
	\providecommand \@ifxundefined [1]{%
		\@ifx{#1\undefined}
	}%
	\providecommand \@ifnum [1]{%
		\ifnum #1\expandafter \@firstoftwo
		\else \expandafter \@secondoftwo
		\fi
	}%
	\providecommand \@ifx [1]{%
		\ifx #1\expandafter \@firstoftwo
		\else \expandafter \@secondoftwo
		\fi
	}%
	\providecommand \natexlab [1]{#1}%
	\providecommand \enquote  [1]{``#1''}%
	\providecommand \bibnamefont  [1]{#1}%
	\providecommand \bibfnamefont [1]{#1}%
	\providecommand \citenamefont [1]{#1}%
	\providecommand \href@noop [0]{\@secondoftwo}%
	\providecommand \href [0]{\begingroup \@sanitize@url \@href}%
	\providecommand \@href[1]{\@@startlink{#1}\@@href}%
	\providecommand \@@href[1]{\endgroup#1\@@endlink}%
	\providecommand \@sanitize@url [0]{\catcode `\\12\catcode `\$12\catcode
		`\&12\catcode `\#12\catcode `\^12\catcode `\_12\catcode `\%12\relax}%
	\providecommand \@@startlink[1]{}%
	\providecommand \@@endlink[0]{}%
	\providecommand \url  [0]{\begingroup\@sanitize@url \@url }%
	\providecommand \@url [1]{\endgroup\@href {#1}{\urlprefix }}%
	\providecommand \urlprefix  [0]{URL }%
	\providecommand \Eprint [0]{\href }%
	\providecommand \doibase [0]{https://doi.org/}%
	\providecommand \selectlanguage [0]{\@gobble}%
	\providecommand \bibinfo  [0]{\@secondoftwo}%
	\providecommand \bibfield  [0]{\@secondoftwo}%
	\providecommand \translation [1]{[#1]}%
	\providecommand \BibitemOpen [0]{}%
	\providecommand \bibitemStop [0]{}%
	\providecommand \bibitemNoStop [0]{.\EOS\space}%
	\providecommand \EOS [0]{\spacefactor3000\relax}%
	\providecommand \BibitemShut  [1]{\csname bibitem#1\endcsname}%
	\let\auto@bib@innerbib\@empty
	%</preamble>
	\bibitem [{\citenamefont {Brunner}\ \emph {et~al.}(2014)\citenamefont
		{Brunner}, \citenamefont {Cavalcanti}, \citenamefont {Pironio}, \citenamefont
		{Scarani},\ and\ \citenamefont {Wehner}}]{brunner2014bell}%
	\BibitemOpen
	\bibfield  {author} {\bibinfo {author} {\bibfnamefont {N.}~\bibnamefont
			{Brunner}}, \bibinfo {author} {\bibfnamefont {D.}~\bibnamefont {Cavalcanti}},
		\bibinfo {author} {\bibfnamefont {S.}~\bibnamefont {Pironio}}, \bibinfo
		{author} {\bibfnamefont {V.}~\bibnamefont {Scarani}},\ and\ \bibinfo {author}
		{\bibfnamefont {S.}~\bibnamefont {Wehner}},\ }\href
	{https://doi.org/10.1103/RevModPhys.86.419} {\bibfield  {journal} {\bibinfo
			{journal} {Rev. Mod. Phys.}\ }\textbf {\bibinfo {volume} {86}},\ \bibinfo
		{pages} {419} (\bibinfo {year} {2014})}\BibitemShut {NoStop}%
	\bibitem [{\citenamefont {Budroni}\ \emph {et~al.}(2022)\citenamefont
		{Budroni}, \citenamefont {Cabello}, \citenamefont {G\"uhne}, \citenamefont
		{Kleinmann},\ and\ \citenamefont {Larsson}}]{budroni2021quantum}%
	\BibitemOpen
	\bibfield  {author} {\bibinfo {author} {\bibfnamefont {C.}~\bibnamefont
			{Budroni}}, \bibinfo {author} {\bibfnamefont {A.}~\bibnamefont {Cabello}},
		\bibinfo {author} {\bibfnamefont {O.}~\bibnamefont {G\"uhne}}, \bibinfo
		{author} {\bibfnamefont {M.}~\bibnamefont {Kleinmann}},\ and\ \bibinfo
		{author} {\bibfnamefont {J.-{\AA}.}\ \bibnamefont {Larsson}},\ }\href
	{https://doi.org/10.1103/RevModPhys.94.045007} {\bibfield  {journal}
		{\bibinfo  {journal} {Rev. Mod. Phys.}\ }\textbf {\bibinfo {volume} {94}},\
		\bibinfo {pages} {045007} (\bibinfo {year} {2022})}\BibitemShut {NoStop}%
	\bibitem [{\citenamefont {Mayers}\ and\ \citenamefont {Yao}(2004)}]{Yao_self}%
	\BibitemOpen
	\bibfield  {author} {\bibinfo {author} {\bibfnamefont {D.}~\bibnamefont
			{Mayers}}\ and\ \bibinfo {author} {\bibfnamefont {A.}~\bibnamefont {Yao}},\
	}\href {http://dl.acm.org/citation.cfm?id=2011827.2011830} {\bibfield
		{journal} {\bibinfo  {journal} {Quantum Inf. Comput.}\ }\textbf {\bibinfo
			{volume} {4}},\ \bibinfo {pages} {273} (\bibinfo {year} {2004})}\BibitemShut
	{NoStop}%
	\bibitem [{\citenamefont {{\v{S}}upi{\'c}}\ and\ \citenamefont
		{Bowles}(2020)}]{vsupic2019self}%
	\BibitemOpen
	\bibfield  {author} {\bibinfo {author} {\bibfnamefont {I.}~\bibnamefont
			{{\v{S}}upi{\'c}}}\ and\ \bibinfo {author} {\bibfnamefont {J.}~\bibnamefont
			{Bowles}},\ }\href {https://doi.org/10.22331/q-2020-09-30-337} {\bibfield
		{journal} {\bibinfo  {journal} {Quantum}\ }\textbf {\bibinfo {volume} {4}},\
		\bibinfo {pages} {337} (\bibinfo {year} {2020})}\BibitemShut {NoStop}%
	\bibitem [{\citenamefont {Kaniewski}(2020)}]{kaniewski2020weak}%
	\BibitemOpen
	\bibfield  {author} {\bibinfo {author} {\bibfnamefont {J.}~\bibnamefont
			{Kaniewski}},\ }\href {https://doi.org/10.1103/PhysRevResearch.2.033420}
	{\bibfield  {journal} {\bibinfo  {journal} {Phys. Rev. Res.}\ }\textbf
		{\bibinfo {volume} {2}},\ \bibinfo {pages} {033420} (\bibinfo {year}
		{2020})}\BibitemShut {NoStop}%
	\bibitem [{\citenamefont {Bharti}\ \emph {et~al.}(2019)\citenamefont {Bharti},
		\citenamefont {Ray}, \citenamefont {Varvitsiotis}, \citenamefont {Warsi},
		\citenamefont {Cabello},\ and\ \citenamefont {Kwek}}]{BRVWCK19}%
	\BibitemOpen
	\bibfield  {author} {\bibinfo {author} {\bibfnamefont {K.}~\bibnamefont
			{Bharti}}, \bibinfo {author} {\bibfnamefont {M.}~\bibnamefont {Ray}},
		\bibinfo {author} {\bibfnamefont {A.}~\bibnamefont {Varvitsiotis}}, \bibinfo
		{author} {\bibfnamefont {N.~A.}\ \bibnamefont {Warsi}}, \bibinfo {author}
		{\bibfnamefont {A.}~\bibnamefont {Cabello}},\ and\ \bibinfo {author}
		{\bibfnamefont {L.-C.}\ \bibnamefont {Kwek}},\ }\href
	{https://doi.org/10.1103/PhysRevLett.122.250403} {\bibfield  {journal}
		{\bibinfo  {journal} {Phys. Rev. Lett.}\ }\textbf {\bibinfo {volume} {122}},\
		\bibinfo {pages} {250403} (\bibinfo {year} {2019})}\BibitemShut {NoStop}%
	\bibitem [{\citenamefont {Saha}\ \emph {et~al.}(2020)\citenamefont {Saha},
		\citenamefont {Santos},\ and\ \citenamefont
		{Augusiak}}]{Saha2020sumofsquares}%
	\BibitemOpen
	\bibfield  {author} {\bibinfo {author} {\bibfnamefont {D.}~\bibnamefont
			{Saha}}, \bibinfo {author} {\bibfnamefont {R.}~\bibnamefont {Santos}},\ and\
		\bibinfo {author} {\bibfnamefont {R.}~\bibnamefont {Augusiak}},\ }\href
	{https://doi.org/10.22331/q-2020-08-03-302} {\bibfield  {journal} {\bibinfo
			{journal} {{Quantum}}\ }\textbf {\bibinfo {volume} {4}},\ \bibinfo {pages}
		{302} (\bibinfo {year} {2020})}\BibitemShut {NoStop}%
	\bibitem [{\citenamefont {Bharti}\ \emph {et~al.}(2022)\citenamefont {Bharti},
		\citenamefont {Ray}, \citenamefont {Xu}, \citenamefont {Hayashi},
		\citenamefont {Kwek},\ and\ \citenamefont {Cabello}}]{bharti2021graph}%
	\BibitemOpen
	\bibfield  {author} {\bibinfo {author} {\bibfnamefont {K.}~\bibnamefont
			{Bharti}}, \bibinfo {author} {\bibfnamefont {M.}~\bibnamefont {Ray}},
		\bibinfo {author} {\bibfnamefont {Z.-P.}\ \bibnamefont {Xu}}, \bibinfo
		{author} {\bibfnamefont {M.}~\bibnamefont {Hayashi}}, \bibinfo {author}
		{\bibfnamefont {L.-C.}\ \bibnamefont {Kwek}},\ and\ \bibinfo {author}
		{\bibfnamefont {A.}~\bibnamefont {Cabello}},\ }\href
	{https://doi.org/10.1103/PRXQuantum.3.030344} {\bibfield  {journal} {\bibinfo
			{journal} {PRX Quantum}\ }\textbf {\bibinfo {volume} {3}},\ \bibinfo {pages}
		{030344} (\bibinfo {year} {2022})}\BibitemShut {NoStop}%
	\bibitem [{\citenamefont {Cabello}(2008)}]{Cabello08}%
	\BibitemOpen
	\bibfield  {author} {\bibinfo {author} {\bibfnamefont {A.}~\bibnamefont
			{Cabello}},\ }\href {https://doi.org/10.1103/PhysRevLett.101.210401}
	{\bibfield  {journal} {\bibinfo  {journal} {Phys. Rev. Lett.}\ }\textbf
		{\bibinfo {volume} {101}},\ \bibinfo {pages} {210401} (\bibinfo {year}
		{2008})}\BibitemShut {NoStop}%
	\bibitem [{\citenamefont {Badzia\c{g}}\ \emph {et~al.}(2009)\citenamefont
		{Badzia\c{g}}, \citenamefont {Bengtsson}, \citenamefont {Cabello},\ and\
		\citenamefont {Pitowsky}}]{Badziag:2009PRL}%
	\BibitemOpen
	\bibfield  {author} {\bibinfo {author} {\bibfnamefont {P.}~\bibnamefont
			{Badzia\c{g}}}, \bibinfo {author} {\bibfnamefont {I.}~\bibnamefont
			{Bengtsson}}, \bibinfo {author} {\bibfnamefont {A.}~\bibnamefont {Cabello}},\
		and\ \bibinfo {author} {\bibfnamefont {I.}~\bibnamefont {Pitowsky}},\ }\href
	{https://doi.org/10.1103/PhysRevLett.103.050401} {\bibfield  {journal}
		{\bibinfo  {journal} {Phys. Rev. Lett.}\ }\textbf {\bibinfo {volume} {103}},\
		\bibinfo {pages} {050401} (\bibinfo {year} {2009})}\BibitemShut {NoStop}%
	\bibitem [{\citenamefont {Yu}\ and\ \citenamefont {Oh}(2012)}]{Yu:2012PRL}%
	\BibitemOpen
	\bibfield  {author} {\bibinfo {author} {\bibfnamefont {S.}~\bibnamefont
			{Yu}}\ and\ \bibinfo {author} {\bibfnamefont {C.~H.}\ \bibnamefont {Oh}},\
	}\href {https://doi.org/10.1103/PhysRevLett.108.030402} {\bibfield  {journal}
		{\bibinfo  {journal} {Phys. Rev. Lett.}\ }\textbf {\bibinfo {volume} {108}},\
		\bibinfo {pages} {030402} (\bibinfo {year} {2012})}\BibitemShut {NoStop}%
	\bibitem [{\citenamefont {Kleinmann}\ \emph {et~al.}(2012)\citenamefont
		{Kleinmann}, \citenamefont {Budroni}, \citenamefont {Larsson}, \citenamefont
		{G{\"u}hne},\ and\ \citenamefont {Cabello}}]{Kleinmann:2012PRL}%
	\BibitemOpen
	\bibfield  {author} {\bibinfo {author} {\bibfnamefont {M.}~\bibnamefont
			{Kleinmann}}, \bibinfo {author} {\bibfnamefont {C.}~\bibnamefont {Budroni}},
		\bibinfo {author} {\bibfnamefont {J.-{\AA}.}\ \bibnamefont {Larsson}},
		\bibinfo {author} {\bibfnamefont {O.}~\bibnamefont {G{\"u}hne}},\ and\
		\bibinfo {author} {\bibfnamefont {A.}~\bibnamefont {Cabello}},\ }\href
	{https://doi.org/10.1103/PhysRevLett.109.250402} {\bibfield  {journal}
		{\bibinfo  {journal} {Phys. Rev. Lett.}\ }\textbf {\bibinfo {volume} {109}},\
		\bibinfo {pages} {250402} (\bibinfo {year} {2012})}\BibitemShut {NoStop}%
	\bibitem [{SM()}]{SM}%
	\BibitemOpen
	\href@noop {} {\bibinfo  {journal} {See Supplemental Material for
			definitions (Appendix~A) and tools used for proving Results~1, 2, and 3
			(Appendices~B, C, and D, respectively)}\ }\BibitemShut {NoStop}%
	\bibitem [{\citenamefont {Cabello}\ \emph {et~al.}(2015)\citenamefont
		{Cabello}, \citenamefont {Kleinmann},\ and\ \citenamefont
		{Budroni}}]{CKB2015}%
	\BibitemOpen
	\bibfield  {journal} {  }\bibfield  {author} {\bibinfo {author} {\bibfnamefont
			{A.}~\bibnamefont {Cabello}}, \bibinfo {author} {\bibfnamefont
			{M.}~\bibnamefont {Kleinmann}},\ and\ \bibinfo {author} {\bibfnamefont
			{C.}~\bibnamefont {Budroni}},\ }\href
	{https://doi.org/10.1103/PhysRevLett.114.250402} {\bibfield  {journal}
		{\bibinfo  {journal} {Phys. Rev. Lett.}\ }\textbf {\bibinfo {volume} {114}},\
		\bibinfo {pages} {250402} (\bibinfo {year} {2015})}\BibitemShut {NoStop}%
	\bibitem [{\citenamefont {Cabello}(2001)}]{Cabello:2001PRLb}%
	\BibitemOpen
	\bibfield  {author} {\bibinfo {author} {\bibfnamefont {A.}~\bibnamefont
			{Cabello}},\ }\href {https://doi.org/10.1103/PhysRevLett.87.010403}
	{\bibfield  {journal} {\bibinfo  {journal} {Phys. Rev. Lett.}\ }\textbf
		{\bibinfo {volume} {87}},\ \bibinfo {pages} {010403} (\bibinfo {year}
		{2001})}\BibitemShut {NoStop}%
	\bibitem [{\citenamefont {Cleve}\ \emph {et~al.}(2004)\citenamefont {Cleve},
		\citenamefont {H\o{}yer}, \citenamefont {Toner},\ and\ \citenamefont
		{Watrous}}]{cleve2004}%
	\BibitemOpen
	\bibfield  {author} {\bibinfo {author} {\bibfnamefont {R.}~\bibnamefont
			{Cleve}}, \bibinfo {author} {\bibfnamefont {P.}~\bibnamefont {H\o{}yer}},
		\bibinfo {author} {\bibfnamefont {B.}~\bibnamefont {Toner}},\ and\ \bibinfo
		{author} {\bibfnamefont {J.}~\bibnamefont {Watrous}},\ }in\ \href@noop {}
	{\emph {\bibinfo {booktitle} {Proceedings of the 19th IEEE Conference on
				Computational Complexity}}}\ (\bibinfo {organization} {IEEE},\ \bibinfo
	{address} {New York},\ \bibinfo {year} {2004})\ pp.\ \bibinfo {pages}
	{236--249}\BibitemShut {NoStop}%
	\bibitem [{\citenamefont {Brassard}\ \emph {et~al.}(2005)\citenamefont
		{Brassard}, \citenamefont {Broadbent},\ and\ \citenamefont
		{Tapp}}]{brassard2005quantum}%
	\BibitemOpen
	\bibfield  {author} {\bibinfo {author} {\bibfnamefont {G.}~\bibnamefont
			{Brassard}}, \bibinfo {author} {\bibfnamefont {A.}~\bibnamefont
			{Broadbent}},\ and\ \bibinfo {author} {\bibfnamefont {A.}~\bibnamefont
			{Tapp}},\ }\href {https://doi.org/10.1007/s10701-005-7353-4} {\bibfield
		{journal} {\bibinfo  {journal} {Found. Phys.}\ }\textbf {\bibinfo {volume}
			{35}},\ \bibinfo {pages} {1877} (\bibinfo {year} {2005})}\BibitemShut
	{NoStop}%
	\bibitem [{\citenamefont {Cubitt}\ \emph {et~al.}(2010)\citenamefont {Cubitt},
		\citenamefont {Leung}, \citenamefont {Matthews},\ and\ \citenamefont
		{Winter}}]{Cubitt:2010PRL}%
	\BibitemOpen
	\bibfield  {author} {\bibinfo {author} {\bibfnamefont {T.~S.}\ \bibnamefont
			{Cubitt}}, \bibinfo {author} {\bibfnamefont {D.}~\bibnamefont {Leung}},
		\bibinfo {author} {\bibfnamefont {W.}~\bibnamefont {Matthews}},\ and\
		\bibinfo {author} {\bibfnamefont {A.}~\bibnamefont {Winter}},\ }\href
	{https://doi.org/10.1103/PhysRevLett.104.230503} {\bibfield  {journal}
		{\bibinfo  {journal} {Phys. Rev. Lett.}\ }\textbf {\bibinfo {volume} {104}},\
		\bibinfo {pages} {230503} (\bibinfo {year} {2010})}\BibitemShut {NoStop}%
	\bibitem [{\citenamefont {Horodecki}\ \emph {et~al.}()\citenamefont
		{Horodecki}, \citenamefont {Horodecki}, \citenamefont {Horodecki},
		\citenamefont {Horodecki}, \citenamefont {Pawlowski},\ and\ \citenamefont
		{Bourennane}}]{horodecki2010contextuality}%
	\BibitemOpen
	\bibfield  {author} {\bibinfo {author} {\bibfnamefont {K.}~\bibnamefont
			{Horodecki}}, \bibinfo {author} {\bibfnamefont {M.}~\bibnamefont
			{Horodecki}}, \bibinfo {author} {\bibfnamefont {P.}~\bibnamefont
			{Horodecki}}, \bibinfo {author} {\bibfnamefont {R.}~\bibnamefont
			{Horodecki}}, \bibinfo {author} {\bibfnamefont {M.}~\bibnamefont
			{Pawlowski}},\ and\ \bibinfo {author} {\bibfnamefont {M.}~\bibnamefont
			{Bourennane}},\ }\Eprint {https://arxiv.org/abs/1006.0468} {arXiv:1006.0468}
	\BibitemShut {NoStop}%
	\bibitem [{\citenamefont {Cabello}\ \emph {et~al.}(2011)\citenamefont
		{Cabello}, \citenamefont {D'Ambrosio}, \citenamefont {Nagali},\ and\
		\citenamefont {Sciarrino}}]{Cabello_QKD}%
	\BibitemOpen
	\bibfield  {author} {\bibinfo {author} {\bibfnamefont {A.}~\bibnamefont
			{Cabello}}, \bibinfo {author} {\bibfnamefont {V.}~\bibnamefont {D'Ambrosio}},
		\bibinfo {author} {\bibfnamefont {E.}~\bibnamefont {Nagali}},\ and\ \bibinfo
		{author} {\bibfnamefont {F.}~\bibnamefont {Sciarrino}},\ }\href
	{https://doi.org/10.1103/PhysRevA.84.030302} {\bibfield  {journal} {\bibinfo
			{journal} {Phys. Rev. A}\ }\textbf {\bibinfo {volume} {84}},\ \bibinfo
		{pages} {030302} (\bibinfo {year} {2011})}\BibitemShut {NoStop}%
	\bibitem [{\citenamefont {Aolita}\ \emph {et~al.}(2012)\citenamefont {Aolita},
		\citenamefont {Gallego}, \citenamefont {Ac\'{\i}n}, \citenamefont {Chiuri},
		\citenamefont {Vallone}, \citenamefont {Mataloni},\ and\ \citenamefont
		{Cabello}}]{aolita2012pra}%
	\BibitemOpen
	\bibfield  {author} {\bibinfo {author} {\bibfnamefont {L.}~\bibnamefont
			{Aolita}}, \bibinfo {author} {\bibfnamefont {R.}~\bibnamefont {Gallego}},
		\bibinfo {author} {\bibfnamefont {A.}~\bibnamefont {Ac\'{\i}n}}, \bibinfo
		{author} {\bibfnamefont {A.}~\bibnamefont {Chiuri}}, \bibinfo {author}
		{\bibfnamefont {G.}~\bibnamefont {Vallone}}, \bibinfo {author} {\bibfnamefont
			{P.}~\bibnamefont {Mataloni}},\ and\ \bibinfo {author} {\bibfnamefont
			{A.}~\bibnamefont {Cabello}},\ }\href
	{https://doi.org/10.1103/PhysRevA.85.032107} {\bibfield  {journal} {\bibinfo
			{journal} {Phys. Rev. A}\ }\textbf {\bibinfo {volume} {85}},\ \bibinfo
		{pages} {032107} (\bibinfo {year} {2012})}\BibitemShut {NoStop}%
	\bibitem [{\citenamefont {G{\"u}hne}\ \emph {et~al.}(2014)\citenamefont
		{G{\"u}hne}, \citenamefont {Budroni}, \citenamefont {Cabello}, \citenamefont
		{Kleinmann},\ and\ \citenamefont {Larsson}}]{Guhne:2013PRA}%
	\BibitemOpen
	\bibfield  {author} {\bibinfo {author} {\bibfnamefont {O.}~\bibnamefont
			{G{\"u}hne}}, \bibinfo {author} {\bibfnamefont {C.}~\bibnamefont {Budroni}},
		\bibinfo {author} {\bibfnamefont {A.}~\bibnamefont {Cabello}}, \bibinfo
		{author} {\bibfnamefont {M.}~\bibnamefont {Kleinmann}},\ and\ \bibinfo
		{author} {\bibfnamefont {J.-{\AA}.}\ \bibnamefont {Larsson}},\ }\href
	{https://doi.org/10.1103/PhysRevA.89.062107} {\bibfield  {journal} {\bibinfo
			{journal} {Phys. Rev. A}\ }\textbf {\bibinfo {volume} {89}},\ \bibinfo
		{pages} {062107} (\bibinfo {year} {2014})}\BibitemShut {NoStop}%
	\bibitem [{\citenamefont {Ca\~nas}\ \emph {et~al.}(2014)\citenamefont
		{Ca\~nas}, \citenamefont {Arias}, \citenamefont {Etcheverry}, \citenamefont
		{G\'omez}, \citenamefont {Cabello}, \citenamefont {Xavier},\ and\
		\citenamefont {Lima}}]{Canas:PRL2014}%
	\BibitemOpen
	\bibfield  {author} {\bibinfo {author} {\bibfnamefont {G.}~\bibnamefont
			{Ca\~nas}}, \bibinfo {author} {\bibfnamefont {M.}~\bibnamefont {Arias}},
		\bibinfo {author} {\bibfnamefont {S.}~\bibnamefont {Etcheverry}}, \bibinfo
		{author} {\bibfnamefont {E.~S.}\ \bibnamefont {G\'omez}}, \bibinfo {author}
		{\bibfnamefont {A.}~\bibnamefont {Cabello}}, \bibinfo {author} {\bibfnamefont
			{G.~B.}\ \bibnamefont {Xavier}},\ and\ \bibinfo {author} {\bibfnamefont
			{G.}~\bibnamefont {Lima}},\ }\href
	{https://doi.org/10.1103/PhysRevLett.113.090404} {\bibfield  {journal}
		{\bibinfo  {journal} {Phys. Rev. Lett.}\ }\textbf {\bibinfo {volume} {113}},\
		\bibinfo {pages} {090404} (\bibinfo {year} {2014})}\BibitemShut {NoStop}%
	\bibitem [{\citenamefont {Cabello}\ \emph {et~al.}(2018)\citenamefont
		{Cabello}, \citenamefont {Gu}, \citenamefont {G{\"u}hne},\ and\ \citenamefont
		{Xu}}]{Cabello:2018PRLEPSILON}%
	\BibitemOpen
	\bibfield  {author} {\bibinfo {author} {\bibfnamefont {A.}~\bibnamefont
			{Cabello}}, \bibinfo {author} {\bibfnamefont {M.}~\bibnamefont {Gu}},
		\bibinfo {author} {\bibfnamefont {O.}~\bibnamefont {G{\"u}hne}},\ and\
		\bibinfo {author} {\bibfnamefont {Z.-P.}\ \bibnamefont {Xu}},\ }\href
	{https://doi.org/10.1103/PhysRevLett.120.130401} {\bibfield  {journal}
		{\bibinfo  {journal} {Phys. Rev. Lett.}\ }\textbf {\bibinfo {volume} {120}},\
		\bibinfo {pages} {130401} (\bibinfo {year} {2018})}\BibitemShut {NoStop}%
	\bibitem [{\citenamefont {Saha}\ \emph {et~al.}(2019)\citenamefont {Saha},
		\citenamefont {Horodecki},\ and\ \citenamefont
		{Paw{\l}owski}}]{saha2019State}%
	\BibitemOpen
	\bibfield  {author} {\bibinfo {author} {\bibfnamefont {D.}~\bibnamefont
			{Saha}}, \bibinfo {author} {\bibfnamefont {P.}~\bibnamefont {Horodecki}},\
		and\ \bibinfo {author} {\bibfnamefont {M.}~\bibnamefont {Paw{\l}owski}},\
	}\href {https://doi.org/10.1088/1367-2630/ab4149} {\bibfield  {journal}
		{\bibinfo  {journal} {New J. Phys.}\ }\textbf {\bibinfo {volume} {21}},\
		\bibinfo {pages} {093057} (\bibinfo {year} {2019})}\BibitemShut {NoStop}%
	\bibitem [{\citenamefont {Ji}\ \emph {et~al.}(2021)\citenamefont {Ji},
		\citenamefont {Natarajan}, \citenamefont {Vidick}, \citenamefont {Wright},\
		and\ \citenamefont {Yuen}}]{Ji:2021CACM}%
	\BibitemOpen
	\bibfield  {author} {\bibinfo {author} {\bibfnamefont {Z.}~\bibnamefont
			{Ji}}, \bibinfo {author} {\bibfnamefont {A.}~\bibnamefont {Natarajan}},
		\bibinfo {author} {\bibfnamefont {T.}~\bibnamefont {Vidick}}, \bibinfo
		{author} {\bibfnamefont {J.}~\bibnamefont {Wright}},\ and\ \bibinfo {author}
		{\bibfnamefont {H.}~\bibnamefont {Yuen}},\ }\href
	{https://doi.org/10.1145/3485628} {\bibfield  {journal} {\bibinfo  {journal}
			{Commun. ACM}\ }\textbf {\bibinfo {volume} {64}},\ \bibinfo {pages} {131}
		(\bibinfo {year} {2021})}\BibitemShut {NoStop}%
	\bibitem [{\citenamefont {Gupta}\ \emph {et~al.}(2023)\citenamefont {Gupta},
		\citenamefont {Saha}, \citenamefont {Xu}, \citenamefont {Cabello},\ and\
		\citenamefont {Majumdar}}]{Gupta2023PRL}%
	\BibitemOpen
	\bibfield  {author} {\bibinfo {author} {\bibfnamefont {S.}~\bibnamefont
			{Gupta}}, \bibinfo {author} {\bibfnamefont {D.}~\bibnamefont {Saha}},
		\bibinfo {author} {\bibfnamefont {Z.-P.}\ \bibnamefont {Xu}}, \bibinfo
		{author} {\bibfnamefont {A.}~\bibnamefont {Cabello}},\ and\ \bibinfo {author}
		{\bibfnamefont {A.~S.}\ \bibnamefont {Majumdar}},\ }\href
	{https://doi.org/10.1103/PhysRevLett.130.080802} {\bibfield  {journal}
		{\bibinfo  {journal} {Phys. Rev. Lett.}\ }\textbf {\bibinfo {volume} {130}},\
		\bibinfo {pages} {080802} (\bibinfo {year} {2023})}\BibitemShut {NoStop}%
	\bibitem [{\citenamefont {Zhen}\ \emph {et~al.}(2023)\citenamefont {Zhen},
		\citenamefont {Mao}, \citenamefont {Zhang}, \citenamefont {Xu},\ and\
		\citenamefont {Sanders}}]{Zhen2023PRL}%
	\BibitemOpen
	\bibfield  {author} {\bibinfo {author} {\bibfnamefont {Y.-Z.}\ \bibnamefont
			{Zhen}}, \bibinfo {author} {\bibfnamefont {Y.}~\bibnamefont {Mao}}, \bibinfo
		{author} {\bibfnamefont {Y.-Z.}\ \bibnamefont {Zhang}}, \bibinfo {author}
		{\bibfnamefont {F.}~\bibnamefont {Xu}},\ and\ \bibinfo {author}
		{\bibfnamefont {B.~C.}\ \bibnamefont {Sanders}},\ }\href
	{https://doi.org/10.1103/PhysRevLett.131.080801} {\bibfield  {journal}
		{\bibinfo  {journal} {Phys. Rev. Lett.}\ }\textbf {\bibinfo {volume} {131}},\
		\bibinfo {pages} {080801} (\bibinfo {year} {2023})}\BibitemShut {NoStop}%
	\bibitem [{\citenamefont {Kirchmair}\ \emph {et~al.}(2009)\citenamefont
		{Kirchmair}, \citenamefont {Z{\"a}hringer}, \citenamefont {Gerritsma},
		\citenamefont {Kleinmann}, \citenamefont {G{\"u}hne}, \citenamefont
		{Cabello}, \citenamefont {Blatt},\ and\ \citenamefont
		{Roos}}]{kirchmair2009Stateindependent}%
	\BibitemOpen
	\bibfield  {author} {\bibinfo {author} {\bibfnamefont {G.}~\bibnamefont
			{Kirchmair}}, \bibinfo {author} {\bibfnamefont {F.}~\bibnamefont
			{Z{\"a}hringer}}, \bibinfo {author} {\bibfnamefont {R.}~\bibnamefont
			{Gerritsma}}, \bibinfo {author} {\bibfnamefont {M.}~\bibnamefont
			{Kleinmann}}, \bibinfo {author} {\bibfnamefont {O.}~\bibnamefont
			{G{\"u}hne}}, \bibinfo {author} {\bibfnamefont {A.}~\bibnamefont {Cabello}},
		\bibinfo {author} {\bibfnamefont {R.}~\bibnamefont {Blatt}},\ and\ \bibinfo
		{author} {\bibfnamefont {C.~F.}\ \bibnamefont {Roos}},\ }\href
	{https://doi.org/10.1038/nature08172} {\bibfield  {journal} {\bibinfo
			{journal} {Nature (London)}\ }\textbf {\bibinfo {volume} {460}},\ \bibinfo
		{pages} {494} (\bibinfo {year} {2009})}\BibitemShut {NoStop}%
	\bibitem [{\citenamefont {Zhang}\ \emph {et~al.}(2013)\citenamefont {Zhang},
		\citenamefont {Um}, \citenamefont {Zhang}, \citenamefont {An}, \citenamefont
		{Wang}, \citenamefont {Deng}, \citenamefont {Shen}, \citenamefont {Duan},\
		and\ \citenamefont {Kim}}]{zhang2013state}%
	\BibitemOpen
	\bibfield  {author} {\bibinfo {author} {\bibfnamefont {X.}~\bibnamefont
			{Zhang}}, \bibinfo {author} {\bibfnamefont {M.}~\bibnamefont {Um}}, \bibinfo
		{author} {\bibfnamefont {J.}~\bibnamefont {Zhang}}, \bibinfo {author}
		{\bibfnamefont {S.}~\bibnamefont {An}}, \bibinfo {author} {\bibfnamefont
			{Y.}~\bibnamefont {Wang}}, \bibinfo {author} {\bibfnamefont {D.-l.}\
			\bibnamefont {Deng}}, \bibinfo {author} {\bibfnamefont {C.}~\bibnamefont
			{Shen}}, \bibinfo {author} {\bibfnamefont {L.-M.}\ \bibnamefont {Duan}},\
		and\ \bibinfo {author} {\bibfnamefont {K.}~\bibnamefont {Kim}},\ }\href
	{https://doi.org/10.1103/PhysRevLett.110.070401} {\bibfield  {journal}
		{\bibinfo  {journal} {Phys. Rev. Lett.}\ }\textbf {\bibinfo {volume} {110}},\
		\bibinfo {pages} {070401} (\bibinfo {year} {2013})}\BibitemShut {NoStop}%
	\bibitem [{\citenamefont {Huang}\ \emph {et~al.}(2003)\citenamefont {Huang},
		\citenamefont {Li}, \citenamefont {Zhang}, \citenamefont {Pan},\ and\
		\citenamefont {Guo}}]{huang2003experimental}%
	\BibitemOpen
	\bibfield  {author} {\bibinfo {author} {\bibfnamefont {Y.-F.}\ \bibnamefont
			{Huang}}, \bibinfo {author} {\bibfnamefont {C.-F.}\ \bibnamefont {Li}},
		\bibinfo {author} {\bibfnamefont {Y.-S.}\ \bibnamefont {Zhang}}, \bibinfo
		{author} {\bibfnamefont {J.-W.}\ \bibnamefont {Pan}},\ and\ \bibinfo {author}
		{\bibfnamefont {G.-C.}\ \bibnamefont {Guo}},\ }\href
	{https://doi.org/10.1103/PhysRevLett.90.250401} {\bibfield  {journal}
		{\bibinfo  {journal} {Phys. Rev. Lett.}\ }\textbf {\bibinfo {volume} {90}},\
		\bibinfo {pages} {250401} (\bibinfo {year} {2003})}\BibitemShut {NoStop}%
	\bibitem [{\citenamefont {Amselem}\ \emph {et~al.}(2009)\citenamefont
		{Amselem}, \citenamefont {R\aa{}dmark}, \citenamefont {Bourennane},\ and\
		\citenamefont {Cabello}}]{amselem2009state}%
	\BibitemOpen
	\bibfield  {author} {\bibinfo {author} {\bibfnamefont {E.}~\bibnamefont
			{Amselem}}, \bibinfo {author} {\bibfnamefont {M.}~\bibnamefont
			{R\aa{}dmark}}, \bibinfo {author} {\bibfnamefont {M.}~\bibnamefont
			{Bourennane}},\ and\ \bibinfo {author} {\bibfnamefont {A.}~\bibnamefont
			{Cabello}},\ }\href {https://doi.org/10.1103/PhysRevLett.103.160405}
	{\bibfield  {journal} {\bibinfo  {journal} {Phys. Rev. Lett.}\ }\textbf
		{\bibinfo {volume} {103}},\ \bibinfo {pages} {160405} (\bibinfo {year}
		{2009})}\BibitemShut {NoStop}%
	\bibitem [{\citenamefont {D'Ambrosio}\ \emph {et~al.}(2013)\citenamefont
		{D'Ambrosio}, \citenamefont {Herbauts}, \citenamefont {Amselem},
		\citenamefont {Nagali}, \citenamefont {Bourennane}, \citenamefont
		{Sciarrino},\ and\ \citenamefont {Cabello}}]{ambrosio2013experimental}%
	\BibitemOpen
	\bibfield  {author} {\bibinfo {author} {\bibfnamefont {V.}~\bibnamefont
			{D'Ambrosio}}, \bibinfo {author} {\bibfnamefont {I.}~\bibnamefont
			{Herbauts}}, \bibinfo {author} {\bibfnamefont {E.}~\bibnamefont {Amselem}},
		\bibinfo {author} {\bibfnamefont {E.}~\bibnamefont {Nagali}}, \bibinfo
		{author} {\bibfnamefont {M.}~\bibnamefont {Bourennane}}, \bibinfo {author}
		{\bibfnamefont {F.}~\bibnamefont {Sciarrino}},\ and\ \bibinfo {author}
		{\bibfnamefont {A.}~\bibnamefont {Cabello}},\ }\href
	{https://doi.org/10.1103/PhysRevX.3.011012} {\bibfield  {journal} {\bibinfo
			{journal} {Phys. Rev. X}\ }\textbf {\bibinfo {volume} {3}},\ \bibinfo {pages}
		{011012} (\bibinfo {year} {2013})}\BibitemShut {NoStop}%
	\bibitem [{\citenamefont {Leupold}\ \emph {et~al.}(2018)\citenamefont
		{Leupold}, \citenamefont {Malinowski}, \citenamefont {Zhang}, \citenamefont
		{Negnevitsky}, \citenamefont {Cabello}, \citenamefont {Alonso},\ and\
		\citenamefont {Home}}]{leupold2018sustained}%
	\BibitemOpen
	\bibfield  {author} {\bibinfo {author} {\bibfnamefont {F.~M.}\ \bibnamefont
			{Leupold}}, \bibinfo {author} {\bibfnamefont {M.}~\bibnamefont {Malinowski}},
		\bibinfo {author} {\bibfnamefont {C.}~\bibnamefont {Zhang}}, \bibinfo
		{author} {\bibfnamefont {V.}~\bibnamefont {Negnevitsky}}, \bibinfo {author}
		{\bibfnamefont {A.}~\bibnamefont {Cabello}}, \bibinfo {author} {\bibfnamefont
			{J.}~\bibnamefont {Alonso}},\ and\ \bibinfo {author} {\bibfnamefont {J.~P.}\
			\bibnamefont {Home}},\ }\href
	{https://doi.org/10.1103/PhysRevLett.120.180401} {\bibfield  {journal}
		{\bibinfo  {journal} {Phys. Rev. Lett.}\ }\textbf {\bibinfo {volume} {120}},\
		\bibinfo {pages} {180401} (\bibinfo {year} {2018})}\BibitemShut {NoStop}%
	\bibitem [{\citenamefont {Tavakoli}\ \emph {et~al.}(2020)\citenamefont
		{Tavakoli}, \citenamefont {Smania}, \citenamefont {V{\'e}rtesi},
		\citenamefont {Brunner},\ and\ \citenamefont
		{Bourennane}}]{tavakoli2020self}%
	\BibitemOpen
	\bibfield  {author} {\bibinfo {author} {\bibfnamefont {A.}~\bibnamefont
			{Tavakoli}}, \bibinfo {author} {\bibfnamefont {M.}~\bibnamefont {Smania}},
		\bibinfo {author} {\bibfnamefont {T.}~\bibnamefont {V{\'e}rtesi}}, \bibinfo
		{author} {\bibfnamefont {N.}~\bibnamefont {Brunner}},\ and\ \bibinfo {author}
		{\bibfnamefont {M.}~\bibnamefont {Bourennane}},\ }\href
	{https://doi.org/10.1126/sciadv.aaw6664} {\bibfield  {journal} {\bibinfo
			{journal} {Sci. Adv.}\ }\textbf {\bibinfo {volume} {6}},\ \bibinfo {pages}
		{eaaw6664} (\bibinfo {year} {2020})}\BibitemShut {NoStop}%
	\bibitem [{\citenamefont {Farkas}\ and\ \citenamefont
		{Kaniewski}(2019)}]{farkas2019self}%
	\BibitemOpen
	\bibfield  {author} {\bibinfo {author} {\bibfnamefont {M.}~\bibnamefont
			{Farkas}}\ and\ \bibinfo {author} {\bibfnamefont {J.}~\bibnamefont
			{Kaniewski}},\ }\href {https://doi.org/10.1103/PhysRevA.99.032316} {\bibfield
		{journal} {\bibinfo  {journal} {Phys. Rev. A}\ }\textbf {\bibinfo {volume}
			{99}},\ \bibinfo {pages} {032316} (\bibinfo {year} {2019})}\BibitemShut
	{NoStop}%
	\bibitem [{\citenamefont {{\v{S}}upi{\'c}}\ \emph {et~al.}(2016)\citenamefont
		{{\v{S}}upi{\'c}}, \citenamefont {Augusiak}, \citenamefont {Salavrakos},\
		and\ \citenamefont {Ac{\'\i}n}}]{vsupic2016self}%
	\BibitemOpen
	\bibfield  {author} {\bibinfo {author} {\bibfnamefont {I.}~\bibnamefont
			{{\v{S}}upi{\'c}}}, \bibinfo {author} {\bibfnamefont {R.}~\bibnamefont
			{Augusiak}}, \bibinfo {author} {\bibfnamefont {A.}~\bibnamefont
			{Salavrakos}},\ and\ \bibinfo {author} {\bibfnamefont {A.}~\bibnamefont
			{Ac{\'\i}n}},\ }\href {https://doi.org/10.1088/1367-2630/18/3/035013}
	{\bibfield  {journal} {\bibinfo  {journal} {New J. Phys.}\ }\textbf {\bibinfo
			{volume} {18}},\ \bibinfo {pages} {035013} (\bibinfo {year}
		{2016})}\BibitemShut {NoStop}%
	\bibitem [{\citenamefont {Gheorghiu}\ \emph {et~al.}(2015)\citenamefont
		{Gheorghiu}, \citenamefont {Kashefi},\ and\ \citenamefont
		{Wallden}}]{gheorghiu2015robustness}%
	\BibitemOpen
	\bibfield  {author} {\bibinfo {author} {\bibfnamefont {A.}~\bibnamefont
			{Gheorghiu}}, \bibinfo {author} {\bibfnamefont {E.}~\bibnamefont {Kashefi}},\
		and\ \bibinfo {author} {\bibfnamefont {P.}~\bibnamefont {Wallden}},\ }\href
	{https://doi.org/10.1088/1367-2630/17/8/083040} {\bibfield  {journal}
		{\bibinfo  {journal} {New J. Phys.}\ }\textbf {\bibinfo {volume} {17}},\
		\bibinfo {pages} {083040} (\bibinfo {year} {2015})}\BibitemShut {NoStop}%
	\bibitem [{\citenamefont {Shrotriya}\ \emph {et~al.}(2021)\citenamefont
		{Shrotriya}, \citenamefont {Bharti},\ and\ \citenamefont
		{Kwek}}]{shrotriya2020self}%
	\BibitemOpen
	\bibfield  {author} {\bibinfo {author} {\bibfnamefont {H.}~\bibnamefont
			{Shrotriya}}, \bibinfo {author} {\bibfnamefont {K.}~\bibnamefont {Bharti}},\
		and\ \bibinfo {author} {\bibfnamefont {L.-C.}\ \bibnamefont {Kwek}},\ }\href
	{https://doi.org/10.1103/PhysRevResearch.3.033093} {\bibfield  {journal}
		{\bibinfo  {journal} {Phys. Rev. Res.}\ }\textbf {\bibinfo {volume} {3}},\
		\bibinfo {pages} {033093} (\bibinfo {year} {2021})}\BibitemShut {NoStop}%
	\bibitem [{\citenamefont {Cabello}(2016{\natexlab{a}})}]{cabello2016Simple}%
	\BibitemOpen
	\bibfield  {author} {\bibinfo {author} {\bibfnamefont {A.}~\bibnamefont
			{Cabello}},\ }\href {https://doi.org/10.1103/PhysRevA.93.032102} {\bibfield
		{journal} {\bibinfo  {journal} {Phys. Rev. A}\ }\textbf {\bibinfo {volume}
			{93}},\ \bibinfo {pages} {032102} (\bibinfo {year}
		{2016}{\natexlab{a}})}\BibitemShut {NoStop}%
	\bibitem [{\citenamefont {Bengtsson}\ \emph {et~al.}(2012)\citenamefont
		{Bengtsson}, \citenamefont {Blanchfield},\ and\ \citenamefont
		{Cabello}}]{Bengtsson:2012PLA}%
	\BibitemOpen
	\bibfield  {author} {\bibinfo {author} {\bibfnamefont {I.}~\bibnamefont
			{Bengtsson}}, \bibinfo {author} {\bibfnamefont {K.}~\bibnamefont
			{Blanchfield}},\ and\ \bibinfo {author} {\bibfnamefont {A.}~\bibnamefont
			{Cabello}},\ }\href {https://doi.org/10.1016/j.physleta.2011.12.011}
	{\bibfield  {journal} {\bibinfo  {journal} {Phys. Lett. A}\ }\textbf
		{\bibinfo {volume} {376}},\ \bibinfo {pages} {374} (\bibinfo {year}
		{2012})}\BibitemShut {NoStop}%
	\bibitem [{\citenamefont {Cabello}\ \emph {et~al.}(1996)\citenamefont
		{Cabello}, \citenamefont {Estebaranz},\ and\ \citenamefont
		{Garc{\'\i}a-Alcaine}}]{Cabello:1996PLA}%
	\BibitemOpen
	\bibfield  {author} {\bibinfo {author} {\bibfnamefont {A.}~\bibnamefont
			{Cabello}}, \bibinfo {author} {\bibfnamefont {J.~M.}\ \bibnamefont
			{Estebaranz}},\ and\ \bibinfo {author} {\bibfnamefont {G.}~\bibnamefont
			{Garc{\'\i}a-Alcaine}},\ }\href
	{https://doi.org/10.1016/0375-9601(96)00134-X} {\bibfield  {journal}
		{\bibinfo  {journal} {Phys. Lett. A}\ }\textbf {\bibinfo {volume} {212}},\
		\bibinfo {pages} {183} (\bibinfo {year} {1996})}\BibitemShut {NoStop}%
	\bibitem [{\citenamefont {Peres}(1991)}]{Peres:1991JPA}%
	\BibitemOpen
	\bibfield  {author} {\bibinfo {author} {\bibfnamefont {A.}~\bibnamefont
			{Peres}},\ }\href {https://doi.org/10.1088/0305-4470/24/4/003} {\bibfield
		{journal} {\bibinfo  {journal} {J. Phys. A}\ }\textbf {\bibinfo {volume}
			{24}},\ \bibinfo {pages} {L175} (\bibinfo {year} {1991})}\BibitemShut
	{NoStop}%
	\bibitem [{\citenamefont {Peres}(1990)}]{Peres1990PLA}%
	\BibitemOpen
	\bibfield  {author} {\bibinfo {author} {\bibfnamefont {A.}~\bibnamefont
			{Peres}},\ }\href {https://doi.org/10.1016/0375-9601(90)90172-K} {\bibfield
		{journal} {\bibinfo  {journal} {Phys. Lett. A}\ }\textbf {\bibinfo {volume}
			{151}},\ \bibinfo {pages} {107} (\bibinfo {year} {1990})}\BibitemShut
	{NoStop}%
	\bibitem [{\citenamefont {Mermin}(1990)}]{Mermin1990PRL}%
	\BibitemOpen
	\bibfield  {author} {\bibinfo {author} {\bibfnamefont {N.~D.}\ \bibnamefont
			{Mermin}},\ }\href {https://doi.org/10.1103/PhysRevLett.65.3373} {\bibfield
		{journal} {\bibinfo  {journal} {Phys. Rev. Lett.}\ }\textbf {\bibinfo
			{volume} {65}},\ \bibinfo {pages} {3373} (\bibinfo {year}
		{1990})}\BibitemShut {NoStop}%
	\bibitem [{\citenamefont {Xu}\ \emph {et~al.}(2020)\citenamefont {Xu},
		\citenamefont {Chen},\ and\ \citenamefont {G\"uhne}}]{Xu:2020PRL}%
	\BibitemOpen
	\bibfield  {author} {\bibinfo {author} {\bibfnamefont {Z.-P.}\ \bibnamefont
			{Xu}}, \bibinfo {author} {\bibfnamefont {J.-L.}\ \bibnamefont {Chen}},\ and\
		\bibinfo {author} {\bibfnamefont {O.}~\bibnamefont {G\"uhne}},\ }\href
	{https://doi.org/10.1103/PhysRevLett.124.230401} {\bibfield  {journal}
		{\bibinfo  {journal} {Phys. Rev. Lett.}\ }\textbf {\bibinfo {volume} {124}},\
		\bibinfo {pages} {230401} (\bibinfo {year} {2020})}\BibitemShut {NoStop}%
	\bibitem [{\citenamefont {Cabello}\ \emph {et~al.}(2005)\citenamefont
		{Cabello}, \citenamefont {Estebaranz},\ and\ \citenamefont
		{García-Alcaine}}]{Cabello05}%
	\BibitemOpen
	\bibfield  {author} {\bibinfo {author} {\bibfnamefont {A.}~\bibnamefont
			{Cabello}}, \bibinfo {author} {\bibfnamefont {J.~M.}\ \bibnamefont
			{Estebaranz}},\ and\ \bibinfo {author} {\bibfnamefont {G.}~\bibnamefont
			{García-Alcaine}},\ }\href {https://doi.org/10.1016/j.physleta.2005.03.067}
	{\bibfield  {journal} {\bibinfo  {journal} {Phys. Lett. A}\ }\textbf
		{\bibinfo {volume} {339}},\ \bibinfo {pages} {425} (\bibinfo {year}
		{2005})}\BibitemShut {NoStop}%
	\bibitem [{\citenamefont {Cabello}\ \emph {et~al.}(2016)\citenamefont
		{Cabello}, \citenamefont {Kleinmann},\ and\ \citenamefont
		{Portillo}}]{Cabello:2016JPA}%
	\BibitemOpen
	\bibfield  {author} {\bibinfo {author} {\bibfnamefont {A.}~\bibnamefont
			{Cabello}}, \bibinfo {author} {\bibfnamefont {M.}~\bibnamefont {Kleinmann}},\
		and\ \bibinfo {author} {\bibfnamefont {J.~R.}\ \bibnamefont {Portillo}},\
	}\href {https://doi.org/10.1088/1751-8113/49/38/38LT01} {\bibfield  {journal}
		{\bibinfo  {journal} {J. Phys. A}\ }\textbf {\bibinfo {volume} {49}},\
		\bibinfo {pages} {38LT01} (\bibinfo {year} {2016})}\BibitemShut {NoStop}%
	\bibitem [{\citenamefont {Pavi\v{c}i{\'c}}\ \emph {et~al.}(2005)\citenamefont
		{Pavi\v{c}i{\'c}}, \citenamefont {Merlet}, \citenamefont {McKay},\ and\
		\citenamefont {Megill}}]{Pavicic:2005JPA}%
	\BibitemOpen
	\bibfield  {author} {\bibinfo {author} {\bibfnamefont {M.}~\bibnamefont
			{Pavi\v{c}i{\'c}}}, \bibinfo {author} {\bibfnamefont {J.-P.}\ \bibnamefont
			{Merlet}}, \bibinfo {author} {\bibfnamefont {B.~D.}\ \bibnamefont {McKay}},\
		and\ \bibinfo {author} {\bibfnamefont {N.~D.}\ \bibnamefont {Megill}},\
	}\href {https://doi.org/10.1088/0305-4470/38/7/013} {\bibfield  {journal}
		{\bibinfo  {journal} {J. Phys. A}\ }\textbf {\bibinfo {volume} {38}},\
		\bibinfo {pages} {1577} (\bibinfo {year} {2005})}\BibitemShut {NoStop}%
	\bibitem [{\citenamefont {Cabello}(2010)}]{Cabello:2010PRL}%
	\BibitemOpen
	\bibfield  {author} {\bibinfo {author} {\bibfnamefont {A.}~\bibnamefont
			{Cabello}},\ }\href {https://doi.org/10.1103/PhysRevLett.104.220401}
	{\bibfield  {journal} {\bibinfo  {journal} {Phys. Rev. Lett.}\ }\textbf
		{\bibinfo {volume} {104}},\ \bibinfo {pages} {220401} (\bibinfo {year}
		{2010})}\BibitemShut {NoStop}%
	\bibitem [{\citenamefont {Liu}\ \emph {et~al.}(2016)\citenamefont {Liu},
		\citenamefont {Hu}, \citenamefont {Chen}, \citenamefont {Huang},
		\citenamefont {Han}, \citenamefont {Li}, \citenamefont {Guo},\ and\
		\citenamefont {Cabello}}]{Liu:2016PRL}%
	\BibitemOpen
	\bibfield  {author} {\bibinfo {author} {\bibfnamefont {B.-H.}\ \bibnamefont
			{Liu}}, \bibinfo {author} {\bibfnamefont {X.-M.}\ \bibnamefont {Hu}},
		\bibinfo {author} {\bibfnamefont {J.-S.}\ \bibnamefont {Chen}}, \bibinfo
		{author} {\bibfnamefont {Y.-F.}\ \bibnamefont {Huang}}, \bibinfo {author}
		{\bibfnamefont {Y.-J.}\ \bibnamefont {Han}}, \bibinfo {author} {\bibfnamefont
			{C.-F.}\ \bibnamefont {Li}}, \bibinfo {author} {\bibfnamefont {G.-C.}\
			\bibnamefont {Guo}},\ and\ \bibinfo {author} {\bibfnamefont {A.}~\bibnamefont
			{Cabello}},\ }\href {https://doi.org/10.1103/PhysRevLett.117.220402}
	{\bibfield  {journal} {\bibinfo  {journal} {Phys. Rev. Lett.}\ }\textbf
		{\bibinfo {volume} {117}},\ \bibinfo {pages} {220402} (\bibinfo {year}
		{2016})}\BibitemShut {NoStop}%
	\bibitem [{\citenamefont {Gonzales-Ureta}\ \emph {et~al.}(2023)\citenamefont
		{Gonzales-Ureta}, \citenamefont {Predojevi\ifmmode~\acute{c}\else
			\'{c}\fi{}},\ and\ \citenamefont {Cabello}}]{Gonzales-Ureta2023}%
	\BibitemOpen
	\bibfield  {author} {\bibinfo {author} {\bibfnamefont {J.~R.}\ \bibnamefont
			{Gonzales-Ureta}}, \bibinfo {author} {\bibfnamefont {A.}~\bibnamefont
			{Predojevi\ifmmode~\acute{c}\else \'{c}\fi{}}},\ and\ \bibinfo {author}
		{\bibfnamefont {A.}~\bibnamefont {Cabello}},\ }\href
	{https://doi.org/10.1103/PhysRevResearch.5.L012035} {\bibfield  {journal}
		{\bibinfo  {journal} {Phys. Rev. Res.}\ }\textbf {\bibinfo {volume} {5}},\
		\bibinfo {pages} {L012035} (\bibinfo {year} {2023})}\BibitemShut {NoStop}%
	\bibitem [{\citenamefont {Wang}\ \emph {et~al.}(2022)\citenamefont {Wang},
		\citenamefont {Zhang}, \citenamefont {Luan}, \citenamefont {Um},
		\citenamefont {Wang}, \citenamefont {Qiao}, \citenamefont {Xie},
		\citenamefont {Zhang}, \citenamefont {Cabello},\ and\ \citenamefont
		{Kim}}]{WangSAdv2022}%
	\BibitemOpen
	\bibfield  {author} {\bibinfo {author} {\bibfnamefont {P.}~\bibnamefont
			{Wang}}, \bibinfo {author} {\bibfnamefont {J.}~\bibnamefont {Zhang}},
		\bibinfo {author} {\bibfnamefont {C.-Y.}\ \bibnamefont {Luan}}, \bibinfo
		{author} {\bibfnamefont {M.}~\bibnamefont {Um}}, \bibinfo {author}
		{\bibfnamefont {Y.}~\bibnamefont {Wang}}, \bibinfo {author} {\bibfnamefont
			{M.}~\bibnamefont {Qiao}}, \bibinfo {author} {\bibfnamefont {T.}~\bibnamefont
			{Xie}}, \bibinfo {author} {\bibfnamefont {J.-N.}\ \bibnamefont {Zhang}},
		\bibinfo {author} {\bibfnamefont {A.}~\bibnamefont {Cabello}},\ and\ \bibinfo
		{author} {\bibfnamefont {K.}~\bibnamefont {Kim}},\ }\href
	{https://doi.org/10.1126/sciadv.abk1660} {\bibfield  {journal} {\bibinfo
			{journal} {Sci. Adv.}\ }\textbf {\bibinfo {volume} {8}},\ \bibinfo {pages}
		{eabk1660} (\bibinfo {year} {2022})}\BibitemShut {NoStop}%
	\bibitem [{\citenamefont {Brunner}\ \emph {et~al.}(2013)\citenamefont
		{Brunner}, \citenamefont {Navascu\'es},\ and\ \citenamefont
		{V\'ertesi}}]{PhysRevLett.110.150501}%
	\BibitemOpen
	\bibfield  {author} {\bibinfo {author} {\bibfnamefont {N.}~\bibnamefont
			{Brunner}}, \bibinfo {author} {\bibfnamefont {M.}~\bibnamefont
			{Navascu\'es}},\ and\ \bibinfo {author} {\bibfnamefont {T.}~\bibnamefont
			{V\'ertesi}},\ }\href {https://doi.org/10.1103/PhysRevLett.110.150501}
	{\bibfield  {journal} {\bibinfo  {journal} {Phys. Rev. Lett.}\ }\textbf
		{\bibinfo {volume} {110}},\ \bibinfo {pages} {150501} (\bibinfo {year}
		{2013})}\BibitemShut {NoStop}%
	\bibitem [{\citenamefont {D'Ambrosio}\ \emph {et~al.}(2014)\citenamefont
		{D'Ambrosio}, \citenamefont {Bisesto}, \citenamefont {Sciarrino},
		\citenamefont {Barra}, \citenamefont {Lima},\ and\ \citenamefont
		{Cabello}}]{PhysRevLett.112.140503}%
	\BibitemOpen
	\bibfield  {author} {\bibinfo {author} {\bibfnamefont {V.}~\bibnamefont
			{D'Ambrosio}}, \bibinfo {author} {\bibfnamefont {F.}~\bibnamefont {Bisesto}},
		\bibinfo {author} {\bibfnamefont {F.}~\bibnamefont {Sciarrino}}, \bibinfo
		{author} {\bibfnamefont {J.~F.}\ \bibnamefont {Barra}}, \bibinfo {author}
		{\bibfnamefont {G.}~\bibnamefont {Lima}},\ and\ \bibinfo {author}
		{\bibfnamefont {A.}~\bibnamefont {Cabello}},\ }\href
	{https://doi.org/10.1103/PhysRevLett.112.140503} {\bibfield  {journal}
		{\bibinfo  {journal} {Phys. Rev. Lett.}\ }\textbf {\bibinfo {volume} {112}},\
		\bibinfo {pages} {140503} (\bibinfo {year} {2014})}\BibitemShut {NoStop}%
	\bibitem [{\citenamefont {L{\"u}ders}(1951)}]{Luders:1951APL}%
	\BibitemOpen
	\bibfield  {author} {\bibinfo {author} {\bibfnamefont {G.}~\bibnamefont
			{L{\"u}ders}},\ }\href {https://doi.org/10.1002/andp.19504430510} {\bibfield
		{journal} {\bibinfo  {journal} {Ann. Phys. (Leipzig)}\ }\textbf {\bibinfo
			{volume} {443}},\ \bibinfo {pages} {323} (\bibinfo {year} {1951})}\ \bibinfo
	{note} {[German; English version: Ann. Phys. (Leipzig) 15, 663 (2006)]} \BibitemShut {NoStop}%
	\bibitem [{\citenamefont {Pokorny}\ \emph {et~al.}(2020)\citenamefont
		{Pokorny}, \citenamefont {Zhang}, \citenamefont {Higgins}, \citenamefont
		{Cabello}, \citenamefont {Kleinmann},\ and\ \citenamefont
		{Hennrich}}]{Pokorny:PRL20}%
	\BibitemOpen
	\bibfield  {author} {\bibinfo {author} {\bibfnamefont {F.}~\bibnamefont
			{Pokorny}}, \bibinfo {author} {\bibfnamefont {C.}~\bibnamefont {Zhang}},
		\bibinfo {author} {\bibfnamefont {G.}~\bibnamefont {Higgins}}, \bibinfo
		{author} {\bibfnamefont {A.}~\bibnamefont {Cabello}}, \bibinfo {author}
		{\bibfnamefont {M.}~\bibnamefont {Kleinmann}},\ and\ \bibinfo {author}
		{\bibfnamefont {M.}~\bibnamefont {Hennrich}},\ }\href
	{https://doi.org/10.1103/PhysRevLett.124.080401} {\bibfield  {journal}
		{\bibinfo  {journal} {Phys. Rev. Lett.}\ }\textbf {\bibinfo {volume} {124}},\
		\bibinfo {pages} {080401} (\bibinfo {year} {2020})}\BibitemShut {NoStop}%
	\bibitem [{\citenamefont {Chiribella}\ and\ \citenamefont
		{Yuan}()}]{chiribella2014measurement}%
	\BibitemOpen
	\bibfield  {author} {\bibinfo {author} {\bibfnamefont {G.}~\bibnamefont
			{Chiribella}}\ and\ \bibinfo {author} {\bibfnamefont {X.}~\bibnamefont
			{Yuan}},\ }\href@noop {} {}\Eprint {https://arxiv.org/abs/1404.3348}
	{arXiv:1404.3348} \BibitemShut {NoStop}%
	\bibitem [{\citenamefont {Davies}\ and\ \citenamefont
		{Lewis}(1970)}]{Davies:JMA70}%
	\BibitemOpen
	\bibfield  {author} {\bibinfo {author} {\bibfnamefont {E.~B.}\ \bibnamefont
			{Davies}}\ and\ \bibinfo {author} {\bibfnamefont {J.~T.}\ \bibnamefont
			{Lewis}},\ }\href {https://doi.org/10.1007/BF01647093} {\bibfield  {journal}
		{\bibinfo  {journal} {Commun. Math. Phys.}\ }\textbf {\bibinfo {volume}
			{17}},\ \bibinfo {pages} {239} (\bibinfo {year} {1970})}\BibitemShut
	{NoStop}%
	\bibitem [{\citenamefont {Ozawa}(1984)}]{Ozawa:JMP1984}%
	\BibitemOpen
	\bibfield  {author} {\bibinfo {author} {\bibfnamefont {M.}~\bibnamefont
			{Ozawa}},\ }\href {https://doi.org/10.1063/1.526000} {\bibfield  {journal}
		{\bibinfo  {journal} {J. Math. Phys. (N.Y.)}\ }\textbf {\bibinfo {volume}
			{25}},\ \bibinfo {pages} {79} (\bibinfo {year} {1984})}\BibitemShut {NoStop}%
	\bibitem [{\citenamefont {Dressel}\ and\ \citenamefont
		{Jordan}(2013)}]{Dressel:PRA2013}%
	\BibitemOpen
	\bibfield  {author} {\bibinfo {author} {\bibfnamefont {J.}~\bibnamefont
			{Dressel}}\ and\ \bibinfo {author} {\bibfnamefont {A.~N.}\ \bibnamefont
			{Jordan}},\ }\href {https://doi.org/10.1103/PhysRevA.88.022107} {\bibfield
		{journal} {\bibinfo  {journal} {Phys. Rev. A}\ }\textbf {\bibinfo {volume}
			{88}},\ \bibinfo {pages} {022107} (\bibinfo {year} {2013})}\BibitemShut
	{NoStop}%
	\bibitem [{\citenamefont {Wagner}\ \emph {et~al.}(2020)\citenamefont {Wagner},
		\citenamefont {Bancal}, \citenamefont {Sangouard},\ and\ \citenamefont
		{Sekatski}}]{Wagner2020deviceindependent}%
	\BibitemOpen
	\bibfield  {author} {\bibinfo {author} {\bibfnamefont {S.}~\bibnamefont
			{Wagner}}, \bibinfo {author} {\bibfnamefont {J.-D.}\ \bibnamefont {Bancal}},
		\bibinfo {author} {\bibfnamefont {N.}~\bibnamefont {Sangouard}},\ and\
		\bibinfo {author} {\bibfnamefont {P.}~\bibnamefont {Sekatski}},\ }\href
	{https://doi.org/10.22331/q-2020-03-19-243} {\bibfield  {journal} {\bibinfo
			{journal} {{Quantum}}\ }\textbf {\bibinfo {volume} {4}},\ \bibinfo {pages}
		{243} (\bibinfo {year} {2020})}\BibitemShut {NoStop}%
	\bibitem [{\citenamefont {Miklin}\ \emph {et~al.}(2020)\citenamefont {Miklin},
		\citenamefont {Borka\l{}a},\ and\ \citenamefont
		{Paw\l{}owski}}]{PhysRevResearch.2.033014}%
	\BibitemOpen
	\bibfield  {author} {\bibinfo {author} {\bibfnamefont {N.}~\bibnamefont
			{Miklin}}, \bibinfo {author} {\bibfnamefont {J.~J.}\ \bibnamefont
			{Borka\l{}a}},\ and\ \bibinfo {author} {\bibfnamefont {M.}~\bibnamefont
			{Paw\l{}owski}},\ }\href {https://doi.org/10.1103/PhysRevResearch.2.033014}
	{\bibfield  {journal} {\bibinfo  {journal} {Phys. Rev. Res.}\ }\textbf
		{\bibinfo {volume} {2}},\ \bibinfo {pages} {033014} (\bibinfo {year}
		{2020})}\BibitemShut {NoStop}%
	\bibitem [{\citenamefont {Chiribella}\ and\ \citenamefont
		{Yuan}(2016)}]{Chiribella:2016IC}%
	\BibitemOpen
	\bibfield  {author} {\bibinfo {author} {\bibfnamefont {G.}~\bibnamefont
			{Chiribella}}\ and\ \bibinfo {author} {\bibfnamefont {X.}~\bibnamefont
			{Yuan}},\ }\href {https://doi.org/10.1016/j.ic.2016.02.006} {\bibfield
		{journal} {\bibinfo  {journal} {Inf. Comput.}\ }\textbf {\bibinfo {volume}
			{250}},\ \bibinfo {pages} {15} (\bibinfo {year} {2016})}\BibitemShut
	{NoStop}%
	\bibitem [{\citenamefont {Chiribella}\ \emph {et~al.}(2020)\citenamefont
		{Chiribella}, \citenamefont {Cabello}, \citenamefont {Kleinmann},\ and\
		\citenamefont {M\"uller}}]{PhysRevResearch.2.042001}%
	\BibitemOpen
	\bibfield  {author} {\bibinfo {author} {\bibfnamefont {G.}~\bibnamefont
			{Chiribella}}, \bibinfo {author} {\bibfnamefont {A.}~\bibnamefont {Cabello}},
		\bibinfo {author} {\bibfnamefont {M.}~\bibnamefont {Kleinmann}},\ and\
		\bibinfo {author} {\bibfnamefont {M.~P.}\ \bibnamefont {M\"uller}},\ }\href
	{https://doi.org/10.1103/PhysRevResearch.2.042001} {\bibfield  {journal}
		{\bibinfo  {journal} {Phys. Rev. Res.}\ }\textbf {\bibinfo {volume} {2}},\
		\bibinfo {pages} {042001(R)} (\bibinfo {year} {2020})}\BibitemShut {NoStop}%
	\bibitem [{\citenamefont {Cabello}(2016{\natexlab{b}})}]{CabelloPRA2016}%
	\BibitemOpen
	\bibfield  {author} {\bibinfo {author} {\bibfnamefont {A.}~\bibnamefont
			{Cabello}},\ }\href {https://doi.org/10.1103/PhysRevA.93.032102} {\bibfield
		{journal} {\bibinfo  {journal} {Phys. Rev. A}\ }\textbf {\bibinfo {volume}
			{93}},\ \bibinfo {pages} {032102} (\bibinfo {year}
		{2016}{\natexlab{b}})}\BibitemShut {NoStop}%
	\bibitem [{\citenamefont {Zimba}\ and\ \citenamefont
		{Penrose}(1993)}]{Zimba:1993SHPS}%
	\BibitemOpen
	\bibfield  {author} {\bibinfo {author} {\bibfnamefont {J.~R.}\ \bibnamefont
			{Zimba}}\ and\ \bibinfo {author} {\bibfnamefont {R.}~\bibnamefont
			{Penrose}},\ }\href {https://doi.org/10.1016/0039-3681(93)90061-N} {\bibfield
		{journal} {\bibinfo  {journal} {Stud. Hist. Philos. Sci. A}\ }\textbf
		{\bibinfo {volume} {24}},\ \bibinfo {pages} {697} (\bibinfo {year}
		{1993})}\BibitemShut {NoStop}%
	\bibitem [{\citenamefont {Gould}\ and\ \citenamefont
		{Aravind}(2010)}]{gould2010isomorphism}%
	\BibitemOpen
	\bibfield  {author} {\bibinfo {author} {\bibfnamefont {E.}~\bibnamefont
			{Gould}}\ and\ \bibinfo {author} {\bibfnamefont {P.~K.}\ \bibnamefont
			{Aravind}},\ }\href {https://doi.org/10.1007/s10701-010-9434-2} {\bibfield
		{journal} {\bibinfo  {journal} {Found. Phys.}\ }\textbf {\bibinfo {volume}
			{40}},\ \bibinfo {pages} {1096} (\bibinfo {year} {2010})}\BibitemShut
	{NoStop}%
	\bibitem [{\citenamefont {Bengtsson}(2012)}]{bengtsson2012gleason}%
	\BibitemOpen
	\bibfield  {author} {\bibinfo {author} {\bibfnamefont {I.}~\bibnamefont
			{Bengtsson}},\ }in\ \href {https://doi.org/10.1063/1.4773124} {\emph
		{\bibinfo {booktitle} {Quantum Theory: Reconsideration on Foundations 6}}},\
	\bibinfo {series} {AIP Conf. Proc.}
	(\bibinfo {organization} {American Institute of Physics},\ \bibinfo {address}
	{Melville, NY},\ \bibinfo {year} {2012})\ Vol.\ \bibinfo {volume} {1508}, p.\ \bibinfo {pages}
	{125}\BibitemShut {NoStop}%
	\bibitem [{\citenamefont {von Neumann}(1950)}]{von1950functional}%
	\BibitemOpen
	\bibfield  {author} {\bibinfo {author} {\bibfnamefont {J.}~\bibnamefont {von
				Neumann}},\ }\href {https://www.jstor.org/stable/j.ctt1b9x13x.2} {\emph
		{\bibinfo {title} {Functional Operators (AM-21), Volume~1: Measures and
				Integrals}}}\ (\bibinfo  {publisher} {Princeton University Press},\ \bibinfo
	{address} {Princeton, NJ},\ \bibinfo {year} {1950})\BibitemShut {NoStop}%
\end{thebibliography}

%%%%%%%%%%%%%%%%%%%%%%%%%%%%%%%%%%%%%%%%%%%%%%%%%%%%%%%%%%%%%%%%%%%

\end{document}